\documentclass{fundam}

\publyear{22}
\papernumber{2102}
\volume{185}
\issue{1}

\usepackage{amsmath,amssymb}

\usepackage{url} 
\usepackage{graphicx}
\usepackage{tikz}
\usepackage{tikzsymbols}
\usetikzlibrary{automata, positioning, arrows}
\usetikzlibrary{patterns}
\usetikzlibrary{shapes.geometric} 

\usepackage[hidelinks]{hyperref}

\usepackage[ruled,linesnumbered]{algorithm2e}
\DontPrintSemicolon

\newenvironment{restate}[2]{\vspace{5pt}\noindent\textbf{#1~#2.}~}{\vspace{5pt}}

\newcommand{\tup}[1]{\mathbf{#1}}

\DeclareMathOperator{\arena}{\mathcal{A}}
\DeclareMathOperator{\dArena}{\arena_\di}
\DeclareMathOperator{\weight}{w}
\DeclareMathOperator{\edgeCost}{\delta_c}
\DeclareMathOperator{\edgePenal}{\delta_p}
\DeclareMathOperator{\dWeight}{\tup{w}}
\DeclareMathOperator{\dCost}{\tup{Cost}}
\DeclareMathOperator{\cost}{Cost}
\DeclareMathOperator{\wCost}{\overline{Cost}}
\newcommand{\di}{d}
\DeclareMathOperator{\targetSet}{F}

\DeclareMathOperator{\leqC}{\leq_{C}}
\DeclareMathOperator{\leqL}{\leq_{L}}
\DeclareMathOperator{\strictLessL}{<_{L}}

\DeclareMathOperator{\playerOne}{\mathcal{P}_1}
\DeclareMathOperator{\playerTwo}{\mathcal{P}_2}
\DeclareMathOperator{\playerI}{\mathcal{P}_i}

\DeclareMathOperator{\hist}{Hist}
\DeclareMathOperator{\play}{Plays}
\DeclareMathOperator{\last}{Last}
\DeclareMathOperator{\successor}{Succ}
\newcommand{\stratSet}[2]{\Sigma_{#1}^{#2}}

\DeclareMathOperator{\game}{\mathcal{G}}
\DeclareMathOperator{\dGame}{\game_\di}
\DeclareMathOperator{\initDGame}{(\dGame,v_0)}

\newcommand{\outcome}[3]{\langle #1, #2 \rangle_{#3}}
\newcommand{\multiOutcome}[2]{\langle #1 \rangle_{#2}}
\newcommand{\multiOutcomeHistory}[2]{\langle #1 \rangle^{\mathrm{H}}_{#2}}

\DeclareMathOperator{\ensureInt}{Ensure}
\newcommand{\ensure}[1]{\ensureInt_{\lesssim}(#1)}
\newcommand{\ensureC}[1]{\ensureInt_{\leqC}(#1)}
\newcommand{\ensureL}[1]{\ensureInt_{\leqL}(#1)}
\DeclareMathOperator{\minimal}{minimal}

\DeclareMathOperator{\paretoInt}{Pareto}
\newcommand{\pareto}[1]{\paretoInt(#1)}
\DeclareMathOperator{\valueInt}{Val}
\newcommand{\upValue}[1]{\overline{\valueInt}(#1)}
\newcommand{\iterEnsure}[2]{\ensureInt^{#1}(#2)}
\newcommand{\iterUpValue}[2]{\overline{\valueInt}^{#1}(#2)}

\newcommand{\lengthF}[1]{{|#1|}_{\targetSet}}

\DeclareMathOperator{\iterInt}{I}
\newcommand{\iter}[2]{\iterInt^{#1}(#2)}
\DeclareMathOperator{\firstOccInt}{Ind}
\newcommand{\firstOcc}[2]{\firstOccInt^{#1}_{#2}}

\newcommand{\tree}{\mathcal{T}}
\newcommand{\stratTree}[1]{\tree_{#1}}
\DeclareMathOperator{\depth}{depth}
\DeclareMathOperator{\height}{height}

\DeclareMathOperator{\cover}{\mathcal{C}}

\DeclareMathOperator{\N}{\mathbb{N}}
\DeclareMathOperator{\NInf}{\overline{\mathbb{N}}}

\newcommand{\IE}{\emph{i.e., }}
\DeclareMathOperator{\NP}{\textsc{NP}}
\DeclareMathOperator{\PSPACE}{\textsc{PSpace}}
\DeclareMathOperator{\Poly}{\textsc{PTime}}
\DeclareMathOperator{\APTime}{\textsc{APTime}}

\DeclareMathOperator{\varLogSet}{D}
\newcommand{\varLog}{t}
\DeclareMathOperator{\bigO}{\mathcal{O}}
\DeclareMathOperator{\maxAC}{|I_{max}|}
\DeclareMathOperator{\maxW}{W}
\newcommand{\evalFormula}[1]{[\![#1]\!]}

\DeclareMathOperator{\generator}{G}
\DeclareMathOperator{\generatorFormula}{gen}

\newcommand{\valExists}[1]{f^{\exists}_{#1}}
\newcommand{\valForall}[1]{f^{\forall}_{#1}}

\newcommand{\CE}{CE }

\DeclareMathOperator{\multiStrat}{\Theta}
\DeclareMathOperator{\multiActions}{A}
\DeclareMathOperator{\penal}{Penalty}

\DeclareMathOperator{\msEnsure}{MEnsure}
\DeclareMathOperator{\msPareto}{MPareto}
\DeclareMathOperator{\cVal}{cVal}
\DeclareMathOperator{\pVal}{pVal}
\newcommand{\msSet}[2]{\mathrm{M}_{#1}^{#2}}
\DeclareMathOperator{\extendedGame}{\mathcal{X}}
\DeclareMathOperator{\CP}{CP}
\DeclareMathOperator{\PC}{PC}

\begin{document}

\title{Multi-weighted Reachability Games and Their Application to Permissiveness\thanks{This paper is an extended version of~\cite{BrihayeG23}.}}


\author{Thomas Brihaye\thanks{Thomas Brihaye -- This work has been supported by the Fonds de
la Recherche Scientifique -- FNRS under Grant n° T.0027.21 (PDR RatBeCoSi)}\\
UMONS - Université de Mons \\
Mons, Belgium\\
thomas.brihaye{@}umons.ac.be
\and Aline Goeminne \thanks{Aline Goeminne -- Postdoctoral Researcher of the Fonds de la Recherche Scientifique -- FNRS.}\\
F.R.S.-FNRS \& UMONS - Université de Mons,\\ Mons, Belgium\\ aline.goeminne{@}umons.ac.be} 

\maketitle

\runninghead{T. Brihaye, A. Goeminne}{Multi-weighted Reachability Games and Their Application to Permissiveness}

\begin{abstract}
  We study two-player multi-weighted reachability games played on a finite directed graph, where an agent, called $\playerOne$, has several quantitative reachability objectives that he wants to optimize against an antagonistic environment, called $\playerTwo$.  In this setting, we ask what cost profiles $\playerOne$ can ensure regardless of the opponent's behavior. Cost profiles are compared thanks to: \emph{(i)} a lexicographic order that ensures the unicity of an upper value and \emph{(ii)} a componentwise order for which we consider the Pareto frontier. We synthesize \emph{(i)} lexico-optimal strategies and  \emph{(ii)} Pareto-optimal strategies. The strategies are obtained thanks to a fixpoint algorithm which also computes the upper value in polynomial time and the Pareto frontier in exponential time. The constrained existence problem is proved in $\Poly$ for the lexicographic order and $\PSPACE$-complete for the componentwise order. Finally, we show  how complexity results about permissiveness of multi-strategies in two-player quantitative reachability games can be derived from the results we obtained in the two-player multi-weighted reachability games setting.
\end{abstract}

\begin{keywords}
two-player games on graphs, multi-weighted reachability games, Pareto-optimal strategies, lexico-optimal strategies, permissiveness of multi-strategies
\end{keywords}


\section{Introduction}

\paragraph*{Two-player Games Played on Graphs.}\emph{Two-player zero-sum games played on graphs} are commonly used in the endeavor to \emph{synthesize} systems that are \emph{correct by construction}. In the two-player zero-sum setting the \emph{system} wants to achieve a given objective whatever the behavior of the \emph{environment}. This situation is modeled by a two-player game in which $\playerOne$ (resp. $\playerTwo$) represents the system (resp. the environment). Each vertex of the graph is owned by one player and they take turn by moving a token from vertex to vertex by following the graph edges. This behavior leads to an infinite sequence of vertices called a \emph{play}. The choice of a player's next move is dictated by his \emph{strategy}. 
In a quantitative setting, edges are equipped with a \emph{weight function} and a \emph{cost function} assigns a  cost to each play. This cost depends on the weights of the edges along the play. With this quantitative perspective, $\playerOne$ wants to \emph{minimize} the cost function. We say that $\playerOne$ can ensure a cost of $x$ if there exists a strategy of $\playerOne$ such that, whatever the strategy followed by $\playerTwo$, the corresponding cost is less than or equal to $x$. An interesting question is thus to determine what are the costs that can be ensured by $\playerOne$. In this document, 
 these costs are called the \emph{ensured values}. Other frequently studied questions are:  Given a threshold $x$, does there exist a strategy of $\playerOne$ that ensures a cost less than or equal to $x$? Is it possible to synthesize such a strategy, or even better, if it exists, a strategy that ensures the best ensured value, \IE an \emph{optimal strategy}?

A well-known studied quantitative objective is the one of \emph{quantitative reachability objective}. A player who wants to achieve such an objective has a subset of vertices, called \emph{target set}, that he wants to reach as quickly as possible. In terms of edge weights, that means that he wants to minimize the cumulative weights until a vertex of the target set is reached. In this setting it is proved that the best ensured value is computed in polynomial time and that optimal strategies exist and do not require memory~\cite{LaroussinieMO06}.

\paragraph{Multi-weighted Reachability Games.} Considering systems with only one cost to minimize may seem too restrictive. Indeed, $\playerOne$ may want to optimize different quantities while reaching his objective. Moreover, optimizing these different quantities may lead to antagonistic behaviors, for instance when a vehicle wants to reach his destination while minimizing both the delay and the energy consumption. This is the reason why in this paper, we study two-player multi-weighted reachability games,  where $\playerOne$ aims at reaching a target while minimizing several costs. In this setting each edge of the graph is  labeled by a $\di$-tuple of $\di$ natural numbers, one per quantity to minimize. Given a sequence of vertices in the game graph, the \emph{cost profile} of $\playerOne$ corresponds to the sum of the weights of the edges, component by component, until a given target set is reached. We consider the multi-dimensional counterpart of the previous studied problems: we wonder what cost profiles are ensured by $\playerOne$. 
 Thus $\playerOne$ needs to arbitrate the trade-off induced by the multi-dimensional setting. In order to do so, we consider two alternatives: the cost profiles can be compared either  via \emph{(i)} a lexicographic order that ranks the objectives \emph{a priori} and leads to a unique minimal ensured value; or via \emph{(ii)} a componentwise  order \footnote{Let $\tup{x} = (x_1,\ldots, x_d)$ and $\tup{y} = (y_1,\ldots, y_d)$, we say that $\tup{x}$ is componentwise smaller than $\tup{y}$ if and only if  for all $i \in \{ 1, \ldots, \di \}$, $x_i \leq y_i$}. In this second situation, $\playerOne$  takes his decision \emph{a posteriori} by choosing an element of the Pareto frontier (the set of minimal ensured values, which is not necessarily a singleton).

 \paragraph{Permissiveness of Multi-strategies.} Although multi-weighted reachability games raise questions that are interesting on their own right, they can be used to study \emph{robustness} of the behavior of $\playerOne$. 
We consider the simpler model of quantitative reachability games.
These games can be seen as multi-weighted reachability games whose edges are labeled with only one natural number.
In this setting, let us assume that  an optimal strategy for $\playerOne$ which allows him to reach his target set is synthesized. 
Unfortunately, if the next move dictated by the optimal strategy  is not available (due to a bug, for example), when $\playerOne$ has to play, the system is blocked from running. 
In order to overcome this lack of \emph{robustness} of the proposed solution concept, \emph{multi-strategies} which propose a set of next moves instead of a single move are considered.
In this way, when $\playerOne$ has to play, he has different possible choices.
With this point of view, we aim at synthesizing the \emph{most permissive multi-strategy}, that is, roughly speaking, the strategy that allows as many behaviors as possible for $\playerOne$ while ensuring certain constraints on the possible costs obtained.
In this paper, we extend the notion of permissiveness of multi-strategies, based on \emph{penalties}, already introduced in~\cite{BDMR09} for qualitative reachability games\footnote{In a qualitative reachability game, $\playerOne$ aims at reaching his target set, regardless of the cost involved.} to quantitative reachability games and we explain how to solve related problems thanks to multi-weighted reachability games.

\paragraph*{Paper Organization.} For sake of concision and clarity of the introduction, the contributions and related works related to permissiveness are provided in Section~\ref{section:permissiveMultiStrat}.
This is the reason why we pursue this section with contributions and related works only related to multi-weighted reachability games.
In Section~\ref{section:prelim}, we introduce all definitions and studied problems related to multi-weighted reachability games while in Section~\ref{section:ensuredValues} and Section~\ref{section:constrainedExistence} we show how to solve them.
Finally, in Section~\ref{section:permissiveMultiStrat}, we focus on permissive multi-strategies in quantitative reachability games by defining all the related concepts, the studied problems and by  explaining how to use multi-weighted reachability games to solve them. To make the paper easier to read, we have also chosen to place the most technical proofs in appendices and to provide only proof intuitions when formal proofs have been eluded.

Let us also mention that this paper is an extended version of~\cite{BrihayeG23}. The original paper dealt only with multi-weighted reachability games. Some of the intuitions of the proofs of results provided in Section~\ref{section:ensuredValues} and Section~\ref{section:constrainedExistence} have been added within these sections as well as formal proofs in Appendix~\ref{appendix:ensuredValues} and Appendix~\ref{app:constrainedExistence}. The entire content of Section~\ref{section:permissiveMultiStrat} concerning permissive multi-strategies in quantitative reachability games is new compared to the short version of this paper.

\paragraph*{Contributions w.r.t. Multi-weighted Reachability Games.}
Our contributions are threefold. First, in Section~\ref{section:algo}, given a two-player multi-weighted reachability game, independently of the order considered, we provide a fixpoint algorithm,  which computes the minimal cost profiles that can be ensured by $\playerOne$. In Section~\ref{section:complexity}, we study the time complexity of this algorithm, depending on the order considered. When considering the lexicographic order (resp. componentwise order), the algorithm runs in polynomial time (resp. exponential time). Moreover, if the number of dimensions is fixed, the computation of the Pareto frontier can be done in pseudo-polynomial time (polynomial if the weights of the game graph are encoded in unary). As a second contribution, in Section~\ref{section:correctStrat}, based on the fixpoint algorithm, we synthesize the optimal strategies (one per order considered). In particular, we show that positional strategies suffice when considering the lexicographic order, although memory is needed in the componentwise case. Finally, in Section~\ref{section:constrainedExistence}, we focus on the natural decision problem associated with our model: the constrained existence problem. Given a two-player multi-weighted reachability game and a cost profile $\tup{x}$, the answer to the constrained existence problem is positive when there exists a strategy of $\playerOne$ that ensures $\tup{x}$. In the lexicographic case, we show that the problem belongs to $\Poly$; although it turns to be $\PSPACE$-complete in the componentwise case.

\paragraph*{Related Work w.r.t. Multi-weighted Reachability Games.}
Up to our knowledge, and quite surprisingly, two-player multi-weighted reachability games, as defined in this paper, were not studied before. Nevertheless, a one-player variant known as multi-constrained routing is known to be $\NP$-complete (see~\cite{PuriT02} for example) . Both exact and approximate algorithms are, for example, provided in~\cite{PuriT02}. The time complexity of their exact algorithm matches our results since it runs in exponential time and they indicate that it is pseudo-polynomial if $\di = 2$.
The one-player setting is also studied in timed automata~\cite{LarsenR05}.

If we focus on two-player settings, another closely related model to multi-weighted reachability games is the one studied in~\cite{FijalkowH13}. The authors consider two-player generalized (qualitative) reachability games. In this setting $\playerOne$ wants to reach several target sets in any order but does not take into account the cost of achieving that purpose. They prove that deciding the winner in such a game is $\PSPACE$-complete. Moreover, they discuss the fact that winning strategies need memory. The memory is used in order to remember which target sets have already been reached. 
In our setting, we assume that there is only one target set   but that weights on edges are tuples and thus the costs to reach are aggregated component by component. 
 Memory is needed because we have to take into consideration the partial sum of weights up to now in order to make the proper choices in the future to ensure the required cost profile.  An example in which an exponential memory is needed is provided in~\cite[Section 3.3.1]{BrihayeGMR23}. Notice that if we would like to study the case where each dimension has its own target set, both types of memory would be needed.

If we consider other objectives than reachability, we can mention different works on multi-dimen\-tional \emph{energy} and \emph{mean-payoff} objectives~\cite{JuhlLR13,ChatterjeeDHR10,ChatterjeeRR12}. In~\cite{BrenguierR15}, they prove that the Pareto frontier in a multi-dimensional mean-payoff game is definable as a finite union of convex sets obtained from linear inequations. The authors also provide a $\Sigma_2^P$ algorithm to decide if this set intersects a convex set defined by linear inequations.

Lexicographic preferences are used in stochastic games with lexicographic (qualitative) reachabili\-ty-safety objectives~\cite{ChatterjeeKWW20}.
The authors prove that lexico-optimal strategies exist but require finite-memory in order to know on which dimensions the corresponding objective is satisfied or not. They also provide an algorithm to compute the best ensured value and compute lexico-optimal strategies thanks to different computations of optimal strategies in single-dimensional games. Finally, they show that deciding if the best ensured value is greater than or equal to a tuple $\tup{x}$ is $\PSPACE$-hard and in $\textsc{NExpTime} \cap \textsc{co-NExpTime}$. 


\section{Preliminaries}
\label{section:prelim}

\subsection{Two-Player Multi-weighted Reachability Games}
\label{section:prelimMultiReachGames}
\subsubsection*{Weighted Arena} We consider games that are played on an \emph{ (weighted) arena} by two players: $\playerOne$ and $\playerTwo$. An arena $\dArena$ is a tuple $(V_1,V_2,E, \dWeight)$ where \emph{(i)} $(V=V_1 \cup V_2, E)$ is a graph such that vertices $V_i$ for $i \in \{1,2\}$ are owned by $\playerI$ and $V_1 \cap V_2 = \emptyset$ and \emph{(ii)} $\dWeight: E \longrightarrow \N^\di$ is a weight function which assigns $\di$ natural numbers to each edge of the graph. The variable $\di$ is called the number of \emph{dimensions}. For all $1 \leq i \leq \di$, we denote by $\weight_i$, with $\weight_i : E \longrightarrow \N$, the projection of $\dWeight$ on the ith component, \IE for all $e \in E$, if $\dWeight(e) = (n_1,\ldots,n_\di)$ then, $\weight_i(e) = n_i$. We define $\maxW$ as the largest weight that can appear in the values of the weight function, \IE $\maxW = \max \{ \weight_i(e) \mid 1 \leq i \leq \di \text{ and } e \in E \}$.

Each time we consider a tuple $\tup{x} \in X^\di$ for some set $X$, we write it in bold and we denote the ith component of this tuple by $x_i$. Moreover, we abbreviate the tuples $(0,\ldots,0)$ and $(\infty,\ldots, \infty)$ by $\tup{0}$ and $\tup{\infty}$ respectively.

\subsubsection*{Plays and Histories} A \emph{play} (resp. \emph{history}) in $\dArena$ is an infinite (resp. finite) sequence of vertices consistent with the structure of the associated arena $\dArena$, \IE if $\rho = \rho_0\rho_1\ldots$ is a play then, for all $n \in \N$, $\rho_n \in V$ and $(\rho_n,\rho_{n+1}) \in E$.
A history may be formally defined in the same way. 
The set of plays (resp. histories) are denoted by $\play_{\dArena}$ (resp. $\hist_{\dArena}$). When the underlying arena is clear from the context we only write $\play$ (resp. $\hist$).
We also denote by $\hist_1$ the set of histories which end in a vertex owned by $\playerOne$, \IE $\hist_{1} = \{ h=h_0h_1\ldots h_n \mid h \in \hist \text{ and } h_n \in V_1 \}$.
For a given vertex $v \in V$, the sets $\play(v)$, $\hist(v)$, $\hist_1(v)$ denote the sets of plays or histories starting in $v$. Finally, for a history $h = h_0\ldots h_n$, the vertex $h_n$ is denoted by $\last(h)$ and $|h| = n$ is the length of $h$.

\subsubsection*{Multi-weighted Reachability Games} We consider \emph{multi-weighted reachability games} such that $\playerOne$ has a target set that he wants to reach from a given initial vertex. Moreover, crossing edges on the arena implies the increasing of the $\di$ cumulated costs for $\playerOne$. While in $1$-weighted reachability game $\playerOne$ aims at reaching his target set as soon as possible (minimizing his cost), in the general $\di$-weighted case he wants to find a \emph{trade-off} between the different components.

More formally, $\targetSet \subseteq V$ which is a subset of vertices that $\playerOne$ wants to reach is called the \emph{target set} of $\playerOne$.
The \emph{cost function}  $\dCost: \play \longrightarrow \NInf^\di$  of $\playerOne$ provides, given a play $\rho$, the cost of $\playerOne$ to reach his target set $\targetSet$ along $\rho$.\footnote{Where the following notation is used: $\NInf = \N \cup \{ \infty \}$}
This cost corresponds to the sum of the weight of the edges, component by component, until he reaches $\targetSet$ or is equal to $\infty$ for all components if it is never the case.
For all $1 \leq i \leq \di$, we denote by $\cost_i: \play \longrightarrow \NInf$, the projection of $\dCost$ on the ith component. Formally, for all $\rho= \rho_0\rho_1 \ldots \in \play$:

$$\cost_i(\rho) = \begin{cases} \sum_{n=0}^{\ell-1}w_i(\rho_n,\rho_{n+1}) & \text{ if } \ell \text{ is the least index such that } \rho_\ell \in \targetSet \\ \infty & \text{ otherwise } \end{cases}$$

and $\dCost(\rho) = (\cost_1(\rho), \ldots, \cost_\di(\rho))$ is called a \emph{cost profile}. 

 If $h = h_0\ldots h_\ell$ is a history, $\dCost(h)= \sum_{n=0}^{\ell-1} \dWeight(h_n,h_{n+1})$ is the accumulated costs, component by component, along the history. We assume that $\dCost(v) = \tup{0}$, for all $v \in V$.

\begin{definition}[Multi-weighted Reachability Game]
Given a target set $\targetSet \subseteq V$, the tuple $\dGame = (\dArena, \targetSet , \dCost)$ is called a $\di$-weighted reachability game, or more generally a multi-weighted reachability game.
\end{definition}

In a $\di$-weighted reachability game $\dGame = (\dArena, \targetSet, \dCost)$, an initial vertex $v_0 \in V$ is often fixed and the game $\initDGame$ is called an \emph{initialized multi-weighted reachability game}. A play (resp. history) of $\initDGame$ is a play (resp. history) of $\dArena$ starting in $v_0$.

In the rest of this document, for the sake of readability we write (initialized) game instead of (initialized) $\di$-weighted reachability game.

\begin{figure}
    \centering
    \begin{tikzpicture}
        \node[draw, inner sep=3pt] (v0) at (0,0){$v_0$};
        \node[draw, circle, inner sep=2pt] (v1) at (2,1){$v_1$};
        \node[draw, circle, inner sep=2pt] (v2) at (2,0){$v_2$};
        \node[draw, circle, inner sep=2pt] (v3) at (2,-1){$v_3$};
        \node[draw, circle, inner sep=2pt] (v4) at (4,0){$v_4$};
        \node[draw, circle, inner sep=2pt] (v5) at (4,-1){$v_5$};

        \draw[->,thick] (v0) to[bend left] node[fill=white, inner sep=0pt]{$(4,2)$} (v1);
        \draw[->,thick] (v0) to node[fill=white, inner sep=0pt]{$(2,4)$} (v2);
        \draw[->,thick] (v0) to[bend right]  (v3);
        \draw[->,thick] (v1) to  (v4);
        \draw[->,thick] (v2) to (v4);
        \draw[->,thick] (v3) to (v4);
        \draw[->,thick] (v3) to (v5);

        \node[draw, circle, inner sep=2pt] (v6) at (6,1){$v_6$};
        \node[draw, circle, inner sep=2pt] (v7) at (6,0){$v_7$};
        \node[draw, circle, inner sep=2pt] (v8) at (6,-1){$v_8$}; 
        \node[draw, circle, inner sep=2pt, accepting] (v9) at (8,0){$v_9$};
        \node[draw,  inner sep=3pt] (v10) at (8,-1){$v_{10}$};

        \draw[->,thick] (v4) to[bend left] node[fill=white, inner sep=0pt]{$(4,2)$} (v6);
        \draw[->,thick] (v4) to node[fill=white, inner sep=0pt]{$(2,4)$} (v7);
        \draw[->,thick] (v6) to[bend left] (v9);
        
        \draw[->,thick] (v7) to (v9);
        \draw[->,thick] (v5) to (v8);
        \draw[->,thick] (v8) to (v10);
        \draw[->,thick] (v10) to (v9);

        \draw[->,thick] (v8) to[bend right] (v5);
        \draw[->,thick] (v9) to[loop right] (v9);

        \draw[dotted] (-0.5,-0.5) -- (8.9,-0.5) -- (8.9,1.5)-- (-0.5,1.5)--(-0.5,-0.5);
        
    \end{tikzpicture}
    \caption{Example of the arena $\arena_2$ of a game $\game_2$. The target set is $\targetSet=\{ v_9 \}$ and the weight function is given by the label of the edges. Edges without a label have a weight of $(1,1)$. The dotted rectangle is a restriction of the arena specifically used in Example~\ref{ex:memoIsRequired}.}
    \label{fig:runningEx}
\end{figure}

\begin{example}
\label{ex:runningEx1}
We consider as a running example the game $\game_2$ such that its arena $\arena_2 = (V_1,V_2,E,\dWeight)$ is depicted in Figure~\ref{fig:runningEx}. In this example the set of vertices of $\playerOne$ (resp. $\playerTwo$) are depicted by rounded (resp. rectangular) vertices and the vertices that are part of the target set are doubly circled/framed. The weight function $\dWeight$ labels the corresponding edges. We follow those conventions all along this document. Here, $V_1 = \{v_1, v_2,v_3,v_4,v_5,v_6,v_7,v_8,v_9\}$, $V_2 = \{ v_0, v_{10} \}$, $\targetSet = \{ v_9 \}$ and, for example, $\dWeight(v_0,v_2) = (2,4)$. For all edges without label, we assume that the weight is $(1,1)$, \emph{e.g.,} $\dWeight(v_3,v_4) = (1,1)$. Do not pay attention to the dotted rectangle for the moment.

Let us now study the cost profiles of two different plays. First, the play $\rho = v_0v_1v_4v_6v_9^\omega$ has a cost profile of $\dCost(\rho) = (4,2)+(1,1)+(4,2)+(1,1) = (10,6)$ since $\rho$ visits $\targetSet$ in $v_9$. Moreover, $\cost_1(\rho) = 10$ and $\cost_2(\rho) = 6$. Second, the play $\rho' = v_0v_3(v_5v_8)^\omega$ has a cost profile of $(\infty, \infty)$ since it does not reach $\targetSet$. 

\end{example}

\subsubsection*{Strategies}  A \emph{strategy} of player $i$, $i \in \{1,2\}$, provides the next action of $\playerI$. Formally, a strategy of $\playerI$ from a vertex $v$ is a function $\sigma_i: \hist_i(v) \longrightarrow V$ such that for all $h \in \hist_i(v)$, $(\last(h),\sigma_i(h)) \in E$. We denote by $\stratSet{i}{v}$ the set of strategies of $\playerI$ from $v \in V$. Notice that in an initialized game $\initDGame$, unless we specify something else, we assume that the strategies are defined from $v_0$. 

Moreover, given two strategies $\sigma_1$ of $\playerOne$ and $\sigma_2$ of $\playerTwo$, there is only one play which is consistent with $(\sigma_1,\sigma_2)$ from $v_0$. This play is called the \emph{outcome} of $(\sigma_1,\sigma_2)$ from $v_0$ and is denoted by $\outcome{\sigma_1}{\sigma_2}{v_0}$.

We differentiate two classes of strategies: \emph{positional} strategies and \emph{finite-memory} strategies. A positional strategy $\sigma_i$ only depends on the last vertex of the history, \IE for all $h, h' \in \hist_i$, if $\last(h) = \last(h')$ then, $\sigma_i(h) = \sigma_i(h')$. It is finite-memory if it can be encoded by a finite-state machine.

\subsubsection*{Partial Orders}  Given two cost profiles $\tup{x}$ and $\tup{y}$ in $\NInf^\di$, $\playerOne$ should be able to decide which one is the most beneficial to him. In order to do so, we consider two \emph{partial orders} in the rest of this document: the \emph{componentwise} order and the \emph{lexicographic} order.

We recall some related definitions. A \emph{partial order} on $X$ is a binary relation ${\lesssim} \subseteq X \times X$ which is reflexive, antisymmetric and transitive. The \emph{strict partial} order $<$ associated with it is given by $x < y$ if and only if $x \lesssim y$ and $x \neq y$, for all $x,y \in X$. A partial order is called a \emph{total order} if and only if for all $x,y \in X$, $x \lesssim y$ or $y \lesssim x$. Given a set $X' \subseteq X$, the set of minimal elements of $X'$ with respect to $\lesssim$ is given by $\minimal(X') = \{ x \in X' \mid \text{ if } y \in X' \text{ and } y \lesssim x, \text{ then } x=y \}$. Moreover, the \emph{upward closure} of $X'$ with respect to $\lesssim$ is the set $\uparrow X' = \{ \tup{x} \in X \mid \exists \tup{y} \in X' \text{ st. } \tup{y} \lesssim \tup{x} \}$. A set $X'$ is said \text{upward closed} if $\uparrow X' = X'$.

In what follows we consider two partial orders on $\NInf^\di$. The \emph{lexicographic} order, denoted by $\leqL$, is defined as follows: for all $\tup{x},\tup{y} \in \NInf^\di$, $\tup{x} \leqL \tup{y}$ if and only if either \emph{(i)} $x_i = y_i$ for all $i \in \{1, \ldots, \di\}$ or \emph{(ii)} there exists $i \in \{ 1, \ldots, \di \}$ such that $x_i < y_i$ and for all $k < i$, $x_k = y_k$. The \emph{componentwise} order, denoted by $\leqC$, is defined as: for all $\tup{x},\tup{y} \in \NInf^\di$, $\tup{x} \leqC \tup{y}$ if and only if  for all $i \in \{ 1, \ldots, \di \}$, $x_i \leq y_i$. Although the lexicographic order is a total order, the componentwise order is not.

\subsection{Studied Problems}
 We are now able to introduce the different problems that are studied in this paper: the \emph{ensured values problem} and the \emph{constrained existence problem}.

\subsubsection{Ensured Values}

Given a game $\dGame$ and a vertex $v$, we define the \emph{ensured values} from $v$ as the cost profiles that $\playerOne$ can ensure from $v$ whatever the behavior of $\playerTwo$.
We denote the set of ensured values from $v$ by $\ensure{v}$, \IE $\ensure{v} = \{ \tup{x} \in \NInf^\di \mid \exists \sigma_1 \in \stratSet{1}{v}\text{ st. } \forall\sigma_2 \in \stratSet{2}{v} \, , \dCost(\outcome{\sigma_1}{\sigma_2}{v}) \lesssim \tup{x} \} $. Moreover, we say that a strategy $\sigma_1$ of $\playerOne$ from $v$ \emph{ensures} the cost profile $\tup{x} \in \NInf^\di$ if for all strategies $\sigma_2$ of $\playerTwo$ from $v$, we have that $\dCost(\outcome{\sigma_1}{\sigma_2}{v}) \lesssim \tup{x}$.

We denote by $\minimal(\ensure{v})$ the set of minimal elements of $\ensure{v}$ with respect to $\lesssim$. If  $\lesssim$ is the lexicographic order, the set of minimal elements of $\ensureL{v}$ with respect to $\leqL$ is a singleton, as $\leqL$ is a total order, and is called the \emph{upper value} from $v$. We denote it by $\upValue{v}$. On the other hand, if $\lesssim$ is the componentwise order, the set of minimal elements of $\ensureC{v}$ with respect to $\leqC$ is called the \emph{Pareto frontier} from $v$ and is denoted by $\pareto{v}$.

\begin{definition}[Ensured Values Problems]
Let $\initDGame$ be an initialized game. Depending on the partial order, we distinguish  two problems:
    \emph{(i)} \textbf{computation of the upper value}, $\upValue{v_0}$, and
    \emph{(ii)} \textbf{computation of the Pareto frontier}, $\pareto{v_0}$.   
\end{definition}

\begin{theorem}
Given an initialized game $\initDGame$,
\begin{enumerate}
    \item The upper value $\upValue{v_0}$ can be computed in polynomial time.\label{item:resEnsure1}
    \item The Pareto frontier can be computed in  exponential time.\label{item:resEnsure2}
    \item If $\di$ is fixed, the Pareto frontier can be computed in pseudo-polynomial time.\label{item:resEnsure3} 
\end{enumerate}
\end{theorem}

Statement~\ref{item:resEnsure1} is obtained by Theorem~\ref{thm:complexityLexico}, Statements~\ref{item:resEnsure2} and~\ref{item:resEnsure3} are proved by Theorem~\ref{thm:complexityComponent}.

A strategy $\sigma_1$ of $\playerOne$ from $v$ is said \emph{Pareto-optimal}
from $v$ if $\sigma_1$ ensures $\tup{x}$ for some $\tup{x} \in \pareto{v}$. If we want to explicitly specify the element $\tup{x}$ of the Pareto frontier which is ensured by the Pareto-optimal strategy we say that the strategy $\sigma_1$ is \emph{ $\tup{x}$-Pareto-optimal} from $v$.
Finally, a strategy $\sigma_1$ of $\playerOne$ from $v$ is said \emph{lexico-optimal} if it ensures the only  $\tup{x} \in \upValue{v}$.

In Section~\ref{section:correctStrat}, we show how to obtain \emph{(i)} a $\tup{x}$-Pareto-optimal strategy from $v_0$ for each $\tup{x} \in \pareto{v_0}$ and \emph{(ii)} a lexico-optimal strategy from $v_0$ which is positional. Notice that, as in Example~\ref{ex:memoIsRequired}, Pareto-optimal strategies sometimes require finite-memory.

\subsubsection{Constrained Existence}
We are also interested in deciding, given a cost profile $\tup{x}$, whether there exists a strategy $\sigma_1$ of $\playerOne$ from $v_0$ that ensures $\tup{x}$. We call this decision problem the \emph{constrained existence problem} (CE problem).

\begin{definition}[Constrained Existence Problem -- CE Problem]
\label{prob-constProb}
Given an initialized game $\initDGame$ and $\tup{x} \in \N^\di$, does there exist a strategy $\sigma_1 \in \stratSet{1}{v_0}$ such that for all strategies $\sigma_2  \in \stratSet{2}{v_0}$,
$ \dCost(\outcome{\sigma_1}{\sigma_2}{v_0}) \lesssim \tup{x}?$
\end{definition}

The complexity results of this problem are summarized in the following theorem which is restated and discussed in Section~\ref{section:constrainedExistence}.

\begin{theorem}
If $\lesssim$ is the lexicographic order, the \CE problem is solved in $\Poly$.\label{item:resConstrained1}
If $\lesssim$ is the componentwise order, the  \CE problem is $\PSPACE$-complete.\label{item:resConstrained3}
\end{theorem}

We conclude this section by showing that memory may be required by $\playerOne$ in order to ensure a given cost profile.  A more sophisticated example, in which an exponential memory is needed, is provided in~\cite[Section 3.3.1]{BrihayeGMR23}.

\begin{example}
\label{ex:memoIsRequired}
    We consider the game such that its arena is a restriction of the arena given in Figure~\ref{fig:runningEx}. This restricted arena is inside the dotted rectangle. For clarity, we assume that the arena is only composed by vertices $v_0,v_1,v_2,v_4,v_6,v_7$ and $v_9$ and their associated edges. We prove that with the componentwise order $\leqC$, memory for $\playerOne$ is required to ensure the cost profile $(8,8)$. There are only two positional strategies of $\playerOne$: $\sigma_1$ defined such that $\sigma_1(v_4) = v_6$ and $\tau_1$ defined such that $\tau_1(v_4) = v_7$. For all the other vertices, $\playerOne$ has no choice. With $\sigma_1$, if $\playerTwo$ chooses $v_1$ from $v_0$, the resulting cost profile is $(10,6)$. In the same way,  with $\tau_1$, if $\playerTwo$ chooses $v_2$ from $v_0$, the resulting cost profile is $(6,10)$. This proves that $\playerOne$ cannot ensure $(8,8)$ from $v_0$ with a positional strategy. This is nevertheless possible if $\playerOne$ plays a finite-memory strategy. Indeed, by taking into account the past choice of $\playerTwo$, $\playerOne$ is able to ensure $(8,8)$: if $\playerTwo$ chooses $v_1$ (resp. $v_2$) from $v_0$ then, $\playerOne$ should choose $v_7$ (resp. $v_6$) from $v_4$ resulting in a cost profile of $(8,8)$ in both cases.
\end{example}


\section{Ensured Values}
\label{section:ensuredValues}

This section is devoted to the computation of the sets $\minimal(\ensure{v})$ for all $v \in V$. In Section~\ref{section:algo}, we provide a fixpoint algorithm which computes these sets. In Section~\ref{section:complexity}, we study the time complexity of the algorithm both for the lexicographic and the componentwise orders.
Finally, in Section~\ref{section:correctStrat}, we synthesize lexico and Pareto-optimal strategies.  

\subsection{Fixpoint Algorithm}
\label{section:algo}

Our algorithm that computes the sets $\minimal(\ensure{v})$ for all $v \in V$ shares the key idea of some classical shortest path algorithms. First, for each $v \in V$, we compute the set of cost profiles that $\playerOne$ ensures from $v$ in $k$ steps. Then, once all these sets are computed, we compute the sets of cost profiles that can be ensured by $\playerOne$ from each vertex but in $k+1$ steps.  And so on, until the sets of cost profiles are no longer updated, meaning that we have reached a fixpoint.

For each $k \in \N$ and each $v \in V$, we define the set $\iterEnsure{k}{v}$ as the set of cost profiles that can be ensured by $\playerOne$ within $k$ steps.
Formally, $\iterEnsure{k}{v} = \{ \tup{x} \in \NInf^\di \mid \exists \sigma_1 \in \stratSet{1}{v} \text{ st. } \forall\sigma_2 \in \stratSet{2}{v} \, , \dCost(\outcome{\sigma_1}{\sigma_2}{v}) \lesssim \tup{x} \, \wedge \, \lengthF{\outcome{\sigma_1}{\sigma_2}{v}} \leq k \}$\footnote{To lighten the notations, we omit the mention of $\lesssim$ in subscript.},
where for all $\rho = \rho_0\rho_1\ldots \in \play$, $\lengthF{\rho} = k $ if $k$ is the least index such that $\rho_k \in \targetSet$ and $\lengthF{\rho} = - \infty$ otherwise.

  Note that the sets $\iterEnsure{k}{v}$ are upward closed and that they are infinite sets except if $\iterEnsure{k}{v} \allowbreak = \{ \tup{\infty} \}$. This is the reason why, in the algorithm, we only store sets of minimal elements denoted by $\iter{k}{v}$.
Thus, the correctness of the algorithm relies on the property that for all $k \in \N$ and all $v \in V$, $\minimal(\iterEnsure{k}{v}) = \iter{k}{v}$.

The fixpoint algorithm is provided by Algorithm~\ref{algo:fixPoint} in which, if $X$ is a set of cost profiles, and $v,v' \in V$, $X + \dWeight(v,v') = \{ \tup{x} + \dWeight(v,v') \mid \tup{x} \in X \}$. For the moment, do not pay attention to Lines~\ref{lineAlgo:stratStart} to~\ref{lineAlgo:stratEnd}, we come back to them later.

\begin{algorithm}
\caption{Fixpoint algorithm}
    \label{algo:fixPoint}
    
\SetAlgoNoEnd
\lFor{$v \in \targetSet$}
    {$\iter{0}{v} = \{ \mathbf{0} \}$ }
\lFor{$v \not \in \targetSet$}
    {$\iter{0}{v} = \{ \mathbf{\infty} \}$ }
\;
\Repeat{$\iter{k+1}{v} = \iter{k}{v}$ for all $v \in V$}
    {
        \For{$v \in V$\label{lineAlgo:beginningFirstForLoop}} 
        {

            \lIf{ $v \in \targetSet$}{$\iter{k+1}{v} = \{ \mathbf{0}$\}\;}            
            \ElseIf{$v \in V_1$}
                {
                    $\iter{k+1}{v} = \displaystyle\minimal\left(\bigcup_{v' \in \successor(v)} \uparrow \iter{k}{v'} + \dWeight(v,v')\right)$ \; \label{lineAlgo:union}

                    \For{$\tup{x} \in \iter{k+1}{v}$\label{lineAlgo:stratStart}}
                        {
                            \lIf{$\tup{x} \in \iter{k}{v}$}{$f^{k+1}_v(\tup{x})  = f^{k}_v(\tup{x})$}
                            
                            \Else{$f^{k+1}_v(\tup{x}) =  (v', \tup{x'})$ where $v'$ and $\tup{x'}$ are such that $v' \in \successor(v)$, $\tup{x} = \tup{x'} + \dWeight(v,v')$ and $\tup{x'} \in \iter{k}{v'}$ \; \label{lineAlgo:stratEnd}
                        }
                    }
                }
            \;
            \ElseIf{$v \in V_2$}
                {
                    $\iter{k+1}{v} = \displaystyle\minimal\left(\bigcap_{v' \in \successor(v)} \uparrow \iter{k}{v'} + \dWeight(v,v')\right)$ \; \label{lineAlgo:inter}
                    }
                
        }
        
    }
    
\end{algorithm}

\begin{table}[ht!]
    \centering
   \scalebox{0.77}{ \begin{tabular}{|c|c|c|c|c|c|c|c|c|c|}
        \hline
         &  $\lesssim$ & $v_0$ & $v_1$,$v_2$ & $v_3$ & $v_4$ & $v_5$ & $v_6$,$v_7$,$v_{10}$ & $v_8$ & $v_9$ \\ \hline \hline
         $\iter{*}{\cdot}$ & $\leqL$ &
         $\{(8,8)\}$&
         $\{(4,6)\}$ &
         $\{(4,4)\}$&
         $\{(3,5)\}$& $\{(3,3)\}$ & $\{(1,1)\}$ & $\{(2,2)\}$ & $\{(0,0)\}$\\
         \cline{2-10}& $\leqC$& $\{(8,8)\}$ &
         $\{(6,4), (4,6) \}$ &
         $\{(4,4)\}$&
         $\{(5,3),(3,5)\}$& $\{(3,3)\}$ & $\{(1,1)\}$ & $ \{(2,2)\}$& $\{(0,0)\}$  \\ 
         \hline
    \end{tabular}}
    \caption{Fixpoint of the fixpoint algorithm reached at step $k^*=4$.}
    \label{tab:fixpoint}
\end{table}

\begin{example}
\label{ex:algo}
We now explain how the fixpoint algorithm runs on Example~\ref{ex:runningEx1}. Table~\ref{tab:fixpoint} represents the fixpoint of the fixpoint algorithm both for the lexicographic and componentwise orders. Remark that the fixpoint is reached with $k^* = 4$, while the algorithm takes one more step in order to check that $\iter{4}{v} = \iter{5}{v}$ for all $v \in V$. We only focus on some relevant steps of the algorithm with the componentwise order $\leqC$.

Let us first assume that the first step is computed and is such that  $\iter{1}{v_9} = \{(0,0)\}$ since $v_9 \in \targetSet$,  $\iter{1}{v} = \{ (1,1) \}$ if $v \in \{ v_6, v_7, v_{10} \}$ and  $\iter{1}{v} = \{ (\infty, \infty) \}$ for all other vertices. We now focus on the computation of $\iter{2}{v_4}$. By Algorithm~\ref{algo:fixPoint}, $\iter{2}{v_4} = \minimal ( \uparrow \iter{1}{v_6} + (4,2) \, \cup \uparrow \iter{1}{v_7} + (2,4) ) = \minimal( \uparrow \{ (5,3) \} \, \cup \uparrow \{ (3,5) \} ) = \{ (5,3), (3,5) \}$.

We now assume: $\iter{3}{v_0} = \{ (\infty, \infty) \}$, $\iter{3}{v_1}= \iter{3}{v_2}= \iter{3}{v_3} = \{ (4,6), (6,4) \}$, $\iter{3}{v_4} = \{(5,3),(3,5)\}$ and $\iter{3}{v_5} = \{ (3,3) \}$. We compute $\iter{4}{v_0}$ which is equal to $ \minimal( \uparrow \{ (4,6),(6,4) \} + (4,2) \, \cap \uparrow \{ (4,6), (6,4)\} + (2,4) \, \cap \uparrow \{ (4,6), (6,4) \} + (1,1) ) = \minimal( \uparrow \{ (8,8), (10,6) \} \, \cap \uparrow \{ (6,10), (8,8) \} \, \cap \uparrow \{ (5,7), (7,5) \} ) = \minimal( \uparrow \{ (8,8) \} \, \cap \uparrow \{ (5,7),(7,5) \}) = \{ (8,8) \}$. Finally, we compute $\iter{4}{v_3} = \minimal( \uparrow \{(6,4),(4,6)\} \,\cup \uparrow \{ (4,4) \}) = \minimal ( \{ (6,4), (4,6), (4,4) \}) = \{(4,4)\}$.
\end{example}

\subsubsection{Termination}
\label{section:algoTermination}

We focus on the termination of the fixpoint algorithm. 

\begin{proposition}
\label{prop:algoTerminatesSteps}
The fixpoint algorithm terminates in less than $|V| + 1$ steps.
\end{proposition}

The proof of this proposition relies on  Propositions~\ref{prop:noCycle} and~\ref{prop:fixPointReach}. Proposition~\ref{prop:noCycle} is interesting on its own. It states that if there exists a strategy $\sigma_1$ of $\playerOne$ which ensures a cost profile $\tup{x} \in \N^\di$ from $v \in V$ then, there exists another strategy $\sigma'_1$ of $\playerOne$ which also ensures $\tup{x}$  from $v$ but such that the number of edges between $v$ and the first occurrence of a vertex in $\targetSet$ is less than or equal to $|V|$, and this regardless of the behavior of $\playerTwo$.

\newcommand{\restateNoCycle}{Given a game $\dGame$, a vertex $v \in V$ and a cost profile $\tup{x} \in \N^\di$, if there exists a strategy $\sigma_1$ of $\playerOne$ such that for all strategies $\sigma_2$ of $\playerTwo$ we have that $\dCost(\outcome{\sigma_1}{\sigma_2}{v}) \lesssim \tup{x}$ then, there exists $\sigma'_1$ of $\playerOne$ such that  for all $\sigma_2$ of $\playerTwo$ we have: \emph{(i)} $\dCost(\outcome{\sigma'_1}{\sigma_2}{v}) \lesssim \tup{x}$ and \emph{(ii)} $\lengthF{\outcome{\sigma'_1}{\sigma_2}{v}} \leq |V|$.}
\begin{proposition}
\label{prop:noCycle}
\restateNoCycle
\end{proposition}


\textbf{Intuition of the proof of Proposition~\ref{prop:noCycle}.} The intuition of the proof is illustrated in Figure~\ref{fig:winningStratTree}. This tree represents all the consistent plays with the strategy $\sigma_1$ of $\playerOne$ and all the possible strategies of $\playerTwo$. Notice that since for all strategies of $\playerTwo$ the target set is reached, all branches of this tree can be assumed to be finite. The main idea is that we remove successively all cycles in the branches of the tree. That ensures that the height of the tree is less than $|V|$.

If there exists a branch with a cycle beginning with a node owned by $\playerOne$, as the one between the two hatched nodes, $\playerOne$ can directly choose to follow the dotted edge. By considering this new strategy, and so by removing this cycle, the cost profiles of the branches of the new obtained tree are less than or equal to those of the old tree. Once all such kind of cycles are removed, we conclude that there is no more cycle in the tree.

Indeed, the only possibility  is that there remains a cycle beginning with a node owned by $\playerTwo$, as the black nodes. In this case, either \emph{(i)} all nodes between the black nodes are owned by$\playerTwo$ or \emph{(ii)} there exists at least a node owned by $\playerOne$ between the two black nodes, as the node in gray. If situation \emph{(i)} occurs there should exist an infinite branch of the tree in which this cycle is repeated infinitely often, but it is impossible because the target set is assumed to be reached along all branches. If it is the situation \emph{(ii)} that occurs, there should exist a branch in which there is a cycle beginning with a node owned by $\playerOne$, as in the figure. But we assumed that this was no longer the case.

Finally, at the end of the procedure, we obtain a tree from which we recover a strategy $\sigma'_1$ of $\playerOne$ from $v$ such that for all strategies $\sigma_2$ of $\playerTwo$ from $v$, $\dCost(\outcome{\sigma'_1}{\sigma_2}{v}) \lesssim \tup{x}$ and $\lengthF{\outcome{\sigma'_1}{\sigma_2}{v}} \leq |V|$.

\begin{figure}[ht]
\begin{center}
    \scalebox{0.8}{\begin{tikzpicture}
    \node[draw, circle, minimum width=0.5cm] (v0) at (0,0){};
    \node[draw, minimum width=0.5cm, minimum height=0.5cm] (v1) at (1,0){};
    \draw[->,thick] (v0) -- (v1);
    
    \node[draw, circle, minimum width=0.5cm] (v21) at (2,-1){};
    \draw[->,thick] (v1) to (v21);

    \node[draw, circle, minimum width=0.5cm, pattern=north west lines] (v31) at (3,-1){};
    \draw[->,thick] (v21) to (v31);

    \node[draw, minimum height=0.5cm, minimum width=0.5cm] (v41) at (4,-1){};
    \draw[->,thick] (v31) to (v41);

    \node[draw, circle, minimum width=0.5cm] (v51) at (5,-1.75){};
    \node[draw, circle, minimum width=0.5cm] (v52) at (5,-1){};
    \node[draw, circle, minimum width=0.5cm] (v53) at (5,-0.25){};
    \draw[->,thick] (v41) to (v51);
    \draw[->,thick] (v41) to (v52);
    \draw[->,thick] (v41) to (v53);

    \node[draw, circle, minimum width=0.5cm, pattern=north west lines] (v61) at (6,-1.75){};
    \draw[->,thick] (v51) to (v61);
    
    \node[draw,circle, accepting, minimum width=0.5cm] (s1) at (8,-1.75){};
     \node[draw,circle,  minimum width=0.5cm] (v71) at (7,-1.75){};
    \draw[->,thick] (v61) to (v71);
    \draw[->,thick] (v71) to (s1);

    \draw[->,thick,dotted] (v31)  to [bend right=70] (v71);

     \node[draw,circle, accepting, minimum width=0.5cm] (s2) at (6,-1){};
     \draw[->,thick] (v52) to (s2);

     \node[draw,circle, accepting, minimum width=0.5cm] (s3) at (6,-0.25){};
     \draw[->,thick] (v53) to (s3);

    \node[draw, minimum width=0.5cm, minimum height=0.5cm, fill=black, opacity=1] (v51bis) at (2,1){};
    \draw[->,thick] (v41) to (v51);
    \draw[->,thick] (v1) to (v51bis);

    \node[draw, circle, minimum width=0.5cm] (v61) at (3,0.5){};
  \node[draw, circle, minimum width=0.5cm, fill=gray, opacity=0.7] (v62) at (3,1.5){};
  \draw[->,thick] (v51bis) to (v61);
  \draw[->,thick] (v51bis) to (v62);

\node[draw,circle, accepting, minimum width=0.5cm] (s4) at (4,0.5){};
\draw[->,thick] (v61) to (s4);

\node[draw, circle, minimum width=0.5cm] (v71) at (4,1.5){};
\draw[->,thick] (v62) to (v71);

\node[draw, minimum height=0.5cm, minimum width=0.5cm] (v81) at (5,1.5){};
\draw[->,thick] (v71) to (v81);

\node[draw, minimum height=0.5cm, minimum width=0.5cm, fill=black, opacity=1] (v91) at (6,1.5){};
\draw[->,thick] (v81) to (v91);

\node[draw,circle, accepting, minimum width=0.5cm] (s5) at (6,2.5){};
\draw[->,thick] (v81) to (s5);

\node[draw, circle, minimum width=0.5cm, fill=gray, opacity=0.7] (v102) at (7,2){};
\draw[->,thick] (v91) to (v102);
\node[draw, circle, minimum width=0.5cm] (v101) at (7,1){};
\draw[->,thick] (v91) to (v101);

\node[draw,circle, accepting, minimum width=0.5cm] (s6) at (8,1){};
\draw[->,thick] (v101) to (s6);

\node[draw,circle, accepting, minimum width=0.5cm] (s7) at (8,2){};
\draw[->,thick] (v102) to (s7);

\end{tikzpicture}
}
\end{center}
\caption{A tree $\stratTree{\sigma_1}$, associated with a strategy $\sigma_1$ of $\playerOne$, which represents all consistent plays with $\sigma_1$ whatever the behavior of $\playerTwo$.}
\label{fig:winningStratTree}
\end{figure}

Let us point out that Proposition~\ref{prop:noCycle} does not imply that $\sigma'_1$ is positional. Indeed, in Example~\ref{ex:memoIsRequired}, the finite-memory strategy is the only strategy that ensures the cost profile $(8,8)$, it satisfies conditions \emph{(i)} and \emph{(ii)} of Proposition~\ref{prop:noCycle} but requires memory.

\begin{proposition}  
\label{prop:fixPointReach}  
We have:
        \emph{(i)} for all $k \in \N$ and for all $v \in V$, $\iterEnsure{k}{v} \subseteq \iterEnsure{k+1}{v}$; and
        \emph{(ii)} there exists $k^* \leq |V|$, such that for all $v \in V$ and for all $\ell \in \N$, $\iterEnsure{k^*+ \ell}{v} = \iterEnsure{k^*}{v}$.
\end{proposition}

Properties stated in Proposition~\ref{prop:fixPointReach} hold by definition of $\iterEnsure{k}{v}$ and Proposition~\ref{prop:noCycle}. Moreover, the step $k^*$ is a particular step of the algorithm  that we call the \emph{fixpoint} of the algorithm. Notice that even if the fixpoint is reached at step $k^*$, the algorithm needs one more step in order to check that the fixpoint is reached. In the remaining part of this document, we write $\iterEnsure{*}{v}$ (resp. $\iter{*}{v}$) instead of $\iterEnsure{k^*}{v}$ (resp. $\iter{k^*}{v}$).

\subsubsection{Correctness}
\label{section:correctCompute}
The fixpoint algorithm (Algorithm~\ref{algo:fixPoint}) exactly computes the sets $\minimal(\ensure{v})$ for all $v \in V$, \IE for all $v \in V$, $\minimal(\ensure{v}) = \iter{*}{v}$. This is a direct consequence of Proposition~\ref{prop:ensureI}.

\newcommand{\restatePropEnsureI}{For all $k \in \N$ and all $v \in V$,
$ \minimal(\iterEnsure{k}{v}) = \iter{k}{v}.$}

\begin{proposition}
\label{prop:ensureI}
\restatePropEnsureI
\end{proposition}

\subsection{Time Complexity}
\label{section:complexity}
In this section we provide the time complexity of the fixpoint algorithm. The algorithm runs in polynomial time for the lexicographic order  and in exponential time for the componentwise order. In this latter case, if $\di$ is fixed, the algorithm is pseudo-polynomial, \IE  polynomial if the weights are encoded in unary. 

\newcommand{\restateThmComplexityLexico}{If $\lesssim$ is the lexicographic order, the fixpoint algorihtm runs in time polynomial in $|V|$ and $\di$.}

\begin{theorem}
\label{thm:complexityLexico}
\restateThmComplexityLexico
\end{theorem}

\newcommand{\restateThmComplexityComponent}{
If $\lesssim$ is the componentwise order, the fixpoint algorithm runs in time polynomial in $\maxW$ and $|V|$ and exponential in $\di$.
}

\begin{theorem}
\label{thm:complexityComponent}
\restateThmComplexityComponent
\end{theorem}

Theorem~\ref{thm:complexityLexico} relies on the fact that Line~\ref{lineAlgo:union} and Line~\ref{lineAlgo:inter} can be performed in polynomial time.  Indeed, in the lexicographic case, for all $k\in \N$ and all $v \in V$, $\iter{k}{v}$ is a singleton. Thus these operations amounts to computing  a minimum or a maximum between at most $|V|$ values. Theorem~\ref{thm:complexityComponent} can be obtained thanks to representations of upward closed sets and operations on them provided in~\cite{DelzannoR00}.

\subsection{Synthesis of Lexico-optimal and Pareto-optimal Strategies}
\label{section:correctStrat}

To this point, we have only explained the computation of the ensured values and we have not yet explained how lexico and Pareto-optimal strategies are recovered from the algorithm. This is the reason of the presence of Lines~\ref{lineAlgo:stratStart} to~\ref{lineAlgo:stratEnd} in Algorithm~\ref{algo:fixPoint}. Notice that in Line~\ref{lineAlgo:stratEnd}, we are allowed to assume that $\tup{x'}$ is in $ \iter{k}{v'}$ instead of $\uparrow \iter{k}{v'}$ because for all $k \in \N$, for all $v \in V_1 \backslash \targetSet$, $\iter{k+1}{v} = \minimal \left(\bigcup_{v' \in \successor(v)} \iter{k}{v} + \dWeight(v,v') \right).$

Roughly speaking, the idea behind the functions $f^k_v$ is the following. At each step $k \geq 1$ of the algorithm and for all vertices $v \in V_1 \backslash \targetSet$, we have computed the set $\iter{k}{v}$. At that point, we know that given $\tup{x} \in \iter{k}{v}$, $\playerOne$ can ensure a cost profile of $\tup{x}$ from $v$ in at most $k$ steps. The role of the function $f^{k}_v$ is to keep in memory which next vertex, $v' \in \successor(v)$, $\playerOne$ should choose and what is the cost profile $\tup{x'} = \tup{x} - \dWeight(v,v')$ which is ensured from $v'$ in at most $k-1$ steps. If different such successors exist one of them is chosen arbitrarily. 

In other words, $f^{k}_v$ provides information about how $\playerOne$ should behave locally in $v$ if he wants to ensure one of the cost profile $\tup{x} \in \iter{k}{v}$ from $v$ in at most $k$ steps. In this section, we explain how, from this local information, we recover a global strategy which is $\tup{x}$-Pareto optimal from $v$ (resp. lexico-optimal from $v$) for some $v \in V$ and some $\tup{x} \in \iter{*}{v}\backslash \{ \tup{\infty}\}$, if $\lesssim$ is the componentwise order (resp. the lexicographic order). 

 We introduce some additional notations. Since for all $k \in \N$ and all $v \in V$, $f^k_v: \iter{k}{v} \longrightarrow V \times \NInf^\di$, if $(v', \tup{x'}) = f^k_v(\tup{x})$ for some $\tup{x} \in \iter{k}{v}$ then, we write $f^k_v(\tup{x})[1] = v'$ and $f^k_v(\tup{x})[2] = \tup{x'}$. Moreover, for all $v \in V$, we write $f^*_v$ instead of $f^{k^*}_v$. Finally, if $X$ is a set of cost profiles, $\min_{\leqL}(X) = \{ \tup{x} \in X \mid \forall \tup{y} \in X,\,  (\tup{y} \leqL \tup{x} \implies \tup{y} = \tup{x}) \}$.

For all $u \in V$ and all $\tup{c} \in \iter{*}{u}\backslash \{ \tup{\infty}\}$, we define a strategy $\sigma^*_1 \in \stratSet{1}{u}$. The aim of this strategy is to ensure $\tup{c}$ from $u$ by exploiting the functions $f^*_v$. The intuition is as follows. If the past history is $hv$ with $v \in V_1$, $\playerOne$ has to take into account the accumulated partial costs $\dCost(hv)$ up to $v$ in order the make adequately his next choice to ensure $\tup{c}$ at the end of the play. For this reason, he selects some $\tup{x} \in \iter{*}{v}$ such that $\tup{x} \lesssim \tup{c} - \dCost(hv)$ and follows the next vertex dictated by $f^*_v(\tup{x})[1]$. 

\begin{definition}
\label{def:optiStrat}
Given $u \in V$ and $\tup{c} \in \iter{*}{u}\backslash \{ \tup{\infty} \}$, we define a strategy $\sigma^*_1 \in \stratSet{1}{u}$ such that for all $hv \in \hist_1(u)$, let $\cover(hv) = \{ \tup{x'} \in \iter{*}{v} \mid \tup{x'} \lesssim \tup{c} - \dCost(hv)\, \wedge \, \tup{x'} \leqL \tup{c} - \dCost(hv) \}$,
$$\sigma^*_1(hv) = \begin{cases} 
v' & \text{ for some } v' \in \successor(v), \text{ if } \cover(hv) = \emptyset \\
f^*_v(\tup{x})[1] & \text{ where } \tup{x} = \min_{\leqL} \cover(hv), \text{ if } \cover(hv) \neq \emptyset \end{cases}.$$
\end{definition}

\begin{remark}
\label{rem:technicalRem}
For some technical issues, when we have to select a representative in a set of incomparable elements, the $\leqL$ order is used in the definitions of $\cover(hv)$ and of the strategy. Nevertheless, Definition~\ref{def:optiStrat} holds both for the lexicographic and the componentwise orders.
\end{remark}

For all $u \in V$ and $\tup{c} \in \iter{*}{u}\backslash \{ \tup{\infty} \}$, the strategy $\sigma^*_1$ defined in Definition~\ref{def:optiStrat} ensures $\tup{c}$ from $u$. In particular, $\sigma^*_1$ is lexico-optimal and $\tup{c}$-Pareto-optimal from $u$.

\newcommand{\restateThmOptiStrat}{Given $u \in V$ and $\tup{c} \in \iter{*}{u}\backslash \{ \tup{\infty} \}$, the strategy $\sigma^*_1\in \stratSet{1}{u}$ defined in Definition~\ref{def:optiStrat} is such that for all $\sigma_2 \in \stratSet{2}{u}$,
    $ \dCost(\outcome{\sigma^*_1}{\sigma_2}{u}) \lesssim \tup{c}.$}

\begin{theorem}
    \label{thm:optiStrat}
    \restateThmOptiStrat
\end{theorem}

Although the strategy defined in Definition~\ref{def:optiStrat} is a lexico-optimal strategy from $u$, it requires finite-memory. However, for the lexicographic order, positional strategies are sufficient.

\newcommand{\restateStratOptiLexicoPosition}{ If $\lesssim$ is the lexicographic order, for $u \in V$ and $\tup{c} \in \iter{*}{u} \backslash \{ \tup{\infty} \}$, the strategy $\vartheta^*_1$ defined as: for all $hv \in \hist_1(u)$,
    $\vartheta^*_1(hv) = f^*_v(\tup{x})[1]$ where $\tup{x}$ is the unique cost profile in $\iter{*}{v}$, is a positional lexico-optimal strategy from $u$.}

\begin{proposition}
\label{prop:stratOptiLexicoPositional}
    \restateStratOptiLexicoPosition
\end{proposition}


\section{Constrained Existence}
\label{section:constrainedExistence}

Finally, we focus on the constrained existence problem (\CE problem). 

\begin{theorem}
\label{thm:constrainedResLexico}
If $\lesssim$ is the lexicographic order, the \CE problem is solved in $\Poly$.
\end{theorem}

Theorem~\ref{thm:constrainedResLexico} is immediate since, in the lexicographic case, we can compute the upper value $\upValue{v_0}$ in polynomial time (Theorem~\ref{thm:complexityLexico}).

\begin{theorem}
\label{thm:constrainedResComponent}
    If $\lesssim$ is the componentwise order, the \CE problem is $\PSPACE$-complete.
\end{theorem}

\textbf{$\PSPACE$-easiness.}
Proposition~\ref{prop:noCycle}
allows us to prove that the \CE problem with the componentwise order is in $\APTime$. The alternating Turing machine works as follows: all vertices of the game owned by $\playerOne$ (resp. $\playerTwo$) correspond to disjunctive states (resp. conjunctive states). A path of length $|V|$ is accepted if and only if, \emph{(i)} the target set is reached along that path and \emph{(ii)} the sum of the weights until an element of the target set is $\leqC \tup{x}$. If such a path exists, there exists a strategy of $\playerOne$ that ensures the cost profile $\tup{x}$.  This procedure is done in polynomial time and since $\APTime = \PSPACE$, we get the result.

\textbf{$\PSPACE$-hardness.} The hardness part of Theorem~\ref{thm:constrainedResComponent} is based on a polynomial reduction from the \textsc{Quantified Subset-Sum} problem, proved $\PSPACE$-complete~\cite[Lemma 4]{Travers06}. This problem is defined as follows. Given a set of natural numbers $N = \{a_1,\ldots, a_n\}$ and a threshold $T \in \N$, we ask if the formula $\Psi = \exists x_1 \in \{ 0, 1 \} \, \forall x_2 \in \{ 0, 1 \} \, \exists x_3 \in \{ 0, 1 \} \ldots \exists x_n \in \{0,1 \},\, \sum_{1 \leq i \leq n} x_i a_i = T$ is true.

In the same spirit as for the QBF problem~\cite{Sipser}, the \textsc{Quantified Subset-Sum} problem can be seen as a two-player game in which two players (Player $\exists$ and Player $\forall$) take turn in order to assign a value to the variables $x_1$, \ldots, $x_n$: Player~$\exists$ (resp. Player~$\forall$) chooses the value of the variables under an existential quantifier (resp. universal quantifier). When a player assigns a value  $1$ to a variable $x_k$, $1 \leq k \leq n$, this player selects the natural number $a_k$ and he does not select it if $x_k$ is assigned to $0$. The goal of Player~$\exists$ is that the sum of the selected natural numbers is exactly equal to $T$ while the goal of Player~$\forall$ is to avoid that. Thus, with this point of view, the formula $\Psi$ is true if and only if Player~$\exists$ has a winning strategy.

\begin{figure}[ht]
    \centering
    \scalebox{0.9}{\begin{tikzpicture}
    \node[draw,circle, inner sep=2pt] (x1) at (0,0){$x_1$};
\node[draw, circle, inner sep=2pt] (y1) at (1.5,1){$x^1_1$};
\node[draw, circle, inner sep=2pt] (n1) at (1.5,-1){$x^0_1$};
\node[draw] (x2) at (3,0){$x_2$};
\node[draw, circle, inner sep=2pt] (y2) at (4.5,1){$x^1_2$};
\node[draw, circle, inner sep=2pt] (n2) at (4.5,-1){$x^0_2$};
\node[draw, circle, inner sep=2pt] (x3) at (6,0){$x_3$};
\node (dots1) at (7.5,0){$\ldots$};

\node[draw,circle, inner sep=2pt] (xn) at (9,0){$x_n$};
\node[draw, circle, inner sep=2pt] (yn) at (10.5,1){$x^1_n$};
\node[draw, circle, inner sep=2pt] (nn) at (10.5,-1){$x^0_n$};
\node[draw, circle, inner sep=3pt, accepting] (target) at (12,0){$y$};

\draw[->,thick] (x1) to node[inner sep=0pt, fill=white]{$(a_1,0)$} (y1);
\draw[->,thick] (x1) to node[inner sep=0pt, fill=white]{$(0, a_1)$} (n1);
\draw[->,thick] (n1) to (x2);
\draw[->,thick] (y1) to (x2);
\draw[->,thick] (x2) to node[inner sep=0pt, fill=white]{$(a_2,0)$} (y2);
\draw[->,thick] (x2) to node[inner sep=0pt, fill=white]{$(0,a_2)$} (n2);

\draw[->,thick] (n2) to (x3);
\draw[->,thick] (y2) to (x3);
\draw[->,thick] (xn) to node[inner sep=0pt, fill=white]{$(a_n,0)$} (yn);
\draw[->,thick] (xn) to node[inner sep=0pt, fill=white]{$(0,a_n)$} (nn);

\draw[->,thick] (nn) to (target);
\draw[->,thick] (yn) to (target);

\draw[->,thick] (target) to[loop right] (target);
    \end{tikzpicture}}
    \caption{Initialized game used in the reduction for the $\PSPACE$-hardness. Edges with no label are assumed to be labeled by $(0,0)$.}
    \label{fig:hardnessGame}
\end{figure}

 In order to encode the equality presents in the Quantified Subset-Sum problem, we use the two inequality constraints in a two-player two-weighted game. The arena of the game is given in Figure~\ref{fig:hardnessGame}, $\playerOne$ aka Player~$\exists$ (resp. $\playerTwo$ aka Player~$\forall$) owns the rounded (resp. rectangular)  vertices corresponding to variables under an existential (resp. universal) quantifier. The target set is only composed of the vertex $y$.  When a player assigns the value $1$ to a variable $x_k$, $1 \leq k \leq n$, the resulting weight of this choice is $(a_k,0)$, while if he assigns the value $0$, the weight is $(0,a_k)$. In this way, if we sum all those weights, we have on the first component the sum of the selected natural numbers and on the second component the sum of the not selected natural numbers. Thereby, there exists a strategy of $\playerOne$ that ensures the cost profile $(T, \sum_{1\leq i \leq n} a_i - T)$ if and only if the formula $\Psi$ is true.

Notice that as $\playerOne$ can consider the previous assignations of variables $x_1, \ldots, x_{k-1}$ to  choose the assignation of a variable $x_k$ to $0$ or $1$, the resulting strategy needs finite-memory.


\section{Permissiveness of Multi-strategies}
\label{section:permissiveMultiStrat}

Even when a strategy that ensures some cost is synthesized, its implementation may fail. This can be due to the occurrence of errors; for example, the action prescribed by the strategy may be unavailable. Synthesizing \emph{robust strategies} against such perturbations is therefore essential.
To address these robustness issues, the classic notion of a player’s strategy can be replaced by the notion of  \emph{multi-strategies}: a multi-strategy for $\playerOne$ prescribes a set of allowed possible actions when it is $\playerOne$'s turn to play (see, for example, \cite{BJW02,BDMR09}), instead of a single action. Thus, once a multi-strategy is fixed for each player, there are several paths in the game graph that are consistent with these multi-strategies from a given initial vertex.
 In this setting, we aim at synthesizing the most permissive multi-strategies. 

Intuitively, a multi-strategy is more \emph{permissive} than another if the first allows more behaviors than the second. The permissiveness of multi-strategies may be compared in different ways. A~qualitative view of permissiveness is studied in \cite{BJW02}, where a multi-strategy is more permissive than another if the set of resulting plays includes those of the second multi-strategy. A~quantitative view is addressed in \cite{BDMR09} via the notion of \emph{penalty} of multi-strategies, where a cost is associated with each edge not chosen by the multi-strategy. Thus, the penalty of a multi-strategy is the highest sum of blocked edges along a play consistent with the multi-strategy. We follow this latter approach in this document.

These notions of permissiveness and penalty raise different problems. Some  of them deal with the existence of a multi-strategy of $\playerOne$ under fixed constraints.
 
\begin{description}
    \item[MCE1] problem aims to ensure both some costs and some penalty, \IE given two thresholds $c$ and $p$, does there exist a multi-strategy $\multiStrat$ for $\playerOne$ such that \emph{(i)} the worst cost of a play consistent with $\multiStrat$ is less than or equal to $c$ and \emph{(ii)} the penalty of $\multiStrat$ is less than or equal to $p$?
    \item[MCE2] problem focuses first on optimizing the costs and only then on optimizing the penalty, \IE given two thresholds $c$ and $p$, does there exist a multi-strategy $\multiStrat$ of $\playerOne$ such that (the worst cost of a play consistent with $\multiStrat$, the penalty of $\multiStrat$) $\leqL (c,p)$? 
    \item[MCE3] problem takes the reverse point of view by optimizing first the penalty and then the costs, \IE given two thresholds $p$ and $c$, does there exist a multi-strategy $\multiStrat$ of $\playerOne$ such that (the penalty of $\multiStrat$, the worst cost of a play consistent with $\multiStrat$) $\leqL (p,c)$? 
\end{description}

Other problems consider the existence of an optimal multi-strategy of $\playerOne$ (regarding its worst ensured cost)  or a most permissive multi-strategy of $\playerOne$. The problems of our interrest are the following:

\begin{description}
    \item[MEV1] problem computes the minimal set of pairs, w.r.t. the component-wise order, $(c,p) \in \N \times \N$ such that the answer to the MCE1 problem is yes.
    \item[MEV2] problem computes the minimal pair, w.r.t. the lexicographic order, $(c,p)$ such that the answer to the MCE2 problem is yes.
    \item[MEV3] problem computes the minimal pair, w.r.t. the lexicographic order, $(p,c)$ such that the answer to the MCE3 problem is yes.
\end{description}

\paragraph*{Contributions w.r.t. Permissiveness of Multi-strategies.} 
In the remainder of this paper, we show how the results obtained for multi-weighted reachability games can be exploited to derive solutions to the above-mentioned problems. In Section~\ref{section:fromMultiToPermi}, a multi-weighted reachability game with two dimensions is built from a quantitative reachability game. Roughly speaking the arena of the $2$-weighted reachability game explicitly represents each possible $\playerOne$'s choice by using a multi-strategy from a vertex $v$, \IE  each subset of successors of $v$. Moreover, the two dimensions of weights on a given edge allow to keep track both the original weight of this edge and the penalty resulting from a $\playerOne$'s choice in this new arena. This construction allows obtaining a useful correspondence between the penalty and the worst cost ensured by a multi-strategy in the quantitative reachability game and the (two-dimensional) cost ensured by a (simple) strategy in the associated multi-weighted reachability game. Thanks to this result and results about multi-weighted reachability games, we prove in Section~\ref{section:constrainedExistencePermi} that the MCE1 problem is $\PSPACE$-complete while the MCE2 and MCE3 problems belong to $\NP$. Finally, in Section~\ref{section:ensuredValuesPermi}, we explain how the pairs of values computed by the MEV1 problem (resp. MEV2 and MEV3 problems) can be computed thanks to an algorithm whose execution time is exponential (resp. thanks to an algorithm that makes a polynomial number of calls to a decision problem in $\NP$).

\paragraph*{Related Works w.r.t. Permissiveness of Multi-strategies.}
 Permissiveness of multi-strategies may be compared in different ways. We mention, in a non-exhaustive way, some related works. In \cite{BJW02}, permissiveness in parity games is studied by considering a qualitative view of permissiveness. Roughly speaking a multi-strategy is more permissive than another one if the set of plays consistent with the first one contains the set of plays consistent with the second one. Unfortunately, there does not necessarily exist a most permissive strategy with this view of permissiveness.

The quantitative view of permissiveness explained above and which we decide to follow is defined in \cite{BDMR09}. Several penalty measures and games are used, and the complexity of computing the most permissive strategies in this context is given. More general parity objectives are then studied in \cite{BouyerMOU11}.  Moreover, penalties are also used in~\cite{GoeminneM25} in order to define and study permissive equilibria in multiplayer reachability games. Other methods have explored permissiveness in two-player games using templates to concisely represent multiple strategies in graph games \cite{AnandNayakSchmuck}. The same approach is employed for the synthesis of secure equilibria in multiplayer games  \cite{NayakS24}.

These issues of permissiveness are also studied in other types of games: stochastic games~\cite{DragerFK0U15} and timed games~\cite{BouyerFM15,ClementJMM20}.

\subsection{Preliminaries}

\paragraph{Quantitative Reachability Games} In the rest of the document we aim at solving problems related to $1-$weighted reachability games, that we now call \emph{quantitative reachability games}. In order to avoid any confusion about notation in the rest of this document, we write $\game = (\arena, \targetSet, \cost)$ and $\edgeCost: E \longrightarrow \NInf$ instead of $\game_1$ and $\weight_1$ respectively. Recall that in this setting $\playerOne$ only wants to minimize the accumulated costs, given par $\edgeCost$, until reaching the target set $\targetSet \subseteq V$. No matter what $\playerTwo$ does.

\paragraph{Multi-strategies} A \emph{multi-strategy} of $\playerOne$ from a vertex $v$ is a function $\multiStrat_1 : \hist_1(v) \longrightarrow 2^V \setminus \{ \emptyset \}$ that assigns to each history $hu \in \hist_1(v)$ a non-empty set of vertices $\multiActions \subseteq V$ such that for all $u' \in \multiActions$, $(u,u') \in E$. 
Notice that a (simple) strategy $\sigma_1$ from $v$ can be seen as a multi-strategy $\multiStrat_1$ from $v$ where, for all $hu \in \hist_1(v)$, $\multiStrat_1(hu)$ is the singleton $\{ \sigma_1(hu) \}$. 
The set of multi-strategies of $\playerOne$ from $v$ is denoted by $\msSet{1}{v}$.

Given a multi-strategy $\multiStrat_1$ of $\playerOne$ and $v \in V$, $\multiOutcome{\multiStrat_1}{v}$ is the set of plays beginning in $v$ that are both consistent with the choices dictated by the multi-strategy $\multiStrat_1$ and with all the possible behaviors of $\playerTwo$.
Before formally defining this set of plays, we define the set of finite prefixes of such plays, written $\multiOutcomeHistory{\multiStrat_1}{v}$, as follows:
\begin{itemize}
    \item $v \in \multiOutcomeHistory{\multiStrat_1}{v}$
    \item for each history $h \in \multiOutcomeHistory{\multiStrat_1}{v}$:
    \begin{itemize}
        \item if $\last(h) \in V_1$, then for all $u \in \multiStrat_1(h)$, $hu \in \multiOutcomeHistory{\multiStrat_1}{v}$;
        \item if $\last(h) \in V_2$, then for all $u \in \successor(\last(h))$, $hu \in \multiOutcomeHistory{\multiStrat_1}{v}$.
    \end{itemize}
\end{itemize}

It follows that a play is part of set $\multiOutcome{\multiStrat_1}{v}$  if and only if all its finite prefixes are part of set $\multiOutcomeHistory{\multiStrat_1}{v}$. 

The definition of the set $\multiOutcome{\multiStrat_1}{v}$, allows to compute what is the worst cost that this multi-strategy can ensure from $v$ whatever the behavior of $\playerTwo$. This \emph{worst ensured cost} of $\multiStrat_1$ is written $\wCost(\multiOutcome{\multiStrat_1}{v})$ and is formally defined as $$\wCost(\multiOutcome{\multiStrat_1}{v}) = \sup \{ \cost(\rho) \mid \rho \in \multiOutcome{\multiStrat_1}{v}\}.$$

Finally, given an initialized quantitative reachability game $(\game,v_0)$, a multi-strategy $\multiStrat_1$ of $\playerOne$ is called a \emph{winning multi-strategy}  in $(\game,v_0)$ if all plays beginning in $v_0$ that are consistent with $\multiStrat_1$ are winning for $\playerOne$,  \IE  for all $\rho = \rho_0\rho_1 \ldots \in \multiOutcome{\multiStrat_1}{v_0} $, there exists $n \in \N$ such that $\rho_n \in \targetSet$. In particular, we have that $\wCost(\multiOutcome{\multiStrat_1}{v_0}) < +\infty$.

\begin{example}
\label{ex:runningExPermi}
Let us consider the quantitative reachability game $\game$ whose arena $\arena = (V_1,V_2, E, \edgeCost)$ is shown in Figure~\ref{fig:quantitativeReachGame}.
In this example, $V_1 = \{v_0,v_1,v_2,v_3,v_6,\}$, $V_2= \{v_4,v_5,v_7,v_8\}$ and $\targetSet = \{ v_5 \}$.
The weights of the edges are given by the numbers to the left of the $\mid$ symbol, \emph{e.g.,} $\edgeCost(v_0,v_8)=1$ and $\edgeCost(v_1,v_2)=2$. Do not pay attention to the numbers to the right of the $\mid$ symbol for the moment. For all edges without a label, we assume that the weight is equal to $1$, \emph{e.g.,} $\edgeCost(v_1,v_3) = 1$.
\begin{figure}[ht!]
\centering
    \begin{tikzpicture}
    \node[draw, circle, inner sep=3pt] (v0) at (0,0){$v_0$};
    \node[draw,circle, inner sep=3pt] (v1) at (2,-1){$v_1$};
    \node[draw,circle, inner sep=3pt] (v2) at (4,0){$v_2$};
    \node[draw,circle, inner sep=3pt] (v3) at (4,-2){$v_3$};
    \node[draw, inner sep=3pt] (v4) at (6,-1){$v_4$};
    \node[draw, inner sep=3pt,accepting] (v5) at (8.5,-1){$v_5$};
    \node[draw,circle, inner sep=3pt] (v6) at (6,1){$v_6$};
    \node[draw, inner sep=3pt] (v7) at (8.5,1){$v_7$};
    \node[draw, inner sep=3pt] (v8) at (2,1){$v_8$};

    \draw[->,thick] (v0) to (v1);
    \draw[->,thick] (v0) to node[fill=white, inner sep=0pt]{$1\mid 2$} (v8);
    \draw[->,thick] (v1) to node[fill=white, inner sep=0pt]{$2\mid 1$} (v2);
    \draw[->,thick] (v1) to  (v3);
    \draw[->,thick] (v8) to  (v1);
    \draw[->,thick] (v2) to (v4);
    \draw[->,thick] (v3) to node[fill=white, inner sep=0pt]{$2\mid 1$} (v4);
    \draw[->,thick] (v4) to (v6);
    \draw[->,thick] (v6) to (v5);
    \draw[->,thick] (v6) to node[fill=white, inner sep=0pt]{$10 \mid 10$} (v7);
    \draw[->,thick] (v7) to (v5);
    \draw[->,thick] (v4) to (v5);
    \draw[->,thick] (v8) to [loop right] (v8);
    \draw[->,thick] (v5) to [loop right] (v5);
    \end{tikzpicture}
    \caption{Example of a quantitative reachability game. The target set is $\targetSet = \{v_5\}$. The weight (resp. penalty) of an edge is given by the number to the left (resp. right) of the $\mid$ symbol. An edge without any label is assumed to be labeled by $1 \mid 1$. }
    \label{fig:quantitativeReachGame}

\end{figure}

Let us consider the multi-strategy $\multiStrat_1$ of $\playerOne$ defined as follows: $\multiStrat_1(v_0) = \{ v_1 \}$;
for all $hv_1 \in \hist(v_0)$, $\multiStrat_1(hv_1) = \{ v_2,v_3\}$;
for all $hv_2$ and $hv_3$ in $\hist(v_0)$, $\multiStrat_1(hv_2) = \multiStrat(hv_3) = \{ v_4 \}$; and
for all $hv_6 \in \hist(v_0)$, $\multiStrat_1(hv_6) = \{ v_5\}$. This multi-strategy is shown in Figure~\ref{fig:quantitativeReachGameMultiStrat}: blue edges are selected by $\multiStrat_1$ while dotted edges are blocked by $\multiStrat_1$. We have that $\multiOutcome{\multiStrat_1}{v_0} = \{  v_0v_1v_2v_4v_6v_5^\omega, v_0v_1v_2v_4v_5^\omega, \allowbreak v_0v_1v_3v_4v_6v_5^\omega, v_0v_1v_3v_4v_5^\omega  \}$.  It follows that $\wCost(\multiOutcome{\multiStrat_1}{v_0}) = \cost(v_0v_1v_2v_4v_6v_5^\omega) \allowbreak= 6$.

\begin{figure}[ht!]
\centering
    \begin{tikzpicture}
    \node[draw, circle, inner sep=3pt] (v0) at (0,0){$v_0$};
    \node[draw,circle, inner sep=3pt] (v1) at (2,-1){$v_1$};
    \node[draw,circle, inner sep=3pt] (v2) at (4,0){$v_2$};
    \node[draw,circle, inner sep=3pt] (v3) at (4,-2){$v_3$};
    \node[draw, inner sep=3pt] (v4) at (6,-1){$v_4$};
    \node[draw, inner sep=3pt,accepting] (v5) at (8.5,-1){$v_5$};
    \node[draw,circle, inner sep=3pt] (v6) at (6,1){$v_6$};
    \node[draw, inner sep=3pt] (v7) at (8.5,1){$v_7$};
    \node[draw, inner sep=3pt] (v8) at (2,1){$v_8$};

    \draw[->,ultra thick,blue] (v0) to (v1);
    \draw[->,thick,dotted] (v0) to node[fill=white, inner sep=0pt]{$1\mid 2$} (v8);
    \draw[->,ultra thick, blue] (v1) to node[fill=white, inner sep=0pt]{\color{black}$2\mid 1$} (v2);
    \draw[->,ultra thick, blue] (v1) to  (v3);
    \draw[->,thick] (v8) to  (v1);
    \draw[->,ultra thick, blue] (v2) to (v4);
    \draw[->,ultra thick, blue] (v3) to node[fill=white, inner sep=0pt]{\color{black}$2\mid 1$} (v4);
    \draw[->,thick] (v4) to (v6);
    \draw[->,ultra thick,blue] (v6) to (v5);
    \draw[->,thick, dotted] (v6) to node[fill=white, inner sep=0pt]{$10 \mid 10$} (v7);
    \draw[->, thick] (v7) to (v5);
    \draw[->,thick] (v4) to (v5);
    \draw[->,thick] (v8) to [loop right] (v8);
    \draw[->,thick] (v5) to [loop right] (v5);
    \end{tikzpicture}
    \caption{The solid edges whose source vertex is owned by $\playerOne$, \IE  blue edges, represent a multi-strategy of $\playerOne$.}
    \label{fig:quantitativeReachGameMultiStrat}

\end{figure}
\end{example}

\paragraph{Permissiveness and Penalties}
Given a quantitative reachability game, our aim is to find a trade-off between a multi-strategy with the least possible worst ensured cost (an \emph{optimal multi-strategy}) and a multi-strategy which allows as many behaviors of $\playerOne$ as possible (a \emph{most permissive multi-strategy}). 

The permissiveness of multi-strategies may be compared in different ways. We here use the concept of penalty of a multi-strategy already defined in a two-player zero-sum setting in~\cite{BDMR09}.
This penalty depends on weights associated with edges not chosen by the multi-strategy, \IE blocked edges, and we prefer a multi-strategy with a penalty as small as possible.
In order to define the penalties properly,  we equip the game with a penalty function $\edgePenal: E \longrightarrow \N$ assigning a non-negative penalty to each edge. Given an edge $(v,v')\in E$ such that $v \in V_1$, $\playerOne$ obtains a penalty of $\edgePenal(v,v')$ if he does not select $v'$ from $v$ in his multi-strategy.
Moreover, such penalties are accumulated along a play.
Formally, given a multi-strategy $\multiStrat_1$ of $\playerOne$ from $v$, we first define the \emph{penalty of $\playerOne$ w.r.t.~$\multiStrat_1$} along a play $\rho = \rho_0\rho_1 \cdots \in \play(v)$, denoted by $\penal_{\multiStrat_1}(\rho)$, by induction on the length of its prefixes:
\begin{itemize}
    \item $\penal_{\multiStrat_1}(\varepsilon) = 0 $ where $\varepsilon$ denotes the empty prefix;
    \item for $h = \rho_0\cdots\rho_k$,
    $\displaystyle \penal_{\multiStrat_1}(hv) =
    \begin{cases}\penal_{\multiStrat_1}(h) + \sum_{v' \in \successor(v) \setminus \multiStrat_1(hv)} \edgePenal(v,v') & \text{ if } v \in V_1 
    \\ \penal_{\multiStrat_1}(h) & \text{ otherwise}
    \end{cases}$;
    \item $\penal_{\multiStrat_1}(\rho) = \lim_{k \rightarrow +\infty} \penal_{\multiStrat_1}(\rho_0\cdots\rho_k)$. Since this is a non-decreasing sequence of natural numbers, this limit is either a natural number or $+\infty$.
\end{itemize}

Finally, the penalty of a multi-strategy $\multiStrat_1$ from $v$, written  $\penal(\multiOutcome{\multiStrat_1}{v})$, is the worst penalty of the plays consistent with $\multiStrat_1$, \IE $\penal(\multiOutcome{\multiStrat_1}{v}) = \sup \{ \penal_{\multiStrat_1}(\rho) \mid \rho \in \multiOutcome{\multiStrat_1}{v}\}$.

In the remaining part of this document, we  assume that a quantitative reachability game is always equipped with a penalty function $\edgePenal$.

\begin{example}
   Let us consider the multi-strategy $\multiStrat_1$ defined in Example~\ref{ex:runningExPermi}. In this example, the penalty of a blocked edge $e \in E$ is given by the number to the right of the $\mid$ symbol in the label of the vertex $e$. In order to compute the penalty of $\multiStrat_1$, we first compute the penalties of plays in $\multiOutcome{\multiStrat_1}{v_0}$: $\penal_{\multiStrat_1}(v_0v_1v_2v_4v_6v_5^\omega)= 12$, $\penal_{\multiStrat_1}(v_0v_1v_2v_4v_5^\omega) = 2$, \allowbreak $\penal_{\multiStrat_1}(v_0v_1v_3v_4v_6v_5^\omega) = 12$
     and $\penal_{\multiStrat_1}(v_0v_1v_3v_4v_5^\omega) = 2$. It follows that $\penal(\multiOutcome{\multiStrat_1}{v}) = 12$.
\end{example}

In the quantitative reachability game provided in Example~\ref{ex:runningExPermi}, the least possible worst ensured cost is $6$ and the least possible penalty to ensure this cost is $12$. This cost and this penalty are achievable thanks to $\multiStrat_1$ provided in Figure~\ref{fig:quantitativeReachGameMultiStrat}. However, it is possible to obtain a better penalty to the detriment of the worst ensured cost. Obviously, the strategy that always allows all possible successors has a penalty of $0$ but the play $v_0v_8^\omega$ is consistent with such a multi-strategy and is not winning for $\playerOne$. In light of this, we prefer to look for a multi-strategy which is as permissive as possible, \IE with the least possible penalty, but such that all consistent plays are winning for $\playerOne$. In Example~\ref{ex:runningExPermi},  the least penalty of a winning multi-strategy is equal to $2$ and the least possible worst ensured cost of a multi-strategy that ensures this penalty is equal to $16$. This can be achieve thanks to the multi-strategy $\multiStrat'_1$ of $\playerOne$ such that $\multiStrat'_1$ is equal to $\multiStrat_1$ except for histories that ends in $v_6$ from which both $v_5$ and $v_7$ are selected by $\multiStrat'_1$.

This example highlights that if $\playerOne$ wants to find a winning multi-strategy that minimizes both the worst ensured cost and the penalty, he has to find a compromise. As in the first part of this paper, we deal with this trade-off  \emph{(i)} by ranking \emph{a priori} these two values according to a priority order and then comparing them thanks to a lexicographic order or \emph{(ii)} by comparing the worst ensured cost and the penalty component by component and by choosing \emph{a posteriori} which pair of worst ensured cost and penalty we prefer to obtain.

\paragraph{Studied Problems} In the same vein as the first part of this paper we consider two kind of problems: the \emph{ensured values by  multi-strategies problems} and the \emph{constrained existence of  multi-strategies problems}. Before formally defining these problems, we need to introduce some notations.

Given a vertex $v\in V$ and a multi-strategy $\multiStrat_1$ of $\playerOne$ from $v$, $$\CP(\multiOutcome{\multiStrat_1}{v}) = (\wCost(\multiOutcome{\multiStrat_1}{v}),\penal(\multiOutcome{\multiStrat_1}{v}))  \text{ and }  \PC(\multiOutcome{\multiStrat_1}{v}) = (\penal(\multiOutcome{\multiStrat_1}{v}), \wCost(\multiOutcome{\multiStrat_1}{v})).$$

We consider what are the worst ensured costs and penalties that can be ensured by a multi-strategy of $\playerOne$ from a given vertex $v$. Those costs and penalties may be compared in three different ways with \emph{(i)}  a componentwise order, \emph{(ii)} a lexicographic order that gives priority to the worst ensured cost and \emph{(iii)} a lexicographic order that gives priority to the penalty. For that reason, for all $v \in V$, we define $\msEnsure^{\kappa}_{\lesssim}(v)$ where $\kappa$ is either $\PC$ or $\CP$ and $\lesssim$ is either $\leqC$ or $\leqL$:\footnote{These conventions are followed in the remaining part  of this document.}

$$ \msEnsure^{\kappa}_{\lesssim}(v) = \{ (x,y) \in \NInf \times \NInf \mid \exists \multiStrat_1 \in \msSet{1}{v} \text{ st. } \kappa(\multiOutcome{\multiStrat_1}{v}) \lesssim (x,y) \}. $$

Additionally, $\msPareto(v) = \minimal (\msEnsure^{\CP}_{\leqC}(v))$,  $\cVal(v) = \minimal (\msEnsure^{\CP}_{\leqL}(v)) $ and $\pVal(v) = \minimal (\msEnsure^{\PC}_{\leqL}(v)) $.
As for (simple) strategies, given a pair $(x,y) \in \NInf \times \NInf$, we say that a multi-strategy $\multiStrat_1$ of $\playerOne $ ensures $(x,y)$ from $v$ if $\kappa(\multiOutcome{\multiStrat_1}{v}) \lesssim (x,y)$, where $\kappa$ is either $\CP$ or $\PC$ and $\lesssim$ is either $\leqC$ or $\leqL$. \\

Given an initialized quantitative reachability game $(\game,v_0)$, the ensured values by multi-strategies problems listed below involve computing $\msPareto(v_0)$, $\cVal(v_0)$ and $\pVal(v_0)$. 

\begin{definition}[Ensured Values by Multi-strategies Problems - MEV Problems]
Let $(\game, v_0)$ be an initialized quantitative reachability game. We distinguish three problems:
\begin{enumerate}
    \item (MEV1) Computing the Pareto frontier $\msPareto(v_0)$.
    \item (MEV2) Computing $\cVal(v_0)$.
    \item (MEV3) Computing $\pVal(v_0)$.
\end{enumerate} 
\end{definition}

\begin{theorem}
Given an initialized quantitative reachability game $(\game,v_0)$,

\begin{enumerate}    
    \item The set $\msPareto(v_0)$ can be computed in exponential time. \label{label:MEV1}
    \item The values $\cVal(v_0)$ and $\pVal(v_0)$ can be computed thanks to an algorithm that makes a polynomial number of calls to a decision problem in $\NP$.\label{label:MEV23}
\end{enumerate}
\end{theorem}

Statement~\ref{label:MEV1} is obtained by Proposition~\ref{prop:paretoFrontierPermi} and Statement~\ref{label:MEV23} is restated and proved as Proposition~\ref{prop:valuesPermi}. \\

We also consider other closely related problems. Given an initialized quantitative reachability game $(\game, v_0)$ and a pair $(x,y) \in \N \times \N$, we would like to decide whether there exists a multi-strategy of $\playerOne$ which ensures $(x,y)$ from $v_0$. This leads to three variants of this problem that we call constrained existence of multi-strategies problems.

\begin{definition}[Constrained Existence of Multi-strategies Problems - MCE Problems]
Let $(\game, v_0)$ be an initialized quantitative reachability game and $(x,y) \in \N \times \N$. We distinguish three problems:

\begin{itemize}
    \item (MCE1) Does there exist a multi-strategy $\multiStrat_1$ such that $\CP(\multiOutcome{\multiStrat_1}{v_0}) \leqC (x,y)$?
    \item (MCE2) Does there exist a multi-strategy $\multiStrat_1$ such that $ \CP(\multiOutcome{\multiStrat_1}{v_0}) \leqL (x,y)$?
    \item (MCE3) Does there exist a multi-strategy $\multiStrat_1$ such that $\PC(\multiOutcome{\multiStrat_1}{v_0}) \leqL (x,y)$?
\end{itemize}
\end{definition}

\begin{remark}
    Notice that we are looking for winning multi-strategies, this is the reason why the upper bound corresponding to the worst ensured cost is assumed to be a natural number. Moreover, if there exists a multi-strategy such that its worst ensured cost is less than or equal to $c$, there exists another multi-strategy with worst ensured cost less than or equal to $c$ and such that its penalty is finite. Indeed, once an element of the target set is reached, $\playerOne$ can select all possible successors without modifying the costs of the consistent plays. This is the reason why the upper bound corresponding to the penalty is assumed to be a natural number.
\end{remark}

\begin{theorem}
    Given an initialized quantitative reachability game $(\game,v_0)$,
    \begin{enumerate}
        \item The MCE1 problem is $\PSPACE$-complete. \label{label:MCE1}
        \item The MCE2 and MCE3 problems belong to $\NP$.\label{label:MCE23}
    \end{enumerate}
\end{theorem}

Statement~\ref{label:MCE1} is proved by Proposition~\ref{prop:constrainedComponentPermi}, while Statement~\ref{label:MCE23}
is obtained thanks to Proposition~\ref{prop:constrainedLexicoPermi}.

\begin{remark}
    Notice that for the componentwise order, since we do not impose any preference on the components, the philosophy behind \emph{(i)} computing the minimal set of ensured values with $\kappa = \PC$  instead of  $\kappa = \CP$ and \emph{(ii)} deciding the MCE1 problem with $\PC$ instead of $\CP$ is the same.
\end{remark}


\subsection{From Multi-weighted Reachability to Permissiveness, and vice versa}
\label{section:fromMultiToPermi}

In this section we introduce the notion of \emph{extended game of a quantitative reachability game}. This is a 2-weighted reachability game, as studied in the first part of this paper, which enjoys some useful properties about correspondence between the worst ensured cost and the penalty of multi-strategies in the initial quantitative reachability game and the cost profile of (simple) strategies in its associated extended game. This property is exploited in Section~\ref{section:constrainedExistencePermi} and Section~\ref{section:ensuredValuesPermi} in order to solve the MEV problems and the MCE problems thanks to results obtained for multi-weighted reachability games in the first part of this paper.

\paragraph{Extended Game of a Quantitative Reachability Game} Given an initialized reachability game $(\game,v_0)$, we first explain how we can build a 2-weighted reachability game $(\mathcal{X}_\kappa,v_0)$ which is an extended game of $(\game,v_0)$. The index $\kappa$ is either $\CP$ when the first component of weights in the $2$-weighted reachability game represents the weight of an edge and the second component represents the penalty of an edge or $\PC$ when these two components are permuted. This construction is taken from Bouyer \emph{at al.}~\cite{BDMR09}. A formal definition of the associated extended game of a quantitative reachability game is provided in Appendix~\ref{appendix:permissiveMultiStrat}. \\

We now provide the intuition for the construction of $\extendedGame_{\CP}$ as we only have to permute the components of the weights in order to obtain $\extendedGame_{\PC}$. Given a vertex $v \in V_1$, we make the $\playerOne$'s choices explicit from $v$. Such a choice corresponds to a subset of successors of the vertex $v$ as illustrated in Figure~\ref{fig:transfoV1} in the particular case where $v$ has three successors $x$, $y$ and $z$. As previously, we follow the convention that a rounded vertex is a vertex of $\playerOne$, a rectangular vertex is a vertex of $\playerTwo$ and a diamond vertex is either a vertex of $\playerOne$ or a vertex of $\playerTwo$.

For example, on one hand, if $\playerOne$ choices to select only the vertex $y$ (and therefore blocks vertices $x$ and $z$), the penalty of this choice corresponds to the sum of the penalties of edges $(v,x)$ and $(v,z)$, \IE $p_1+p_3$. This corresponds to the second component of the label of the edge $(v,(v,\{y\}))$. Then $\playerTwo$ has no choice and moves to vertex $y$ with a weight of $(\edgeCost(v,y),0)$.

On the other hands, if $\playerOne$ choices to select two vertices, vertices $x$ and $y$ for example, the corresponding penalty is only $p_3$ but  $\playerTwo$ has to resolve the non-determinism caused by $\playerOne$'s choice by either going to $x$ with weight $(\edgeCost(v,x),0)$ or to $y$ with weight $(\edgeCost(v,y),0)$.

\begin{figure}[ht!]
\centering
\scalebox{0.9}{\begin{tikzpicture}

    \node[draw,circle, inner sep=3pt] (v) at (0,0){$v$};
    \node[draw,diamond,inner sep=3pt] (x) at (2.5,1){$x$};
    \node[draw,diamond, inner sep=3pt] (y) at (2.5,0){$y$};
    \node[draw,diamond, inner sep=3pt] (z) at (2.5,-1){$z$};
    \draw[->,thick] (v) to[bend left] node[fill=white, inner sep=0pt] {$c_1\mid p_1$} (x);
    \draw[->,thick] (v) to node[fill=white, inner sep=0pt]{$c_2 \mid p_2$} (y);
    \draw[->,thick] (v) to[bend right] node[fill=white, inner sep=0pt] {$c_3 \mid p_3$} (z);

    \draw[->, ultra thick] (3.5,0) to (5.5,0);

    \node[draw,circle, inner sep=3pt] (vbis) at (6.5,0){$v$};
    \node[draw,inner sep=3pt] (vx) at (10.5,3){$v,\{x\}$};
    \node[draw,inner sep=3pt] (vxy) at (10.5,2){$v,\{x,y\}$};
    \node[draw,inner sep=3pt] (vxz) at (10.5,1){$v, \{x,z\}$};
    \node[draw,inner sep=3pt] (vxyz) at (10.5,0){$v,\{x,y,z\}$};
    \node[draw,inner sep=3pt] (vy) at (10.5,-1){$v,\{y\}$};
    \node[draw,inner sep=3pt] (vyz) at (10.5,-2){$v,\{y,z\}$};
    \node[draw,inner sep=3pt] (vz) at (10.5,-3){$v,\{z\}$};

    \node[draw,diamond,inner sep=3pt] (xbis) at (14.5,2){$x$};
    \node[draw,diamond, inner sep=3pt] (ybis) at (14.5,0){$y$};
    \node[draw,diamond, inner sep=3pt] (zbis) at (14.5,-2){$z$};

    \draw[->,thick] (vbis) to[bend left=30] node[fill=white, inner sep=0pt]{\scriptsize$(0,p_2+p_3)$} (vx);
    \draw[->,thick] (vbis) to[bend left=20] node[fill=white, inner sep=0pt]{\scriptsize$(0,p_3)$} (vxy);
    \draw[->,thick] (vbis) to[bend left=10] node[fill=white, inner sep=0pt]{\scriptsize$(0,p_2)$} (vxz);
    \draw[->,thick] (vbis) to node[fill=white, inner sep=0pt]{\scriptsize$(0,0)$} (vxyz);
    \draw[->,thick] (vbis) to[bend right=10] node[fill=white, inner sep=0pt]{\scriptsize$(0,p_1+p_3)$} (vy);
    \draw[->,thick] (vbis) to[bend right=20] node[fill=white, inner sep=0pt]{\scriptsize$(0,p_1)$} (vyz);
    \draw[->,thick] (vbis) to[bend right=30] node[fill=white, inner sep=0pt]{\scriptsize$(0,p_1+p_2)$} (vz);

    \draw[->,thick] (vx) to[bend left=10] node[fill=white,inner sep=0pt,pos=0.75]{\scriptsize$(c_1,0)$} (xbis);
    \draw[->,thick] (vxy) to node[fill=white,inner sep=0pt,pos=0.75]{\scriptsize$(c_1,0)$} (xbis);
    \draw[->,thick] (vxy) to node[fill=white,inner sep=0pt,pos=0.75]{\scriptsize$(c_2,0)$} (ybis);
    \draw[->,thick] (vxz) to[bend right=10] node[fill=white,inner sep=0pt,pos=0.75]{\scriptsize$(c_1,0)$} (xbis);
    \draw[->,thick] (vxz) to node[fill=white,inner sep=0pt,pos=0.8]{\scriptsize$(c_3,0)$} (zbis);
    \draw[->,thick] (vxyz) to[bend right=20] node[fill=white,inner sep=0pt,pos=0.75]{\scriptsize$(c_1,0)$} (xbis);
    \draw[->,thick] (vxyz) to node[fill=white,inner sep=0pt,pos=0.75]{\scriptsize$(c_2,0)$} (ybis);
    \draw[->,thick] (vxyz) to[bend right=10] node[fill=white,inner sep=0pt,pos=0.8]{\scriptsize$(c_3,0)$} (zbis);
    \draw[->,thick] (vy) to[bend right=10] node[fill=white,inner sep=0pt,pos=0.75]{\scriptsize$(c_2,0)$} (ybis);
    \draw[->,thick] (vyz) to[bend right=20] node[fill=white,inner sep=0pt,pos=0.75]{\scriptsize$(c_2,0)$} (ybis);
    \draw[->,thick] (vyz) to node[fill=white,inner sep=0pt,pos=0.75]{\scriptsize$(c_3,0)$} (zbis);
    \draw[->,thick] (vz) to[bend right=10] node[fill=white,inner sep=0pt,pos=0.75]{\scriptsize$(c_3,0)$} (zbis);

    \end{tikzpicture}
    }
    \caption{Transformation of a vertex of $\playerOne$ in a quantitative reachability game into its corresponding vertex in the extended game $\extendedGame_{\CP}$.}
    \label{fig:transfoV1}
    \end{figure}

Given $v \in V_2$, no transformation is needed in the associated extended game. Notice that since we consider the penalty of edges blocked by $\playerOne$ we can assume that the penalty of an edge whose source is owned by $\playerTwo$ is equal to $0$. This is illustrated in Figure~\ref{fig:transfoV2} in the particular case where $v$ has three successors $x$, $y$ and $z$.

Finally, given $v\in V$, $v$ is part of the target set in $\game$, if and only if, its corresponding vertex is part of the target set in $\extendedGame_\kappa$.\\

\begin{figure}[ht!]
\centering
\scalebox{0.9}{\begin{tikzpicture}

    \node[draw, inner sep=3pt] (v) at (0,0){$v$};
    \node[draw,diamond,inner sep=3pt] (x) at (2.5,1){$x$};
    \node[draw,diamond, inner sep=3pt] (y) at (2.5,0){$y$};
    \node[draw,diamond, inner sep=3pt] (z) at (2.5,-1){$z$};
    \draw[->,thick] (v) to[bend left] node[fill=white, inner sep=0pt] {$c_1\mid 0$} (x);
    \draw[->,thick] (v) to node[fill=white, inner sep=0pt]{$c_2 \mid 0$} (y);
    \draw[->,thick] (v) to[bend right] node[fill=white, inner sep=0pt] {$c_3 \mid 0$} (z);

    \draw[->, ultra thick] (3.5,0) to (5.5,0);

    \node[draw, inner sep=3pt] (vbis) at (6.5,0){$v$};
    \node[draw,diamond,inner sep=3pt] (xbis) at (9,1){$x$};
    \node[draw,diamond, inner sep=3pt] (ybis) at (9,0){$y$};
    \node[draw,diamond, inner sep=3pt] (zbis) at (9,-1){$z$};
    \draw[->,thick] (vbis) to[bend left] node[fill=white, inner sep=0pt] {$(c_1, 0)$} (xbis);
    \draw[->,thick] (vbis) to node[fill=white, inner sep=0pt]{$(c_2, 0)$} (ybis);
    \draw[->,thick] (vbis) to[bend right] node[fill=white, inner sep=0pt] {$(c_3, 0)$} (zbis);
    
    \end{tikzpicture}
    }
    \caption{Transformation of a vertex of $\playerTwo$ in a quantitative reachability game into its corresponding vertex in the extended game $\extendedGame_{\CP}$.}
    \label{fig:transfoV2}
    \end{figure}

An interesting property of the extended game $\extendedGame_\kappa$ of a quantitative reachability game $\game$, is that, given $(c,p) \in \N \times \N$, there exists a multi-strategy $\multiStrat_1$ of $\playerOne$ from $v$ in $\game$ such that its worst ensured cost is equal to $c$ and its penalty is equal to $p$, if and only if, there exists a strategy $\sigma_1$ of $\playerOne$ from $v$ in $\extendedGame_{\CP}$ (resp. in $\extendedGame_{\PC}$) which ensures $(c,p)$ (resp. $(p,c)$).


\begin{proposition}
\label{prop:correspondanceExtended}
    Let $\game$ be a quantitative reachability game and $\extendedGame_{\kappa}$ be its associated extended game, let $v \in V$ and let $(c,p) \in \N \times \N$,

    \begin{center}there exists $\multiStrat_1 \in \msSet{1}{v}$ such that $\kappa(\multiOutcome{\multiStrat_1}{v}) \lesssim \begin{cases} (c,p) & \text{ if } \kappa = \CP \\ (p,c) & \text{ if } \kappa = \PC\end{cases}$ \\
    if and only if\\
    there exists $\sigma_1 \in \stratSet{1}{v}$ such that for all $\sigma_2 \in \stratSet{2}{v}$, $\dCost^X(\outcome{\sigma_1}{\sigma_2}{v}) \lesssim \begin{cases} (c,p) & \text{ if } \kappa = \CP \\ (p,c) & \text{ if } \kappa = \PC\end{cases}$,
    \end{center}
    where $\lesssim$ is either $\leqC$ or $\leqL$.
\end{proposition}

\begin{example} The initialized associated extended game $(\extendedGame_{\CP},v_0)$ of the quantitative reachability game $(\game,v_0)$ of Example~\ref{ex:runningExPermi} is provided in Figure~\ref{fig:extendedGame}. The strategy $\sigma_1$ of $\playerOne$ from $v_0$ defined as $\sigma_1(hv_0) = (v_0,\{v_1\})$, $\sigma_1(hv_1) = (v_1,\{v_2,v_3 \})$, $\sigma_1(hv_2) =  (v_2,\{v_4\})$, $\sigma_1(hv_3) = (v_3, \{v_4\})$ and $\sigma_1(hv_6) = (v_6, \{v_5\})$, depicted by the blue edges, corresponds to the multi-strategy $\multiStrat_1$ of Figure~\ref{fig:quantitativeReachGameMultiStrat}. Indeed, we have that for all $\sigma_2 \in \stratSet{2}{v_0}$, $\dCost^X(\outcome{\sigma_1}{\sigma_2}{v_0}) \leqC (6,12)$.

\begin{figure}[ht!]
    \centering
    \scalebox{1}{\begin{tikzpicture}[scale=0.9]

    \node[draw,circle, inner sep=3pt] (v0) at (0,0){$v_0$};
    \node[draw,inner sep=3pt] (v0v1) at (-2,-1.5){$v_0, \{v_1\}$};
    \node[draw,inner sep=3pt] (v0v1v8) at (0,-1.5){$v_0, \{v_1,v_8\}$};
    \node[draw,inner sep=3pt] (v0v8) at (2,-1.5){$v_0, \{v_8\}$};

    \draw[->,ultra thick, blue] (v0) to[bend right] node[fill=white, inner sep=0pt]{\scriptsize \color{black} $(0,2)$} (v0v1);
    \draw[->,thick] (v0) to node[fill=white, inner sep=0pt]{\scriptsize$(0,0)$} (v0v1v8);

    \draw[->,thick] (v0) to[bend left] node[fill=white, inner sep=0pt]{\scriptsize$(0,1)$} (v0v8);

    \node[draw,circle, inner sep=3pt] (v1) at (-2,-3){$v_1$};

    \node[draw, inner sep=3pt] (v8) at (2,-3){$v_8$};

    \draw[->,thick] (v0v1) to node[fill=white, inner sep=0pt]{\scriptsize$(1,0)$} (v1);
    \draw[->,thick] (v0v1v8) to node[fill=white, inner sep=0pt]{\scriptsize$(1,0)$} (v1);
    \draw[->,thick] (v0v1v8) to node[fill=white, inner sep=0pt]{\scriptsize$(1,0)$} (v1);
    \draw[->,thick] (v0v1v8) to node[fill=white, inner sep=0pt]{\scriptsize$(1,0)$} (v8);

    \draw[->,thick] (v0v8) to node[fill=white, inner sep=0pt]{\scriptsize$(1,0)$} (v8);
    \draw[->,thick] (v8) to[loop right] node[fill=white, inner sep=0pt]{\scriptsize$(1,0)$} (v8);
    \draw[->,thick] (v8) to node[fill=white, inner sep=0pt]{\scriptsize$(1,0)$} (v1);

    \node[draw, inner sep=3pt] (v1v2) at (-4,-4.5){$v_1, \{v_2\}$};
    \node[draw, inner sep=3pt] (v1v2v3) at (-2,-4.5){$v_1, \{v_2,v_3\}$};
    \node[draw, inner sep=3pt] (v1v3) at (0,-4.5){$v_1, \{v_3\}$};

    \draw[->,thick] (v1) to[bend right] node[fill=white, inner sep=0pt]{\scriptsize$(0,1)$} (v1v2);
    \draw[->,ultra thick, blue] (v1) to node[fill=white, inner sep=0pt]{\scriptsize \color{black} $(0,0)$} (v1v2v3);
     \draw[->,thick] (v1) to[bend left] node[fill=white, inner sep=0pt]{\scriptsize$(0,1)$} (v1v3);

     \node[draw,circle,inner sep = 3pt] (v2) at (-4,-6){$v_2$};
     \node[draw,circle, inner sep=3pt] (v3) at (0,-6){$v_3$};

    \draw[->,thick] (v1v2) to node[fill=white, inner sep=0pt]{\scriptsize$(2,0)$} (v2);
    \draw[->,thick] (v1v2v3)  to node[fill=white, inner sep=0pt]{\scriptsize$(2,0)$} (v2);
     \draw[->,thick] (v1v2v3)  to node[fill=white, inner sep=0pt]{\scriptsize$(1,0)$} (v3);
     \draw[->,thick] (v1v3)  to node[fill=white, inner sep=0pt]{\scriptsize$(1,0)$} (v3);

     \node[draw, inner sep=3pt] (v2v4) at (-4,-7.5){$v_2,\{v_4\}$};
     \node[draw,inner sep=3pt] (v3v4) at (0,-7.5){$v_3,\{v_4\}$};

    \draw[->,ultra thick,blue] (v2) to node[fill=white, inner sep=0pt]{\scriptsize \color{black}$(0,0)$} (v2v4);
    \draw[->,ultra thick,blue] (v3) to node[fill=white, inner sep=0pt]{\scriptsize \color{black}$(0,0)$} (v3v4);

    \node[draw,inner sep=3pt] (v4) at (-2,-9){$v_4$};
    \draw[->,thick] (v2v4) to node[fill=white, inner sep=0pt]{\scriptsize$(1,0)$} (v4);
    \draw[->,thick] (v3v4) to node[fill=white, inner sep=0pt]{\scriptsize$(2,0)$} (v4);

    \node[draw,circle, inner sep=3pt] (v6) at (6,-9){$v_6$};
    \node[draw,inner sep=3pt] (v6v7) at (8,-7.5){$v_6,\{v_7\}$};
    \node[draw,inner sep=3pt] (v6v5v7) at (6,-7.5){$v_6,\{v_5,v_7\}$};
    \node[draw,inner sep=3pt] (v6v5) at (4,-7.5){$v_6,\{v_5\}$};

    \draw[->,thick] (v4) to node[fill=white,inner sep=0pt]{\scriptsize$(1,0)$} (v6);
    \draw[->,thick] (v6) to node[fill=white,inner sep=0pt]{\scriptsize$(0,0)$} (v6v5v7);
    \draw[->,thick] (v6) to[bend right] node[fill=white,inner sep=0pt]{\scriptsize$(0,1)$} (v6v7);
    \draw[->,ultra thick, blue] (v6) to[bend left] node[fill=white,inner sep=0pt]{\scriptsize \color{black} $(0,10)$} (v6v5);

    \node[draw,inner sep=3pt, accepting] (v5) at (4,-6){$v_5$};
    \node[draw,inner sep=3pt] (v7) at (8,-6){$v_7$};

    \draw[->,thick] (v4) to[bend right=15] node[fill=white,inner sep=0pt]{\scriptsize$(1,0)$} (v5);
     \draw[->,thick] (v6v5) to node[fill=white,inner sep=0pt]{\scriptsize$(1,0)$} (v5);
      \draw[->,thick] (v6v5v7) to node[fill=white,inner sep=0pt]{\scriptsize$(1,0)$} (v5);

       \draw[->,thick] (v6v5v7) to node[fill=white,inner sep=0pt]{\scriptsize$(10,0)$} (v7);
        \draw[->,thick] (v6v7) to node[fill=white,inner sep=0pt]{\scriptsize$(10,0)$} (v7);

     \draw[->,thick] (v7) to node[fill=white,inner sep=0pt]{\scriptsize$(1,0)$} (v5);

      \draw[->,thick] (v5) to[loop above] node[fill=white,inner sep=0pt]{\scriptsize$(1,0)$} (v5);

    \end{tikzpicture}
    }
    \caption{Associated extended game $\extendedGame_{\CP}$ of the quantitative reachability game illustrated in Figure~\ref{fig:quantitativeReachGameMultiStrat}.}
    \label{fig:extendedGame}
\end{figure}
    
\end{example}

\begin{remark}
\label{rem:polynomialNumberOfStates}
    Notice that even if the construction of the associated extended game of a quantitative reachability game $\game$ leads to an exponential blow-up of the size of the game, the number of vertices owned by $\playerOne$ remains unchanged, and so polynomial.
\end{remark}


Let $\game$ be a quantitative reachability game   and  $\extendedGame_{\kappa}$ be its associated extended game.  
Given a pair $(x,y) \in \N \times \N$, by (a naive application of) Proposition~\ref{prop:noCycle}, we know that  if there exists a strategy  
$\sigma_1$ of $\playerOne$ in $\extendedGame_{\kappa}$  which ensures $(x,y)$, then there exists another strategy $\sigma'_1$ that also ensures $(x,y)$ and such that for all plays consistent with it, the target set is reached within $|V^X|$ steps, where $V^X$ is the set of vertices in $\extendedGame_{\kappa}$ and so $|V^X|$ is exponential in $|V|$. 

However, a closer look at  the construction of the extended game allows to show that  the target set is in fact reached within at most $2\cdot|V|$  steps. Indeed,
let us recall that Proposition~\ref{prop:noCycle} is obtained by removing adequately cycles along plays consistent with $\sigma_1$. Moreover, given a play in the extended game, it is made up of a succession of vertices from $\game$ and vertices specific to $\extendedGame_{\kappa}$. Notice there is always at most one such later kind of vertex between two vertices from $\game$.
Finally, if there is a  cycle in a play of the extended game, there exists a cycle, in the same play, between two vertices that were initially in $\game$.
For these reasons, the target set is reached within at most $2\cdot|V|$  steps along all consistent plays with $\sigma'_1$. 
Those observations lead to the following result.

\begin{corollary}[of Proposition~\ref{prop:noCycle}]
\label{cor:noCycleExtended}
    Let $\game$ be a quantitative reachability game, $v \in V$ be a vertex  and  $(x,y) \in \N \times \N$. In the associated extended game of $\game$:  if there exists a strategy $\sigma_1$ of $\playerOne$ such that for all strategies $\sigma_2$ of $\playerTwo$, $ \dCost^X(\outcome{\sigma_1}{\sigma_2}{v}) \lesssim (x,y)$, then there exists a strategy $\sigma'_1$ of $\playerOne$ such that for all strategies $\sigma_2$ of $\playerTwo$  we have \emph{(i)} $\dCost^X(\outcome{\sigma_1}{\sigma_2}{v}) \lesssim (x,y)$ and \emph{(ii)}  $\lengthF{\outcome{\sigma_1}{\sigma_2}{v}} \leq 2\cdot|V|$.
\end{corollary}


\subsection{Constrained Existence of Multi-strategies}
\label{section:constrainedExistencePermi}
This section is devoted to the proofs of complexity results related to the constrained existence of multi-strategies problems. We first prove that the MCE1 problem is $\PSPACE$-complete (Proposition~\ref{prop:constrainedComponentPermi}) and then that the MCE2 and MCE3 problems belong to $\NP$ (Proposition~\ref{prop:constrainedLexicoPermi}).

\begin{proposition}
\label{prop:constrainedComponentPermi}
    The MCE1 problem is $\PSPACE$-complete.
\end{proposition}

\textbf{$\PSPACE$-easiness} Thanks to Corollary~\ref{cor:noCycleExtended} we are able to prove that the MCE1 problem belongs to $\APTime$ and since $\APTime = \PSPACE$ we get the result. The alternating Turing machine works as follows. The states of the alternating Turing machine are split between existential states and universal states. Existential states correspond to $\playerOne$'s states in the quantitative reachability game. From these states a subset of successors is non-deterministically guessed. Universal states have two roles: either they resolve the non-determinism caused by a choice of an existential state or they correspond to a vertex of $\playerTwo$ in the quantitative reachability game. Three  polynomial counters are used in order to keep track information about the execution of the algorithm: \emph{(i)} a first counter, upper-bounded by $2\cdot |V|$,  keeps track the number of states visited along the execution of the algorithm; \emph{(ii)} a second counter, upper-bounded by $x$, accumulates the costs along the execution; and \emph{(iii)} a third counter, upper-bounded by $y$, accumulates the penalties (caused by choices of existential states) along the execution. A path of lenght $2 \cdot |V|$ is accepted, if and only if, \emph{(i)} the target set is reached along that path and \emph{(ii)} the values of the counters are less than or equal to their upper bound when the target set is reached. \\

\textbf{$\PSPACE$-hardness} The hardness part of Proposition~\ref{prop:constrainedComponentPermi} is due to a polynomial reduction from the Constrained Existence problem in multi-weighted reachability games (see Definition~\ref{prob-constProb}) which is proved $\PSPACE$-complete (Theorem~\ref{thm:constrainedResComponent})\footnote{Notice that this result holds even when $d=2$.}. Let $(\game_2,v_0)$ be an initialized $2$-weighted reachability game. Let us sketch out the construction of the corresponding quantitative reachability game $(\game',v_0)$ equipped with a penalty function.  This construction is inspired by a similar one provided in~\cite{BDMR09}.

Each vertex in $\game$ becomes a vertex in $\game'$ and is owned by the same player. Moreover, the target set is the same in both games.
Finally, an edge $(v,v')$ labeled by $(c_1,c_2)$ in $\game$ is transformed into a gadget in $\game'$, as shown in Figure~\ref{fig:hardnessPermi} where a diamond represents either a $\playerOne$'s vertex or a $\playerTwo$'s vertex.

\begin{figure}[ht!]
\centering
\scalebox{0.9}{\begin{tikzpicture}

    \node[draw,diamond, inner sep=4.5pt] (v) at (0,0){$v$};
    \node[draw,diamond, inner sep=3pt] (y) at (2.5,0){$v'$};
    
    \draw[->,thick] (v) to node[fill=white, inner sep=0pt]{$(c_1,c_2)$} (y);
    
    \draw[->, ultra thick] (3.5,0) to (5.5,0);

    \node[draw,diamond, inner sep=4.5pt] (vbis) at (6.5,0){$v$};

    \node[draw,circle, inner sep=7pt] (int) at (9,0){};

    \node[draw, diamond, inner sep=3pt] (v'bis) at (11.5,0){$v'$};

    \node[draw,circle, inner sep=3pt] (perp) at (9,2){$\perp$};
    
    \draw[->,thick] (vbis) to node[fill=white, inner sep=1pt]{$c_1 \mid 0$} (int);
    \draw[->,thick] (int) to node[fill=white, inner sep=1pt]{$0 \mid 0$} (v'bis);
    \draw[->,thick] (int) to node[fill=white, inner sep=1pt]{$0 \mid c_2$} (perp);

    \draw[->,thick] (perp) to [loop right] node[fill=white,inner sep=1pt]{$0 \mid 0$}(perp);
    \end{tikzpicture}
    }
    \caption{Transformation of an edge in  a $2$-weighted reachability game to a corresponding gadget in a quantitative reachability game.}
    \label{fig:hardnessPermi}
    \end{figure}

In view of this reduction, given $(x,y) \in \N \times \N$, there exists a (simple) strategy $\sigma_1$ of $\playerOne$ in $(\game,v_0)$ such that for all strategies $\sigma_2$ of $\playerTwo$, $\dCost(\outcome{\sigma_1}{\sigma_2}{v_0}) \leqC (x,y)$ if and only if there exists a multi-strategy $\multiStrat_1$ of $\playerOne$ in $(\game',v_0)$ such that $\CP(\multiOutcome{\multiStrat_1}{v_0}) \leqC (x,y)$.

\begin{proposition}
\label{prop:constrainedLexicoPermi}
    The MCE2 et MCE3 problems belong to $\NP$.
\end{proposition}

\begin{proof}Let $(\game,v_0)$ be an initialized quantitative reachability game and let $(x,y) \in \N \times \N$. Thanks to Proposition~\ref{prop:correspondanceExtended}, we know that deciding the MCE2 problem amounts to finding a strategy $\sigma_1$ of $\playerOne$ from $v_0$ in the initialized extended game $(\extendedGame_{\CP},v_0)$ associated with $(\game,v_0)$ such that for all strategies $\sigma_2$ of $\playerTwo$, we have that $\dCost^X(\outcome{\sigma_1}{\sigma_2}{v_0}) \leqL (x,y)$. Moreover, thanks to Proposition~\ref{prop:stratOptiLexicoPositional} it is sufficient to guess a memoryless strategy $\sigma'_1$. In view of Remark~\ref{rem:polynomialNumberOfStates}, it amounts to guessing a subset of successors for each vertex $v \in V_1$. There is a polynomial number of such vertices since they are vertices in $(\game,v_0)$. Once $\sigma'_1$ is fixed, we can restrict $(\extendedGame_{\CP},v_0)$ to a multi-weighted reachability game in which each vertex owned by $\playerOne$ has only one successor. Finally, we compute $\upValue{v_0}$ in this restricted game in polynomial time (by Theorem~\ref{thm:complexityLexico}) as this restricted game has a polynomial size. We conclude: $\sigma'_1$ ensures $(x,y)$ if and only if $\upValue{v_0} \leqL (x,y)$.

The MCE3 problem can be decided exactly in the way by considering the initialized extended game $(\extendedGame_{\PC},v_0)$ instead of $(\extendedGame_{\CP},v_0)$
\end{proof}


\subsection{Ensured Values by Multi-strategies}
\label{section:ensuredValuesPermi}

In this section we explain why, given a quantitative reachability game $(\game,v_0)$,  the set $\msPareto(v_0)$ can be computed in exponential time (Proposition~\ref{prop:paretoFrontierPermi}) while the values $\cVal(v_0)$ and $\pVal(v_0)$ can be computed thanks to an algorithm that makes a polynomial number of calls to a decision problem in $\NP$ (Proposition~\ref{prop:valuesPermi}). 
For all these results we exploit the fact that computing $\msPareto(v_0)$ (resp. $\cVal(v_0)$ and $\pVal(v_0)$) returns the same set of pairs (resp. the same pair) if the computation is performed in $(\game,v_0)$ or in its associated extended game $(\extendedGame_{\kappa},v_0)$ (by Proposition~\ref{prop:correspondanceExtended}).

\begin{proposition}
\label{prop:paretoFrontierPermi}
Given an initialized quantitative reachability game $(\game,v_0)$, the set $\msPareto(v_0)$ can be computed in exponential time.
\end{proposition}

This result is a direct consequence of Theorem~\ref{thm:complexityComponent} applied to the initialized associated extended game $(\extendedGame_{\CP},v_0)$ of $(\game,v_0)$.

\begin{proposition}
\label{prop:valuesPermi}
Given an initialized quantitative reachability game $(\game,v_0)$, the values $\cVal(v_0)$ and $\pVal(v_0)$ can be computed thanks to an algorithm that makes a polynomial number of calls to a decision problem in $\NP$.
\end{proposition}
We first state an intermediate result which holds thanks to Corollary~\ref{cor:noCycleExtended}. 

\begin{lemma}
\label{lemma:upperBoundLexico}
     Let $(\game,v_0)$ be an initialized quantitative reachability game and $(\extendedGame_{\kappa},v_0)$  be its associated initialized extended game, let $(x,y) \in \N \times \N$, if there exists a strategy $\sigma_1$ of $\playerOne$ such that for all strategies $\sigma_2$ of $\playerTwo$, we have that $\dCost^X(\outcome{\sigma_1}{\sigma_2}{v_0}) \leqL (x,y)$, then there exists a strategy $\sigma'_1$ of $\playerOne$ such that for all strategies $\sigma_2$ of $\playerTwo$ we have that $$\dCost^X(\outcome{\sigma_1}{\sigma_2}{v_0}) \leqL \begin{cases} \min_{\leqL} \{ (x,y), (b_1, b_2)  \} &\text{ if } \kappa = \CP \\ \min_{\leqL} \{ (x,y), (b_2, b_1)  \} & \text{ if } \kappa = \PC \end{cases} $$ where $b_1=2\cdot |V| \cdot \max_{e \in E} \edgeCost(e)$ and $b_2=2 \cdot |V| \cdot  \max_{v \in V} \sum_{v' \in \successor(v)} \edgePenal(v,v')$.
\end{lemma}

Let us now move on the proof of Proposition~\ref{prop:valuesPermi}.

\begin{proof}[Proof of Proposition~\ref{prop:valuesPermi}]
    Let $(\game,v_0)$ be an initialized quantitative reachability game and  $(\extendedGame_{\CP},v_0)$ be its initialized extended game.
    
    Let $b_1=2\cdot |V| \cdot \max_{e \in E} \edgeCost(e)$ and $b_2=2 \cdot |V| \cdot  \max_{v \in V} \sum_{v' \in \successor(v)} \edgePenal(v,v')$. Let us compute $\cVal(v_0)$. Due to Lemma~\ref{lemma:upperBoundLexico} the first step amounts to finding the least $x^*$ such that the MCE2 problem is true with $(x,y) = (x^*,b_2)$. This value can be found thanks to a binary search between $0$ and $b_1$  by deciding a polynomial number of times the MCE2 problem.  Once that value $x^*$ is  found, we repeat the procedure in order to find the least $y^*$ such that the MCE2 problem is true with $(x,y) = (x^*,y^*)$. As previously this value is obtained that to a binary search between $0$ and $b_2$. This means deciding the MCE2 problem a polynomial number of times. At the end of this procedure, we obtain  $\cVal(v_0) = (x^*,y^*)$. 

    The procedure to compute $\pVal(v_0)$ is similar except that we consider the initialized extended game $(\extendedGame_{\PC},v_0)$, we decide a polynomial number of times the MCE3 problem and once $(x^*,y^*)$ are obtained we have that $\pVal(v_0) = (y^*,x^*)$. 
    
\end{proof}

\nocite{*}
\bibliographystyle{fundam}
\bibliography{biblio}

\begin{thebibliography}{10}
\providecommand{\url}[1]{\texttt{#1}}
\providecommand{\urlprefix}{URL }
\expandafter\ifx\csname urlstyle\endcsname\relax
  \providecommand{\doi}[1]{doi:\discretionary{}{}{}#1}\else
  \providecommand{\doi}{doi:\discretionary{}{}{}\begingroup \urlstyle{rm}\Url}\fi
\providecommand{\eprint}[2][]{\url{#2}}

\bibitem{BrihayeG23}
Brihaye T, Goeminne A.
\newblock Multi-weighted Reachability Games.
\newblock In: Bournez O, Formenti E, Potapov I (eds.), Reachability Problems - 17th International Conference, {RP} 2023, Nice, France, October 11-13, 2023, Proceedings, volume 14235 of \emph{Lecture Notes in Computer Science}. Springer, 2023 pp. 85--97.
\newblock \doi{10.1007/978-3-031-45286-4\_7}.
\newblock \urlprefix\url{https://doi.org/10.1007/978-3-031-45286-4\_7}.

\bibitem{LaroussinieMO06}
Laroussinie F, Markey N, Oreiby G.
\newblock Model-Checking Timed {ATL} for Durational Concurrent Game Structures.
\newblock In: Asarin E, Bouyer P (eds.), Formal Modeling and Analysis of Timed Systems, 4th International Conference, {FORMATS} 2006, Paris, France, September 25-27, 2006, Proceedings, volume 4202 of \emph{Lecture Notes in Computer Science}. Springer, 2006 pp. 245--259.
\newblock \doi{10.1007/11867340\_18}.
\newblock \urlprefix\url{https://doi.org/10.1007/11867340\_18}.

\bibitem{BDMR09}
Bouyer P, Duflot M, Markey N, Renault G.
\newblock Measuring Permissivity in Finite Games.
\newblock In: {CONCUR} 2009, volume 5710 of \emph{LNCS}. Springer, 2009 pp. 196--210.
\newblock \doi{10.1007/978-3-642-04081-8\_14}.
\newblock \urlprefix\url{https://doi.org/10.1007/978-3-642-04081-8\_14}.

\bibitem{PuriT02}
Puri A, Tripakis S.
\newblock Algorithms for the Multi-constrained Routing Problem.
\newblock In: Penttonen M, Schmidt EM (eds.), Algorithm Theory - {SWAT} 2002, 8th Scandinavian Workshop on Algorithm Theory, Turku, Finland, July 3-5, 2002 Proceedings, volume 2368 of \emph{Lecture Notes in Computer Science}. Springer, 2002 pp. 338--347.
\newblock \doi{10.1007/3-540-45471-3\_35}.
\newblock \urlprefix\url{https://doi.org/10.1007/3-540-45471-3\_35}.

\bibitem{LarsenR05}
Larsen KG, Rasmussen JI.
\newblock Optimal Conditional Reachability for Multi-priced Timed Automata.
\newblock In: Sassone V (ed.), Foundations of Software Science and Computational Structures, 8th International Conference, {FOSSACS} 2005, Held as Part of the Joint European Conferences on Theory and Practice of Software, {ETAPS} 2005, Edinburgh, UK, April 4-8, 2005, Proceedings, volume 3441 of \emph{Lecture Notes in Computer Science}. Springer, 2005 pp. 234--249.
\newblock \doi{10.1007/978-3-540-31982-5\_15}.
\newblock \urlprefix\url{https://doi.org/10.1007/978-3-540-31982-5\_15}.

\bibitem{FijalkowH13}
Fijalkow N, Horn F.
\newblock Les jeux d'accessibilit{\'{e}} g{\'{e}}n{\'{e}}ralis{\'{e}}e.
\newblock \emph{Tech. Sci. Informatiques}, 2013.
\newblock \textbf{32}(9-10):931--949.
\newblock \doi{10.3166/tsi.32.931-949}.
\newblock \urlprefix\url{https://doi.org/10.3166/tsi.32.931-949}.

\bibitem{BrihayeGMR23}
Brihaye T, Goeminne A, Main JCA, Randour M.
\newblock {Reachability Games and Friends: {A} Journey Through the Lens of Memory and Complexity ({I}nvited Talk)}.
\newblock In: Bouyer P, Srinivasan S (eds.), {FSTTCS} 2023, volume 284 of \emph{LIPIcs}. Schloss Dagstuhl - Leibniz-Zentrum f{\"{u}}r Informatik, 2023 pp. 1:1--1:26.
\newblock \doi{10.4230/LIPICS.FSTTCS.2023.1}.
\newblock \urlprefix\url{https://doi.org/10.4230/LIPIcs.FSTTCS.2023.1}.

\bibitem{JuhlLR13}
Juhl L, Larsen KG, Raskin J.
\newblock Optimal Bounds for Multiweighted and Parametrised Energy Games.
\newblock In: Liu Z, Woodcock J, Zhu H (eds.), Theories of Programming and Formal Methods - Essays Dedicated to Jifeng He on the Occasion of His 70th Birthday, volume 8051 of \emph{Lecture Notes in Computer Science}. Springer, 2013 pp. 244--255.
\newblock \doi{10.1007/978-3-642-39698-4\_15}.
\newblock \urlprefix\url{https://doi.org/10.1007/978-3-642-39698-4\_15}.

\bibitem{ChatterjeeDHR10}
Chatterjee K, Doyen L, Henzinger TA, Raskin J.
\newblock Generalized Mean-payoff and Energy Games.
\newblock In: Lodaya K, Mahajan M (eds.), {IARCS} Annual Conference on Foundations of Software Technology and Theoretical Computer Science, {FSTTCS} 2010, December 15-18, 2010, Chennai, India, volume~8 of \emph{LIPIcs}. Schloss Dagstuhl - Leibniz-Zentrum f{\"{u}}r Informatik, 2010 pp. 505--516.
\newblock \doi{10.4230/LIPIcs.FSTTCS.2010.505}.
\newblock \urlprefix\url{https://doi.org/10.4230/LIPIcs.FSTTCS.2010.505}.

\bibitem{ChatterjeeRR12}
Chatterjee K, Randour M, Raskin J.
\newblock Strategy Synthesis for Multi-Dimensional Quantitative Objectives.
\newblock In: Koutny M, Ulidowski I (eds.), {CONCUR} 2012 - Concurrency Theory - 23rd International Conference, {CONCUR} 2012, Newcastle upon Tyne, UK, September 4-7, 2012. Proceedings, volume 7454 of \emph{Lecture Notes in Computer Science}. Springer, 2012 pp. 115--131.
\newblock \doi{10.1007/978-3-642-32940-1\_10}.
\newblock \urlprefix\url{https://doi.org/10.1007/978-3-642-32940-1\_10}.

\bibitem{BrenguierR15}
Brenguier R, Raskin J.
\newblock Pareto Curves of Multidimensional Mean-Payoff Games.
\newblock In: Kroening D, Pasareanu CS (eds.), Computer Aided Verification - 27th International Conference, {CAV} 2015, San Francisco, CA, USA, July 18-24, 2015, Proceedings, Part {II}, volume 9207 of \emph{Lecture Notes in Computer Science}. Springer, 2015 pp. 251--267.
\newblock \doi{10.1007/978-3-319-21668-3\_15}.
\newblock \urlprefix\url{https://doi.org/10.1007/978-3-319-21668-3\_15}.

\bibitem{ChatterjeeKWW20}
Chatterjee K, Katoen J, Weininger M, Winkler T.
\newblock Stochastic Games with Lexicographic Reachability-Safety Objectives.
\newblock In: Lahiri SK, Wang C (eds.), Computer Aided Verification - 32nd International Conference, {CAV} 2020, Los Angeles, CA, USA, July 21-24, 2020, Proceedings, Part {II}, volume 12225 of \emph{Lecture Notes in Computer Science}. Springer, 2020 pp. 398--420.
\newblock \doi{10.1007/978-3-030-53291-8\_21}.
\newblock \urlprefix\url{https://doi.org/10.1007/978-3-030-53291-8\_21}.

\bibitem{DelzannoR00}
Delzanno G, Raskin J.
\newblock Symbolic Representation of Upward-Closed Sets.
\newblock In: Graf S, Schwartzbach MI (eds.), Tools and Algorithms for Construction and Analysis of Systems, 6th International Conference, {TACAS} 2000, Held as Part of the European Joint Conferences on the Theory and Practice of Software, {ETAPS} 2000, Berlin, Germany, March 25 - April 2, 2000, Proceedings, volume 1785 of \emph{Lecture Notes in Computer Science}. Springer, 2000 pp. 426--440.
\newblock \doi{10.1007/3-540-46419-0\_29}.
\newblock \urlprefix\url{https://doi.org/10.1007/3-540-46419-0\_29}.

\bibitem{Travers06}
Travers SD.
\newblock The complexity of membership problems for circuits over sets of integers.
\newblock \emph{Theor. Comput. Sci.}, 2006.
\newblock \textbf{369}(1-3):211--229.
\newblock \doi{10.1016/j.tcs.2006.08.017}.
\newblock \urlprefix\url{https://doi.org/10.1016/j.tcs.2006.08.017}.

\bibitem{Sipser}
Sipser M.
\newblock Introduction to the Theory of Computation.
\newblock Cengage Learning, 2012.

\bibitem{BJW02}
Bernet J, Janin D, Walukiewicz I.
\newblock Permissive strategies: from parity games to safety games.
\newblock \emph{{RAIRO} Theor. Informatics Appl.}, 2002.
\newblock \textbf{36}(3):261--275.
\newblock \doi{10.1051/ITA:2002013}.

\bibitem{BouyerMOU11}
Bouyer P, Markey N, Olschewski J, Ummels M.
\newblock Measuring Permissiveness in Parity Games: Mean-Payoff Parity Games Revisited.
\newblock In: {ATVA} 2011, volume 6996 of \emph{LNCS}. Springer, 2011 pp. 135--149.
\newblock \doi{10.1007/978-3-642-24372-1\_11}.

\bibitem{GoeminneM25}
Goeminne A, Monmege B.
\newblock {Permissive Equilibria in Multiplayer Reachability Games}.
\newblock In: Endrullis J, Schmitz S (eds.), {CSL} 2025, volume 326 of \emph{LIPIcs}. Schloss Dagstuhl - Leibniz-Zentrum f{\"{u}}r Informatik, 2025 pp. 23:1--23:17.
\newblock \doi{10.4230/LIPICS.CSL.2025.23}.
\newblock \urlprefix\url{https://doi.org/10.4230/LIPIcs.CSL.2025.23}.

\bibitem{AnandNayakSchmuck}
Anand A, Nayak SP, Schmuck AK.
\newblock Synthesizing Permissive Winning Strategy Templates for Parity Games.
\newblock In: {CAV} 2023, volume 13964 of \emph{LNCS}. Springer, 2023 pp. 436--458.
\newblock \doi{10.1007/978-3-031-37706-8_22}.

\bibitem{NayakS24}
Nayak SP, Schmuck A.
\newblock Most General Winning Secure Equilibria Synthesis in Graph Games.
\newblock In: {TACAS} 2024, volume 14572 of \emph{LNCS}. Springer, 2024 pp. 173--193.
\newblock \doi{10.1007/978-3-031-57256-2\_9}.
\newblock \urlprefix\url{https://doi.org/10.1007/978-3-031-57256-2\_9}.

\bibitem{DragerFK0U15}
Dr{\"{a}}ger K, Forejt V, Kwiatkowska MZ, Parker D, Ujma M.
\newblock Permissive Controller Synthesis for Probabilistic Systems.
\newblock \emph{Log. Methods Comput. Sci.}, 2015.
\newblock \textbf{11}(2).
\newblock \doi{10.2168/LMCS-11(2:16)2015}.
\newblock \urlprefix\url{https://doi.org/10.2168/LMCS-11(2:16)2015}.

\bibitem{BouyerFM15}
Bouyer P, Fang E, Markey N.
\newblock Permissive strategies in timed automata and games.
\newblock \emph{Electron. Commun. Eur. Assoc. Softw. Sci. Technol.}, 2015.
\newblock \textbf{72}.
\newblock \doi{10.14279/TUJ.ECEASST.72.1015}.
\newblock \urlprefix\url{https://doi.org/10.14279/tuj.eceasst.72.1015}.

\bibitem{ClementJMM20}
Clement E, J{\'{e}}ron T, Markey N, Mentr{\'{e}} D.
\newblock Computing Maximally-Permissive Strategies in Acyclic Timed Automata.
\newblock In: Bertrand N, Jansen N (eds.), Formal Modeling and Analysis of Timed Systems - 18th International Conference, {FORMATS} 2020, Vienna, Austria, September 1-3, 2020, Proceedings, volume 12288 of \emph{Lecture Notes in Computer Science}. Springer, 2020 pp. 111--126.
\newblock \doi{10.1007/978-3-030-57628-8\_7}.
\newblock \urlprefix\url{https://doi.org/10.1007/978-3-030-57628-8\_7}.

\bibitem{multiweighted}
Brihaye T, Goeminne A.
\newblock Multi-Weighted Reachability Games, 2023.
\newblock \eprint{2308.09625}, \urlprefix\url{https://arxiv.org/abs/2308.09625}.

\bibitem{BruyereHR18}
Bruy{\`{e}}re V, Hautem Q, Raskin J.
\newblock Parameterized complexity of games with monotonically ordered omega-regular objectives.
\newblock In: Schewe S, Zhang L (eds.), 29th International Conference on Concurrency Theory, {CONCUR} 2018, September 4-7, 2018, Beijing, China, volume 118 of \emph{LIPIcs}. Schloss Dagstuhl - Leibniz-Zentrum f{\"{u}}r Informatik, 2018 pp. 29:1--29:16.
\newblock \doi{10.4230/LIPIcs.CONCUR.2018.29}.
\newblock \urlprefix\url{https://doi.org/10.4230/LIPIcs.CONCUR.2018.29}.

\end{thebibliography}


\appendix


\section{Additional Contents of Section~\ref{section:ensuredValues}: \nameref{section:ensuredValues}}
\label{appendix:ensuredValues}

\subsection{Fixpoint Algorithm}
\subsubsection{Termination}
\label{appendix:termination}

\begin{restate}{Proposition}{\ref{prop:noCycle}}
    \restateNoCycle
\end{restate}

The proof of Proposition~\ref{prop:noCycle} relies on the notion of strategy tree that we introduce hereunder.

\textbf{Strategy tree.} Given a game $\dGame$, $\tree$ is a tree rooted at $v$ for some $v \in V$ if \emph{(i)} $\tree$ is a subset of non-empty histories of $\dGame$, \IE $\tree \subseteq \hist(v)$, \emph{(ii)} $v \in \tree$ and \emph{(iii)} if $tu \in \tree$ then, $t \in \tree$. All $t \in \tree$ are called \emph{nodes} of the tree and the particular node $v$ is called the \emph{root} of the tree. As for histories in a game, for all $tu \in \tree$, $\last(tu) = u$. The \emph{depth} of a node $t \in \tree$, written $\depth(t)$, is equal to $|t|$ and its \emph{height}, denoted by $\height(t)$, is given by $ \sup \{ |\last(t)t'| \mid t' \in V^* \text{ and } tt' \in \tree \}$. The height of the tree corresponds to the height of its root. A node $t \in \tree$ is called a \emph{leaf} if $\height(t) =0$. We denote by $\tree_{\restriction t}$, the subtree of $\tree$ rooted at $t$ for some $t \in \tree$, that is the set of non-empty histories such that $t' \in \tree_{\restriction t}$  if and only if $ t'= tw$ for some $w \in V^*$. Finally, a (finite or infinite) \emph{branch} of the tree is a (finite or infinite) sequence of nodes $n_0n_1\ldots$ such that for all $k \in \N$, $(\last(n_k),\last(n_{k+1})) \in E$. The cost of an infinite branch is defined similarly as the cost a play: $\cost_i(n_0n_1\ldots)= \sum_{k=0}^{\ell-1} \weight_i(\last(n_k),\last(n_{k+1}))$ if $\ell$ is the least index such that $\last(n_\ell) \in \targetSet$ and $\cost_i(n_0n_1\ldots)= \infty$ otherwise. This definition may be easily adapted if the branch is finite.

When we fix a strategy $\sigma_1 \in \stratSet{1}{v}$ for some $v \in V$, we can see all the possible outcomes consistent with a strategy of $\playerTwo$ as a tree consistent with $\sigma_1$.
Given $\sigma_1 \in \stratSet{1}{v}$, the \emph{strategy tree $\stratTree{\sigma_1}$ of} $\sigma_1$ is such that: \emph{(i)} the root  of the tree is $v$, \emph{(ii)} for all $t \in \stratTree{\sigma_1}$, if $\last(t) \in \targetSet$ then, for all $t' \in V^+$, $tt' \not \in \stratTree{\sigma_1}$.
Otherwise,  if $\last(t) \in V_1 \backslash \targetSet$ then, $tv' \in \stratTree{\sigma_1}$ with $v' =\sigma_1(t)$. Else if $\last(t) \in V_2 \backslash \targetSet$,  for all $v' \in \successor(\last(t))$ we have $tv' \in \stratTree{\sigma_1}$. 

In the same way, a tree $\tree$ which satisfies the following conditions allows to define a strategy $\sigma_{\tree}$ of $\playerOne$. For all $t \in \tree$,
\begin{itemize}
    \item if $\last(t) = u \in \targetSet$, then there is no $u' \in V$ such that $tu'  \in \tree$;
    \item if $\last(t) = u \in V_2 \backslash \targetSet$ then, for all $u' \in \successor(u)$, $ tu' \in \tree$;
    \item if $\last(t) = u \in V_1 \backslash \targetSet$ then there exists a unique $u' \in \successor(u)$ such that $tu' \in \tree$; and $\sigma_{\tree}(t) = u'$.    
\end{itemize}

Notice that in this way, $\sigma_{\tree} \in \stratSet{1}{v}$ is not well defined on histories which are not consistent with $\sigma_{\tree}$ and some $\sigma_2 \in \stratSet{2}{v}$. This is not a problem for our purpose, we may assume that for such histories $h \in \hist_1(v)$, $\sigma_{\tree}(h) = v'$ for some arbitrary (fixed) $v' \in \successor(\last(h))$. We call this well defined strategy the \emph{strategy associated with} $\tree$.

\begin{example}
\label{ex:runningEx2}
    We illustrate the notion of strategy tree by considering the game described in Example~\ref{ex:runningEx1}. Let us recall that the game arena is given in Figure~\ref{fig:runningEx}.

    We define a strategy $\sigma_1$ of $\playerOne$ from $v_0$ as, for all $hv \in \hist_1(v_0)$:  $\sigma(hv) = v'$ with $v'=v_4$ if $ v \in \{ v_1,v_2,v_3\}$, $v' = v_6$ if  $hv \in \{ v_0v_2v_4, v_0v_3v_4\}$, $ v'= v_7$ if $hv = v_0v_1v_4$, $ v' = v_9$ if $v \in \{ v_6,v_7,v_9\}$,  $v' = v_8$ if  $v= v_5$ and $v' = v_{10}$  if  $v= v_8$. This strategy is a finite-memory strategy since the choice made in $v_4$ depends on the past history: if it crossed vertices $v_2$ or $v_3$ the next vertex is $v_6$ while it is $v_7$ if it crossed $v_1$.  We also define a positional strategy $\sigma_2$ of $\playerTwo$ from $v_0$ as, $\sigma_2(v_0) = v_2$ and $\sigma_2(v_{10}) = v_9$.  The outcome of the strategy profile $(\sigma_1,\sigma_2)$ from $v_0$ is $\outcome{\sigma_1}{\sigma_2}{v_0} = v_0v_2v_4v_6v_9^\omega$ and its cost profile is $\dCost(\outcome{\sigma_1}{\sigma_2}{v_0}) = (8,8)$.
    
    The strategy tree of $\sigma_1$ is \begin{multline*}     
    \stratTree{\sigma_1}  = \{v_0, v_0v_1, v_0v_2, v_0v_3, v_0v_1v_4, v_0v_2v_4, v_0v_3v_4, v_0v_1v_4v_7, \\v_0v_2v_4v_6, v_0v_3v_4v_6, v_0v_1v_4v_7v_9,  v_0v_2v_4v_6v_9, v_0v_3v_4v_6v_9\} \end{multline*} and is drawn in Figure~\ref{fig:stratTreeRunningEx}. The root of this tree is the node $n_0=v_0$ and there are three leaves $v_0v_1v_4v_7v_9,  v_0v_2v_4v_6v_9,$ and $ v_0v_3v_4v_6v_9$. The height of the root, $\height(v_0)$, is equal to $4$ and the depth of the node $n' = v_0v_2v_4$, $\depth(n')$, is equal to $2$.
\end{example}

\begin{figure}[ht]
    \centering
    \begin{tikzpicture}
        \node[draw, inner sep=3pt] (v0) at (0,0){$v_0$};
        \node[draw, rounded corners=6pt] (v0v1) at (1.5,1){$v_0v_1$};
        \node[draw, rounded corners=6pt] (v0v2) at (1.5,0){$v_0v_2$};
        \node[draw, rounded corners=6pt] (v0v3) at (1.5,-1){$v_0v_3$};

        \draw[->,thick] (v0) to[bend left] (v0v1);
        \draw[->,thick] (v0) to (v0v2);
        \draw[->,thick] (v0) to[bend right] (v0v3);

        \node[draw, rounded corners=6pt] (v0v1v4) at (3.5,1){$v_0v_1v_4$};
        \node[draw, rounded corners=6pt] (v0v1v4v7) at (6,1){$v_0v_1v_4v_7$};
        \node[draw, rounded corners=6pt] (v0v1v4v7v9) at (8.5,1){$v_0v_1v_4v_7v_9$};
        \draw[->,thick] (v0v1) to (v0v1v4);
        \draw[->,thick] (v0v1v4) to (v0v1v4v7);
        \draw[->,thick] (v0v1v4v7) to (v0v1v4v7v9);

        \node[draw, rounded corners=6pt] (v0v2v4) at (3.5,0){$v_0v_2v_4$};
        \node[draw, rounded corners=6pt] (v0v2v4v7) at (6,0){$v_0v_2v_4v_6$};
        \node[draw, rounded corners=6pt] (v0v2v4v7v9) at (8.5,0){$v_0v_2v_4v_6v_9$};
        \draw[->,thick] (v0v2) to (v0v2v4);
        \draw[->,thick] (v0v2v4) to (v0v2v4v7);
        \draw[->,thick] (v0v2v4v7) to (v0v2v4v7v9);

        \node[draw, rounded corners=6pt] (v0v3v4) at (3.5,-1){$v_0v_3v_4$};
        \node[draw, rounded corners=6pt] (v0v3v4v7) at (6,-1){$v_0v_3v_4v_6$};
        \node[draw, rounded corners=6pt] (v0v3v4v7v9) at (8.5,-1){$v_0v_3v_4v_6v_9$};
        \draw[->,thick] (v0v3) to (v0v3v4);
        \draw[->,thick] (v0v3v4) to (v0v3v4v7);
        \draw[->,thick] (v0v3v4v7) to (v0v3v4v7v9);

    \end{tikzpicture}
    \caption{Strategy tree $\stratTree{\sigma_1}$ of the strategy $\sigma_1$ defined in Example~\ref{ex:runningEx2}.}
    \label{fig:stratTreeRunningEx}
\end{figure}



\begin{proof}[Proof of Proposition~\ref{prop:noCycle}]
Let $\sigma_1$ be a strategy of $\playerOne$ from $v$ such that for all strategies $\sigma_2$ of $\playerTwo$ the target set $\targetSet$ is reached along $\outcome{\sigma_1}{\sigma_2}{v}$. Let us assume that the strategy tree of $\sigma_1$, $\stratTree{\sigma_1}$, is given by that of Figure~\ref{fig:winningStratTree}. As the target set is reached for all strategies of $\playerTwo$, all branches of the tree are finite. Formally, $\height(v) \neq \infty $ and  all branches $n_0 \ldots n_k$ end in a leaf with $\last(n_k) \in \targetSet$ and $\dCost(n_0\ldots n_k) \lesssim \tup{x}$. Notice that since $\stratTree{\sigma_1}$ is rooted at $v$, $n_0 = v$.\\

We remove cycles in the tree one by one: we begin with $\tree_0 = \stratTree{\sigma_1}$ and at each step of the procedure these two properties are preserved, \emph{(i)} the height of the root is finite  and \emph{(ii)} for all (finite) branches $n_0\ldots n_k$ of the tree, $\dCost(n_0 \ldots n_k) \lesssim \tup{x}$. Cycles are removed until there are none left and we obtain a tree $\tree_{\alpha^*}$, for some $\alpha^* \in \N$, which respects \emph{(i)} and \emph{(ii)}. Moreover, because there is no more cycle, the height of $\tree_{\alpha^*}$ is less than $|V|$.
Finally, from $\tree_{\alpha^*}$ we recover a strategy $\sigma'_1$ such that for all $\sigma_2 \in \stratSet{2}{v}$, we have that $\dCost(\outcome{\sigma'_1}{\sigma_2}{v}) \lesssim \tup{x}$ and $\lengthF{\outcome{\sigma'_1}{\sigma_2}{v}} \leq |V|$.\\

More precisely, let us assume that we want to build $\tree_{\alpha+1}$ from $\tree_{\alpha}$. There still exists a branch $ n_0\ldots n_k \ldots n_\ell \ldots n_m$ in $\tree_{\alpha}$ such that 
\begin{equation}
     \text{\emph{(i) }} \exists k, \ell \in \N, 0 \leq k < \ell \leq m, \text{ such that } \last(n_k)= \last(n_\ell) \text{ and \emph{(ii)} } \last(n_k) \in V_1 . \label{eq:treeAlphaCondition} 
\end{equation}

This is for example the case of the hatched nodes in Figure~\ref{fig:winningStratTree}. In this situation, $\playerOne$ generates an unnecessary cycle since the target set is not reached along this cycle and the cost profile increase along the cycle.  Thus $\playerOne$ can avoid this unnecessary cycle by directly choosing the doted edge.

We have to build a new tree $\tree_{\alpha+1}$ from $\tree_{\alpha}$. Let us recall that, while branches are sequences of nodes, trees are sets of non-empty histories. Thus, for all $w \in \tree_{\alpha}$, $w \in \hist(v)$. 
The new tree $\tree_{\alpha + 1}$ is obtained as follows: for all $w \in \tree_{\alpha}$:
\begin{itemize}
    \item if $w = \last(n_0) \ldots \last(n_{k}) \ldots \last(n_\ell) w_{k+1} \ldots w_n$ for some $n \in \N$ then, the cycle \\ $\last(n_k) \ldots \last(n_\ell)$ is removed and so
    $\last(n_0)\ldots  \last(n_k)w_{k+1} \ldots w_n  \in \tree_{\alpha +1}$;
    \item else, $w \in \tree_{\alpha+1}$.
\end{itemize}

We now consider $\tree_{\alpha^*}$ in which there is no more branch that satisfies conditions~\eqref{eq:treeAlphaCondition}. We prove that there is no branch $n_0\ldots n_k \ldots n_\ell \ldots n_m$ of $\tree_{\alpha^*}$ such that \emph{(i)} $\exists k, \ell \in \N$, $ 0 \leq k < \ell \leq m $, such that $\last(n_k) = \last(n_\ell)$ and \emph{(ii)} $\last(n_k) \in V_2$.

Let us assume the contrary in order to obtain a contradiction.
We consider the following two cases:

\begin{itemize}
    \item if for all $k < \xi < \ell $, $\last(n_\xi) \in V_2$ then,  $\last(n_0) \ldots (\last(n_{k}) \ldots \last(n_{\ell-1}))^j$ should be nodes of $\tree_{\alpha^*}$ for all $j \in \N$. This is in contradiction with the fact that the height of the root of $\tree_{\alpha^*}$ is finite.
    \item Otherwise, we consider the least index $\xi \in \N$ such that $k < \xi < \ell $ and $\last(n_\xi) \in V_1$. This is for example the case of the gray node between the two black nodes in Figure~\ref{fig:winningStratTree}.
    
    In this case, the node $t = \last(n_0) \ldots \last(n_k) \ldots \last(n_\xi) \ldots \last(n_\ell) \ldots \last(n_{\xi})$ should be a node of $\tree_{\alpha^*}$. Thus, there is  at least one branch in the tree which has $t$ as a prefix. Which contradicts the assumption that there is no more cycle generates by $\playerOne$ in $\tree_{\alpha^*}$.
\end{itemize}

Since we have removed all cycles of the finite branches of $\stratTree{\sigma_1}$. We have that the height of $\tree_{\alpha^*}$ is less than $|V|$ and the cost profiles of the branches may only decrease, because the weights on the edges are natural numbers. It means that for all branches $n_0 \ldots n_k$ of $\tree_{\alpha^*}$, $\dCost(n_0\ldots n_k) \lesssim \tup{x}.$

To conclude, we recover from $\tree_{\alpha^*}$ the strategy $\sigma'_1$ associated with $\tree_{\alpha^*}$ as explained in the paragraph about strategy tree at the beginning of Appendix~\ref{appendix:termination}. For all $\sigma_2 \in \stratSet{2}{v}$, we have \emph{(i)} $\dCost(\outcome{\sigma'_1}{\sigma_2}{v}) \lesssim \tup{x}$ and \emph{(ii)} $\lengthF{\outcome{\sigma'_1}{\sigma_2}{v}} \leq |V|$. 
\end{proof}

\subsubsection{Correctness}

Some of the arguments of our proofs rely on the following lemma and its corollary.

\begin{lemma}
\label{lemma:propOnIter}
For all $k \in \N$,
\begin{enumerate}
    
    \item for all $v \in V_1 \backslash \targetSet$, for all $\tup{x} \in \iter{k+1}{v}$, there exist $v' \in \successor(v)$ and $\tup{x} \in \iter{k}{v'}$ such that $\tup{x} = \tup{x'} + \dWeight(v,v')$. \label{item:propOnIter2}
    \item for all $v \in V_2 \backslash \targetSet$, for all $\tup{x} \in \iter{k+1}{v}$, for all $v' \in \successor(v)$, there exists $\tup{x'} \in \iter{k}{v'}$ such that $\tup{x'} + \dWeight(v,v') \lesssim \tup{x}$. \label{item:propOnIter3}
\end{enumerate}
\end{lemma}

\begin{corollary}
\label{cor:propOnFixpoint}
For all $v \in V$,
\begin{itemize}
    \item If $v \in V_1 \backslash \targetSet$ then, for all $\tup{x} \in \iter{*}{v}$, there exist $v' \in \successor(v)$ and $\tup{x} \in \iter{*}{v'}$ such that $\tup{x} = \tup{x'} + \dWeight(v,v')$. 
    \item If $v \in V_2 \backslash \targetSet$ then, for all $\tup{x} \in \iter{*}{v}$, for all $v' \in \successor(v)$, there exists $\tup{x'} \in \iter{*}{v'}$ such that $\tup{x'} + \dWeight(v,v') \lesssim \tup{x}$.
\end{itemize}
\end{corollary}

This section is devoted to the proof of the following theorem.

\begin{theorem}
\label{thm:correctness}
For all $v \in V$,
 $\minimal(\ensure{v}) = \iter{*}{v}.$
\end{theorem}

This is a direct consequence of Proposition~\ref{prop:ensureI}.

\begin{restate}{Proposition}{\ref{prop:ensureI}}
\restatePropEnsureI
\end{restate}

\begin{proof}
We proceed by induction on $\ell$. \underline{Base case} $\ell=0$, let $v \in V$. If $v \in \targetSet$ then, $\iterEnsure{0}{v} = \{ \tup{x} \in \NInf^\di \mid \tup{0} \lesssim \tup{x} \}$ and $\minimal(\iterEnsure{0}{v}) = \{ \tup{0} \}$ which is equal to $\iter{0}{v}$ by Algorithm~\ref{algo:fixPoint}. Else, if $v \not \in \targetSet$, $\iterEnsure{0}{v} = \{ \tup{\infty} \}$ and $\iter{0}{v} = \{ \tup{\infty} \}$.

Let us assume that the assertion is true for all $0 \leq \ell \leq k$ and let us prove it is still true for $\ell = k+1$.
In particular, the following equality holds:

\begin{equation}
    \minimal(\iterEnsure{k}{v}) = \iter{k}{v} \label{eq:inductionHypoEnsuredVal}
\end{equation}

Since $\iterEnsure{k}{v}$ is upward closed, we have that:
\begin{equation}
\iterEnsure{k}{v} =  \uparrow\iter{k}{v} \label{eq:inductionHypoEnsuredValUp}
\end{equation}

If $v \in \targetSet$, we have $\minimal(\iterEnsure{k}{v}) = \{ \tup{0} \} = \iter{k}{v}$ for all $k \in \N$. This is the reason why we assume $v \not\in \targetSet$  in the rest of the proof.

For all $v \not \in \targetSet$, we prove that $\iterEnsure{k+1}{v} = \begin{cases}\displaystyle \bigcup_{v' \in \successor(v)} \uparrow \iter{k}{v'} + \dWeight(v,v') & \text{ if } v\in V_1\\ \displaystyle \bigcap_{v' \in \successor(v)} \uparrow \iter{k}{v'} + \dWeight(v,v') & \text{ if } v \in V_2\end{cases}$. That proves that $\minimal(\iterEnsure{k+1}{v}) = \iter{k+1}{v}$.

\begin{itemize} 
\item We first prove the inclusion $\subseteq$. Let $\tup{x} \in \iterEnsure{k+1}{v}$. 

We know that there exists a strategy $\sigma_1^{k+1} \in \stratSet{1}{v}$ such that for all strategies $\sigma_2\in\stratSet{2}{v}$  we have that \begin{equation}\dCost(\outcome{\sigma_1^{k+1}}{\sigma_2}{v}) \lesssim \tup{x} \quad \text{ and } \quad \lengthF{\outcome{\sigma_1^{k+1}}{\sigma_2}{v}} \leq k+1. \label{eq:correctEnsure1}\end{equation}

\begin{itemize}
    \item If $v \in V_1$, let $v' = \sigma_1^{k+1}(v)$. We consider ${\sigma_1^{k+1}}_{\restriction v}: \hist_1(v') \longrightarrow V: hu \mapsto \sigma_1^{k+1}(vhu)$.

    We have for all $\sigma_2\in \stratSet{2}{v}$:

    \begin{align*}
        \dCost(\outcome{\sigma_1^{k+1}}{\sigma_2}{v}) &= \dCost(v \outcome{{\sigma_1^{k+1}}_{\restriction v}}{{\sigma_2}_{\restriction v}}{v'})\\
        &= \dWeight(v,v') + \dCost(\outcome{{\sigma_1^{k+1}}_{\restriction v}}{{\sigma_2}_{\restriction v}}{v'}) & (v \not \in \targetSet)
    \end{align*}

    Thus in particular, for all $\sigma_2 \in \stratSet{2}{v'}$:

    $$ \dWeight(v,v') + \dCost(\outcome{{\sigma_1^{k+1}}_{\restriction v}}{{\sigma_2}}{v'}) \lesssim \tup{x} $$

    and 

    $$ \lengthF{\outcome{{\sigma_1^{k+1}}_{\restriction v}}{{\sigma_2}}{v'}} = \lengthF{\outcome{{\sigma_1^{k+1}}}{{\sigma_2}}{v}} - 1 \leq k. $$

    Meaning that $\tup{x} - \dWeight(v,v') \in \iterEnsure{k}{v'}$.
    By~Equation~\eqref{eq:inductionHypoEnsuredValUp}, $\iterEnsure{k}{v'} = \uparrow\iter{k}{v'}$. It follows that $x \in \uparrow\iter{k}{v'} + \dWeight(v,v')$ and we obtain the result we were looking for: $\tup{x} \in \bigcup_{v' \in \successor(v)} \uparrow\iter{k}{v'} + \dWeight(v,v')$.

    \item If $v \in V_2$, for all $v' \in \successor(v)$ and for all strategies $\sigma_2 \in \stratSet{2}{v'}$, we have by Equation~\eqref{eq:correctEnsure1}:

    \begin{align*}
        \dCost(v \outcome{{\sigma_1^{k}}_{\restriction v}}{\sigma_2}{v'}) &= \dWeight(v,v') +  \dCost(\outcome{{\sigma_1^{k}}_{\restriction v}}{\sigma_2}{v'}) & (v \not \in \targetSet) \\
        &\lesssim \tup{x}
    \end{align*}

    and, since $v\not \in \targetSet$, 

    $$\lengthF{\outcome{{\sigma_1^{k}}_{\restriction v}}{\sigma_2}{v'}} = \lengthF{\outcome{\sigma_1^{k+1}}{\sigma_2}{v}} -1 \leq k. $$

    It follows that for all $v' \in \successor(v)$, $\tup{x} - \dWeight(v,v') \in \iterEnsure{k}{v'} = \uparrow \iter{k}{v'}$, by Equation~\eqref{eq:inductionHypoEnsuredValUp}. Thus we conclude that $\tup{x} \in \bigcap_{v' \in \successor(v)} \uparrow \iter{k}{v'} + \dWeight(v,v')$.
    
    \end{itemize}
    \item We now prove the inclusion $\supseteq$. 

    \begin{itemize}
        \item If $v \in V_1$, let $\tup{x} \in \displaystyle \bigcup_{v' \in \successor(v)} \uparrow \iter{k}{v'} + \dWeight(v,v')$. It means that there exists $\tup{y} \in \uparrow\iter{k}{v'} $ such that $\tup{x} = \tup{y} +\dWeight(v,v')$ for some $v' \in \successor(v)$.
        
        By Equation~\eqref{eq:inductionHypoEnsuredValUp}, $\tup{y} \in \iterEnsure{k}{v'}$, thus there exists $\sigma_1^k \in \stratSet{1}{v'}$ such that for all $\sigma_2 \in \stratSet{2}{v'}$ we have:
        \begin{equation}
            \dCost(\outcome{\sigma_1^k}{\sigma_2}{v'}) \lesssim \tup{y} \quad \text{ and } \quad \lengthF{\outcome{\sigma_1^k}{\sigma_2}{v'}} \leq k.\label{eq:correctEnsure2}
        \end{equation} 

        We consider $\sigma_1^{k+1} \in \stratSet{1}{v}$ defined as $$\sigma_1^{k+1}(vh) = \begin{cases} v' & \text{ if } h \text{ is the empty history, \IE } vh=v \\ \sigma_1^{k}(h) & \text{ otherwise} \end{cases}.$$

        Let $\sigma_2 \in \stratSet{2}{v}$, 

        \begin{align*}
            \dCost(\outcome{\sigma_1^{k+1}}{\sigma_2}{v}) &= \dWeight(v,v') + \dCost( \outcome{{\sigma_1^{k+1}}_{\restriction v}}{{\sigma_2}_{\restriction v}}{v'}) & (v \not \in \targetSet)\\
            &=\dWeight(v,v') + \dCost( \outcome{{\sigma_1^{k}}}{{\sigma_2}_{\restriction v}}{v'}) \lesssim \dWeight(v,v') + \tup{y} & \text{(By Eq.~\eqref{eq:correctEnsure2})}
        \end{align*}

        Moreover, $\lengthF{\outcome{\sigma_1^{k+1}}{\sigma_2}{v}} = 1 + \lengthF{\outcome{\sigma_1^{k}}{{\sigma_2}_{\restriction v}}{v'}} \leq 1 + k$.

        We conclude that $\tup{x} \in \iterEnsure{k+1}{v}$.

        \item If $v \in V_2$, let $\tup{x} \in \displaystyle \bigcap_{v' \in \successor(v)} \uparrow \iter{k}{v'} + \dWeight(v,v')$.

        For all $v' \in \successor(v')$, there exists $\tup{y'} \in \uparrow \iter{k}{v'}$ such that $\tup{x} = \tup{y'} + \dWeight(v,v')$.         
        By Eq.~\eqref{eq:inductionHypoEnsuredValUp}, there exists $\sigma_1^{v'} \in \stratSet{1}{v'}$ such that for all $\sigma_2 \in \stratSet{2}{v'}$,

        \begin{equation*}
            \dCost(\outcome{\sigma_1^{v'}}{\sigma_2}{v'}) \lesssim \tup{y'} \quad \text{ and } \quad \lengthF{\outcome{\sigma_k^{v'}}{\sigma_2}{v'}} \leq k
        \end{equation*}

        Let us consider $\sigma^{k+1}_1 \in \stratSet{1}{v}$ defined as: $\sigma^{k+1}_1(vv'h) = \sigma^{v'}_1(v'h)$ for all $vv'h \in \hist_1(v)$. Let $\sigma_2 \in \stratSet{2}{v}$, if $\sigma_2(v) = v'$, we have:

        \begin{align*}
            \dCost(\outcome{\sigma_1^{k+1}}{\sigma_2}{v}) &= \dCost(v \outcome{{\sigma_1^{k+1}}_{\restriction v}}{{\sigma_2}_{\restriction_v}}{v'}) = \dCost(v \outcome{{\sigma_1^{v'}}}{{\sigma_2}_{\restriction_v}}{v'})\\
            &= \dWeight(v,v') +\dCost(\outcome{{\sigma_1^{v'}}}{{\sigma_2}_{\restriction_v}}{v'}) & (v\not \in \targetSet)\\
            &\lesssim \dWeight(v,v') + \tup{y'}.
        \end{align*}
        
        Moreover, $\lengthF{\outcome{\sigma_1^{k+1}}{\sigma_2}{v}} = \lengthF{v \outcome{{\sigma_1^{v'}}}{{\sigma_2}_{\restriction_v}}{v'}} = \lengthF{\outcome{{\sigma_1^{v'}}}{{\sigma_2}_{\restriction_v}}{v'}} + 1 \leq k +1 $. In conclusion, $\tup{x} \in \iterEnsure{k+1}{v}$. 
    \end{itemize}

\end{itemize}
\end{proof}

\subsection{Time Complexity}

Let us recall that $\maxW = \max \{ \weight_i(e) \mid 1 \leq i \leq \di \text{ and } e \in E \}$. We also explicitly restate a remark done in the main part of the paper (in Section~\ref{section:correctStrat}).

\begin{remark}
In Line~\ref{lineAlgo:stratEnd}, we are allowed to assume that $\tup{x'}$ is in $ \iter{k}{v'}$ instead of $\uparrow \iter{k}{v'}$ thanks to Lemma~\ref{lemma:propOnIterUnion} stated just after this remark.
\end{remark}

\begin{lemma}
    \label{lemma:propOnIterUnion}
    For all $k \in \N$, for all $v \in V_1 \backslash \targetSet$, $$\iter{k+1}{v} =  \displaystyle \minimal \left(\bigcup_{v' \in \successor(v)} \iter{k}{v} + \dWeight(v,v') \right).$$
\end{lemma}


\subsubsection{Lexicographic Order}

In this section we prove Theorem~\ref{thm:complexityLexico}.

\begin{restate}{Theorem}{\ref{thm:complexityLexico}}
    \restateThmComplexityLexico
\end{restate}

By abuse of notation, the only $\tup{x} \in \iter{k}{v}$ is denoted by $\iterUpValue{k}{v}$.

\begin{proposition}
\label{prop:unionIsAMinLexico}
If $v \in V_1 \backslash \targetSet$,
$\iter{k+1}{v} = \min_{\leqL} \{ \iterUpValue{k}{v'}  + \dWeight(v,v') \mid v' \in \successor(v)\}$.
\end{proposition}

\begin{proof}
    Let $v\in V_1 \backslash \targetSet$, 

    \begin{align*}
        \iter{k+1}{v} &= \displaystyle \minimal \left( \bigcup_{v' \in \successor(v)}  \uparrow \iter{k}{v'} + \dWeight(v,v') \right) & \text{ By Algorithm~\ref{algo:fixPoint}} \\
                      &= \displaystyle \minimal \left( \bigcup_{v' \in \successor(v)} \iter{k}{v'} + \dWeight(v,v') \right)  & \text{ By Lemma~\ref{lemma:propOnIterUnion}}   \\
                      &= \minimal( \{ \iterUpValue{k}{v'} + \dWeight(v,v') \mid v' \in \successor(v) \})\\
                      &= \min_{\leqL} \{ \iterUpValue{k}{v'} + \dWeight(v,v') \mid v' \in \successor(v) \} 
    \end{align*}
\end{proof}

\begin{proposition}
\label{prop:interIsAMaxLexico}
If $v \in V_2\backslash \targetSet$,
$\iter{k+1}{v} = \max_{\leqL} \{ \iterUpValue{k}{v'}  + \dWeight(v,v') \mid v' \in \successor(v)\}$.
\end{proposition}

\begin{proof}
We begin the proof by a remark: if $\tup{x}, \tup{y} \in \NInf^\di$ then, 
\begin{equation}
    \uparrow \{ \tup{x} \}\, \cap \, \uparrow \{ \tup{y} \} = \uparrow \{ \max_{\leqL} \{ \tup{x}, \tup{y} \} \} \label{eq:interLexico}
\end{equation}

    Let $v \in V_2 \backslash \targetSet$,

    \begin{align*}
    \iter{k+1}{v} &= \displaystyle \minimal \left( \bigcap_{v' \in \successor(v)}  \uparrow \iter{k}{v'} + \dWeight(v,v') \right) & \text{ By Algorithm~\ref{algo:fixPoint}} \\
    &= \minimal \left( \bigcap_{v' \in \successor(v)} \uparrow \{ \iterUpValue{k}{v'} + \dWeight(v,v') \} \right) \\
    &= \minimal( \max_{\leqL} \{ \iterUpValue{k}{v'} + \dWeight(v,v') \mid v' \in \successor(v) \}) & \text{ By Equation}~\eqref{eq:interLexico}\\
    &= \max_{\leqL} \{ \iterUpValue{k}{v'} + \dWeight(v,v') \mid v' \in \successor(v) \}
    \end{align*}
\end{proof}

\begin{proof}[Proof of Theorem~\ref{thm:complexityLexico}]
By Proposition~\ref{prop:algoTerminatesSteps}, Algorithm~\ref{algo:fixPoint} mainly consists in $(|V|+1) \cdot |V| \approx |V|^2$ operations of the type $\min_{\leqL} \{ \iterUpValue{k}{v'} + \dWeight(v,v') \mid v' \in \successor(v)\} $ or $\max_{\leqL} \{ \iterUpValue{k}{v'} + \dWeight(v,v') \mid v' \in \successor(v)\}$ which are done in $\bigO(|V|\cdot \di)$. It follows that the global complexity of the algorithm if the considered order is the lexicographic order is in $\bigO(|V|^3\cdot \di)$. 
\end{proof}


\subsubsection{Componentwise Order}

This section is devoted to the proof of Theorem~\ref{thm:complexityComponent}.

\begin{restate}{Theorem}{\ref{thm:complexityComponent}}
\restateThmComplexityComponent
\end{restate}

During the computation of the fixpoint algorithm, even if we have a finite representation of the infinite   sets $\uparrow \iter{k}{v}$ by only storing their minimal elements $\iter{k}{v}$, we need to explain how to manipulate them efficiently.
In particular, we explicit how given some accurate representation of  $\uparrow \iter{k}{v}$ and $\uparrow \iter{k}{v'}$ for some $k \in \N$ and $v,v' \in V$ we compute: \emph{(i)} the union $\uparrow \iter{k}{v}\, \cup \uparrow \iter{k}{v'}$, \emph{(ii)} the intersection $\uparrow \iter{k}{v}\, \cap \uparrow \iter{k}{v'}$, \emph{(iii)} the translation $\uparrow \iter{k}{v} + \dWeight(v,v')$ and \emph{(iv)}  the set of minimal elements $\minimal(\uparrow \iter{k}{v})$. Inspired by the approach explained in~\cite{DelzannoR00}, we use a part of the logic of upward closed sets in order to express the infinite sets $\uparrow \iter{k}{v}$ in a convenient way.\\

Let $\varLogSet = \{ \varLog_1, \ldots, \varLog_{\di}\}$ be a set of $\di$ variables, if $ \generator = \{ \tup{x^1}, \ldots , \tup{x^n}\}$ for some $n \in \N$ and $\tup{x^1}, \ldots , \tup{x^n} \in \NInf^\di$, we can express $\uparrow \generator$ as a formula $\phi$:

\begin{equation} \phi = \bigvee_{1 \leq i \leq n} (\varLog_1 \geq x^i_1) \wedge \ldots \wedge (\varLog_\di \geq x^i_\di).\label{eq:phi}\end{equation}

We define the size of the formula, denoted by $|\phi|$, by $n \cdot \di$. Additionally, the set $\generator$ is called the set of \emph{generators} of $\uparrow \generator$, or equivalently the set of generators of $\phi$. Thus $\generator$ allows to encode the formula in a succinct way. Notice that the fewer generators there are, the more succinct the formula to express $\uparrow\generator$ is.  Moreover, the set of tuples that evaluates formula $\phi$ to true are denoted by $\evalFormula{\phi}$, \IE
$\evalFormula{\phi} = \{ \tup{c} \in \NInf^\di \mid \bigvee_{1 \leq i \leq n } (c_1 \geq x^i_1) \wedge \ldots (c_\di \geq x^i_\di) \}.$
Thus, in particular, $\evalFormula{\phi} = \uparrow \generator$. Conversely, if we have a formula $\phi$ as in Equation~\eqref{eq:phi}, it represents an upward closed set $\evalFormula{\phi}$ and its set of generators is given by $\generatorFormula(\phi) = \{ \tup{x^1}, \ldots, \tup{x^n} \}$.

For each $\uparrow\iter{k}{v}$, we denote by $\phi(k,v)$ the corresponding formula.
In Proposition~\ref{prop:opOnClosedSets}, we explain how unions, intersections and translations of sets of the type $\uparrow \iter{k}{v}$ are done and what are the complexities of those operations. 

\begin{proposition}[\cite{DelzannoR00}]
\label{prop:opOnClosedSets}
Given two sets $\iter{k}{v} = \{ \tup{x^1}, \ldots , \tup{x^n}\}$, for some $n \in \N$, and $\iter{k}{v'} = \{ \tup{y^1}, \ldots , \tup{y^m}\}$, for some $m \in \N$, such that their upward closures are expressed respectively by $\phi(k,v)$ and $\phi(k,v')$, we have:

\begin{enumerate}
    \item \textbf{Union}: the set $ X = \uparrow \iter{k}{v}\, \cup \uparrow \iter{k}{v'}$ is expressed thanks to the formula $$\psi = \displaystyle \bigvee_{1\leq i \leq n+m } (\varLog_1 \geq z^i_1) \wedge \ldots \wedge ( \varLog_\di \geq z^i_\di) $$ where $\tup{z^i} = \tup{x^i}$ if $1 \leq i \leq n$ and $\tup{z^i} = \tup{y^{i-n}}$ if $n+1 \leq i \leq n+m$. Thus $|\psi| = |\phi(k,v)| +|\phi(k,v')|$ and this operation is done in $\bigO(|\phi(k,v)| +|\phi(k,v')|)$.\label{item:union}
    \item \textbf{Intersection}: the set $ X = \uparrow \iter{k}{v}\, \cap \uparrow \iter{k}{v'}$ is expressed thanks to the formula $$\psi = \displaystyle \bigvee_{1\leq i \leq n} \bigvee_{1 \leq j \leq m} (\varLog_1 \geq \max \{ x^i_1, y^j_1\}) \wedge \ldots \wedge ( \varLog_\di \geq \max \{ x^i_\di, y^j_\di \}). $$ Thus $|\psi| = |\phi(k,v)| \cdot |\phi(k,v')|$ and this operation is done in $\bigO(|\phi(k,v)| \cdot |\phi(k,v')|)$.\label{item:intersection}
    \item \textbf{Translation}: the set $ X = \uparrow \iter{k}{v} + \tup{c}$ is expressed thanks to the formula $$\psi = \displaystyle \bigvee_{1\leq i \leq n } (\varLog_1 \geq x^i_1 + c_1) \wedge \ldots \wedge ( \varLog_\di \geq x^i_\di + c_\di).$$ Thus $|\psi| = |\phi(k,v)|$ and this operation is done in $\bigO(|\phi(k,v)|)$.
    
\end{enumerate}

\end{proposition}

Even if the sets $\iter{k}{v}$ and $\iter{k}{v'}$ are minimal, an union or an intersection as  described in Statements~\ref{item:union} and~\ref{item:intersection} in Proposition~\ref{prop:opOnClosedSets} may produce a formula $\psi$ such that set $\generatorFormula(\psi)$ is not minimal. Therefore we consider the minimization of a set of generators in order to obtain a (minimal) new set of generators that encodes a new formula $\phi'$ in such a way that $\evalFormula{\phi'} = \evalFormula{\phi}$ and $|\phi'|$ is as small as possible. Notice that the translation operation preserves the minimality of the set of generators.

\begin{proposition}[\cite{DelzannoR00}]
\label{prop:minimizationComplexity}
If an upward closed set $X$ is expressed by $\phi$ with $\generator  = \generatorFormula(\phi) $ and $X' = \minimal(X)$, then $\generator'= \minimal(\generator)$ can be computed in $\bigO(|\phi|^2)$.    
\end{proposition}

\begin{remark}
\label{rem:minUpwardClosedSetIsMinGenerators}
Notice that in the previous proposition, as $X$ is upward closed, $\generator' = \minimal(\generator) = \minimal(\uparrow \generator) = \minimal(X)$.
\end{remark}

The key idea in order to obtain an algorithm at most polynomial in $\maxW$ and $|V|$ and exponential in $\di$ is to ensure that the size of the formulae, and so their sets of generators, that represent the sets $\uparrow\iter{k}{v}$ do not grow too much. The size of such  a formula depends on the number of elements in $\iter{k}{v}$ and the number of dimensions $\di$. Since for all $k \in \N$, $\iterEnsure{k}{v} \subseteq \iterEnsure{k+1}{v} \subseteq \iter{|V|}{v}$ and $|\iterEnsure{|V|}{v}\backslash \{ \tup{\infty} \}| \leq (\maxW \cdot |V|)^\di $, the maximal size of a set $\iter{k}{v} = \minimal(\iterEnsure{k}{v})$ is also bounded by $(\maxW\cdot |V|)^\di$. Let $\maxAC = \maxW ^\di\cdot |V|^\di$ be this (rough) upper-bound.

\begin{proposition}
\label{prop:complexityUnionComponent}
For all $k \in \N$ and $v \in V_1\backslash \targetSet$, 
the operation $$\iter{k+1}{v} = \minimal\left(\bigcup_{v'\in \successor(v)} \uparrow \iter{k}{v'} + \dWeight(v,v')\right)$$ can be computed in $\bigO( \di^2\cdot \maxW^{2\di} \cdot |V|^{2\cdot\di + 2})$.
\end{proposition}

\begin{proof}
    Let $k \in \N$ and $v \in V_1 \backslash \targetSet$. For all $v' \in \successor(v)$, we denote by $\phi(k,v')$ the formulae that express the sets $\uparrow \iter{k}{v'}$. By hypothesis, for all $v' \in \successor(v)$, $|\phi(k,v')| \leq \maxAC \cdot \di$. 

    \begin{itemize}
    \item \textbf{Translations.} We first compute the sets $\uparrow \iter{k}{v'} + \dWeight(v,v')$ by computing their associated formulae that we denote by $\phi(k,v') + \dWeight(v,v')$. Notice that computing each of those formulae can be done in $\bigO(|\phi(k,v')|) = \bigO(\maxAC \cdot \di)$ and that the obtained formula has a size $|\phi(k,v') + \dWeight(v,v')| = |\phi(k,v')|$,  by Proposition~\ref{prop:opOnClosedSets}. In conclusion, the computation of all translations is done in $\bigO(|V| \cdot \maxAC \cdot \di)$.
    \item \textbf{Unions.} Let us denote by $\psi$ the formula that expresses $\displaystyle \bigcup_{v' \in \successor(v)} \uparrow \iter{k}{v'} + \dWeight(v,v')$ which can be obtained thanks to successive unions and by Proposition~\ref{prop:opOnClosedSets} in \begin{align*}\bigO( \sum_{v' \in \successor(v)} |\phi(k,v') +\dWeight(v,v')|) &= \bigO( \sum_{v' \in \successor(v)} |\phi(k,v') |) =  \bigO(\maxAC \cdot \di \cdot|V|).\end{align*}
    \item \textbf{Minimization of the set of generators of $\psi$.} It remains to compute thanks to $\psi$ the set $\minimal(\displaystyle \bigcup_{v' \in \successor(v)} \uparrow \iter{k}{v'} + \dWeight(v,v'))$ which corresponds to the minimization of the set of generators of $\psi$ by Remark~\ref{rem:minUpwardClosedSetIsMinGenerators}. This can be done in $\bigO(|\psi|^2) = \bigO(\maxAC^2\cdot d^2 \cdot |V|^2)$, by Proposition~\ref{prop:minimizationComplexity}.
    \end{itemize}

    In conclusion, the global complexity of computing $\iter{k+1}{v}$ for $v \in V_1 \backslash \targetSet$ is $\bigO(\maxW^{2\di}\cdot |V|^{2\di} \cdot \di^2 \cdot |V|^2) =  \bigO ( \di^2 \cdot \maxW^{2\di}\cdot |V|^{2\di +2})$.
\end{proof}

\begin{proposition}
\label{prop:complexityInterComponent}
For all $k \in \N$ and $v \in V_2 \backslash \targetSet$, the operation $$\iter{k+1}{v} = \minimal\left(\bigcap_{v'\in \successor(v)} \uparrow \iter{k}{v'} + \dWeight(v,v')\right)$$ can be computed in $\bigO(\di^4 \cdot 
\maxW^{4\di} \cdot |V|^{4\di + 1})$.
\end{proposition}

\begin{proof}[Proof of Proposition~\ref{prop:complexityInterComponent}]

 Let $k \in \N$ and $v \in V_2 \backslash \targetSet$. For all $v' \in \successor(v)$, we denote by $\phi(k,v')$ the formulae that express the sets $\uparrow \iter{k}{v'}$. By hypothesis, for all $v' \in \successor(v)$, $|\phi(k,v')| \leq \maxAC \cdot \di$. The line~\ref{lineAlgo:inter} of Algorithm~\ref{algo:fixPoint} may be replaced by:

\begin{algorithm}
\SetAlgoNoEnd
$I = \iter{k}{w}$ for some $w \in \successor(v)$\; \label{algo:interBisBegin}
\For{$v'\in \successor(v)$}
{
    $J = \uparrow \iter{k}{v'} + \dWeight(v,v')$ \; \label{algo:interBis1}
    $I = \minimal(\uparrow I \cap J) \;$\label{algo:interBis2}
    
}
$\iter{k+1}{v} = I$\;\label{algo:interBisEnd}
\end{algorithm}

Let us analyze the complexity of those lines. We assume that $\phi(J)$ and  $\phi(\uparrow I)$ are the formulae that express the sets $J$  and $\uparrow I$ respectively. Thanks to the minimization of the set of generators of the corresponding formula in Line~\ref{algo:interBis2} , the formulae $\phi(J)$ and $\phi(\uparrow I)$ have a size at most equal to $\maxAC \cdot \di$.

\begin{itemize}
    \item \textbf{Complexity of Line~\ref{algo:interBis1}}: $\bigO(|\phi(k,v')|) = \bigO(\maxAC \cdot \di)$, by Proposition~\ref{prop:opOnClosedSets};
    \item \textbf{Complexity of Line~\ref{algo:interBis2}}: the intersection is done in $\bigO(\maxAC^2 \cdot \di^2)$, by Proposition~\ref{prop:opOnClosedSets}, and generates a formula $\phi(\uparrow I \cap J)$ of size at most $|\phi(\uparrow I \cap J)| \leq \maxAC^2 \cdot \di^2$. The minimization of the set of generators of $\uparrow I \cap J$ is done in $\bigO(\maxAC^4 \cdot \di^4)$, by Proposition~\ref{prop:minimizationComplexity}, and allows to encode a formula of size at most $\maxAC \cdot \di$ which also expresses $\uparrow I \cap J$.

\item The \textbf{global complexity of Lines~\ref{algo:interBisBegin} to~\ref{algo:interBisEnd}}, is $|V| \cdot ( \bigO(\maxAC \cdot \di) + \bigO(\maxAC^2\di^2) + \bigO(\maxAC^4 \cdot \di^4)) = \bigO(|V| \cdot \maxAC^4 \cdot \di^4) = \bigO(\di^4 \cdot \maxW^{4\di} \cdot |V|^{4\di+1})$. 
\end{itemize}
$ $
\end{proof}

\begin{proof}[Proof of Theorem~\ref{thm:complexityComponent}]
Since the algorithm terminates in less than $|V|+1$ steps (Proposition~\ref{prop:algoTerminatesSteps}), the fixpoint algorithm consists of about $|V|$ repetitions of the procedure between Line~\ref{lineAlgo:beginningFirstForLoop} and Line~\ref{lineAlgo:inter}. This procedure is a for loop on all the vertices of the game graph which computes essentially either an operation $\displaystyle\minimal\left(\bigcup_{v'\in \successor(v)} \uparrow \iter{k}{v'} + \dWeight(v,v')\right)$ or $\displaystyle\minimal\left(\bigcap_{v'\in \successor(v)} \uparrow \iter{k}{v'} + \dWeight(v,v')\right)$. Thus, by Proposition~\ref{prop:complexityUnionComponent} and Proposition~\ref{prop:complexityInterComponent} the complexity of the fixpoint algorithm is in $\cdot \bigO( |V|^2  \cdot  \max \{ \di^2 \cdot \maxW^{2\di}  \cdot |V|^{2\di + 2}\,,\,  \di^4 \cdot \maxW^{4\di} \cdot |V|^{4\di + 1} \}) =  \bigO( |V|^2 \cdot \di^4 \cdot\maxW^{4\di} \cdot |V|^{4\di + 1}) = \bigO( \di^4\cdot W^{4\di} \cdot |V|^{4\di + 3})$. 
\end{proof}

\subsection{Synthesis of Lexico-optimal and Pareto-optimal Strategies}

In this section, we prove Theorem~\ref{thm:optiStrat} and Proposition~\ref{prop:stratOptiLexicoPositional}.

\begin{restate}{Theorem}{\ref{thm:optiStrat}}
\restateThmOptiStrat
\end{restate}

To prove Theorem~\ref{thm:optiStrat}, we consider the strategy tree $\stratTree{\sigma^*_1}$ of $\sigma^*_1$ and introduce a labeling function of the tree nodes which allows to keep track some properties on these nodes. This labeling function and properties are detailed in the following proposition.

\begin{proposition}
\label{prop:labeling}
    For $u \in V$ and $\tup{c} \in \iter{*}{u} \backslash \{ \tup{\infty} \} $, if $\stratTree{\sigma^*_1}$ is the strategy tree of the strategy $\sigma^*_1$ as defined in Definition~\ref{def:optiStrat} then, there exists a labeling function $\tau: \stratTree{\sigma^*_1} \longrightarrow \N^\di$ such that, $\tau(u) = \tup{c} \in \iter{*}{u}$ and, for all $hv \in \stratTree{\sigma^*_1}$ such that $|hv| \geq 1$:
    \begin{enumerate}
        \item $\tau(hv) \in \iter{*}{v}$; \label{item:inv1Label}
        \item If $\last(h) \in V_1$ then, $(v, \tau(hv)) = f^*_{\last(h)}(\tau(h))$;\label{item:inv2Label}

        \item $\tau(hv) \leqL \tau(h) - \dWeight(\last(h),v)$; \label{item:inv3Label}
        \item $\tau(hv) = \min_{\leqL} \{ \tup{x'} \in \iter{*}{v} \mid \tup{x'} \lesssim \tup{c} - \dCost(hv)\, \wedge \, \tup{x'} \leqL \tup{c} - \dCost(hv) \}$. \label{item:inv4Label}
    \end{enumerate}
\end{proposition}

\begin{remark}The same remark as in Remark~\ref{rem:technicalRem} is applicable in the context of Proposition~\ref{prop:labeling}. Even if the lexicographic order $\leqL$ is used in the statement of the properties of the labeling function $\tau$, Proposition~\ref{prop:labeling} holds both for the lexicographic order and the componentwise order.
\end{remark}

The intuition behind the properties on the labeling function $\tau$ in Proposition~\ref{prop:labeling} is the following one.
The first property ensures that the values of $\tau$ is one of the ensured value at the fixpoint in the set corresponding to the last vertex of the node. The second property ensures that the construction of $\tau$ is consistent with the strategy $\sigma^*_1$. The third property ensures that when we follow a branch of the tree, the value of $\tau$ decreases along it, this to guarantee that the target set is actually reached. The fourth condition ensures that when we follow a branch of the tree, at the end, the cost is below $\tup{c}$. 
The most important of them are summarized in Figure~\ref{fig:invariantLabeling}.

\begin{figure}[ht]
\centering
\scalebox{0.9}{
\begin{tikzpicture}
    \node[draw, circle, inner sep=2pt] (u0) at (0,0){$u$};
    \draw[->,thick] (u0) to (1,0);
    \node (dots1) at (1.3,0){$\ldots$};

    \node[draw,circle, inner sep=2pt] (hv) at (4,1){$hv$};
    \node[draw, inner sep=3pt] (hw) at (4,-1){$hw$};

    \node[draw, diamond, inner sep=0pt] (hvv') at (6,1){$hvv'$};
    \node[draw, diamond, inner sep=0pt] (hww1) at (6,-0.3){$hww_1$};
    \node[draw, diamond, inner sep=0pt] (hww2) at (6,-1.7){$hww_2$};

    \draw[->,thick] (hv) to (hvv');
    \draw[->,thick] (hw) to[bend left] (hww1);
    \draw[->,thick] (hw) to[bend right] (hww2);

    \draw[->,thick] (3,0.8) to[bend left] (hv);
    \draw[->,thick] (3,-0.8) to[bend right] (hw);
    
    \node (dots2) at (7,1){\ldots};
    \node (dots3) at (7,-0.3){\ldots};
    \node (dots4) at (7,-1.7){\ldots};
    \draw[->,thick] (hw) to (5,-1);
    \node (dots5) at (5.3,-1){$\ldots$};    

    \node [draw, dotted] (label1) at (0,1.5){\begin{minipage}{3cm} $\tau(u) = \tup{c} \in \iter{*}{u}$\end{minipage}};
    \draw[->, dotted] (u0) to (label1);

    \node[draw, dotted] (label2) at (0,-1.5){\begin{minipage}{4.8cm} 
    If $t$ is $v$ or $w$:
    $\tau(ht) \in \iter{*}{t}$\\
    $\tau(ht) \lesssim \tup{c} - \dCost(ht)$\\
    $\tau(ht) \leqL \tau(h) - \dWeight(\last(h),t)$\end{minipage}};
    \draw[->,dotted] (hw) to (4,-1.5) to (label2);

    \node[draw, dotted] (label3) at (10.5,1){\begin{minipage}{6cm} 
    $v' = f^*_v(\tup{x})[1]$ with $\tup{x} = \min_{\leqL} \cover(hv)$\\
    $\tau(hvv') \in \iter{*}{v'}$\\
    $\tau(hvv') \lesssim \tup{c} - \dCost(hvv')$\\
    $\tau(hvv') \leqL \tau(hv) - \dWeight(v,v')$\end{minipage}};
    \draw[->,dotted] (label3.north) to (10.5,2.5) to (6,2.5) to (hvv');

    \node[draw, dotted] (label4) at (10.5,-1.5){\begin{minipage}{6cm} 
     $\forall w' \in \successor(w):$
    $\tau(hvw') \in \iter{*}{w'}$\\
    $\tau(hvw') \lesssim \tup{c} - \dCost(hvw')$\\
    $\tau(hvw') \leqL \tau(hv) - \dWeight(v,w')$\end{minipage}};
    \draw[->, dotted] (label4.south) to (10.5,-2.8) to (6,-2.8) to (hww2.south);

\end{tikzpicture}
}
\caption{Labeling function associated with the strategy tree $\stratTree{\sigma^*_1}$}
\label{fig:invariantLabeling}
\end{figure}

In order to prove Proposition~\ref{prop:labeling}, we need some technical results about the sets $\iter{k}{v}$ and the functions $f^*_v$.  For all $v \not \in \targetSet$, for all $\tup{x} \in \iter{*}{v}\backslash \{ \infty \}$, we introduce the notation $\firstOcc{\tup{x}}{v}$ to denote the first index such that $\tup{x} \in \iter{\firstOcc{\tup{x}}{v}}{v}$. 

\begin{lemma}
\label{lemma:monotonicity}
For all $v \not \in \targetSet$, if  $\tup{x} \in \iter{*}{v}\backslash \{ \infty \} $ then, for all $\ell \geq \firstOcc{\tup{x}}{v}$, $\tup{x} \in \iter{\ell}{v}$. 
\end{lemma}

This lemma states that if a cost profile $\tup{x}$ is in the fixpoint $\iter{*}{v}$ for some $v$ then, this cost profile stays in $\iter{k}{v}$ from its first appearance to the fixpoint.

\begin{proof}
Let $v \not \in \targetSet$ and let $\tup{x} \in \iter{*}{v}\backslash \{ \infty \} $. In the rest of the proof we set $n = \firstOcc{\tup{x}}{v}$.

    To obtain a contradiction, let us assume that there exists $\ell$ such that $\ell > n$ such that $\tup{x} \not \in \iter{\ell}{v}$.

    Because $ \ell > n $ and by Proposition~\ref{prop:fixPointReach}, $\tup{x} \in \iterEnsure{\ell}{v}$. Since, by Proposition~\ref{prop:ensureI}, $ \tup{x} \not \in \iter{\ell}{v} = \minimal(\iterEnsure{\ell}{v})$, there exists $\tup{x^*} \in \iterEnsure{\ell}{v}$ such that $\tup{x^*} < \tup{x}$. Once again, by Proposition~\ref{prop:fixPointReach}, $\tup{x^*} \in \iterEnsure{*}{v}$. 
     But, we have assumed that $\tup{x} \in \iter{*}{v}$ and  by Theorem~\ref{thm:correctness} we have that $\iter{*}{v}= \minimal(\iterEnsure{*}{v})$. Thus $\tup{x^*} <\tup{x}$ and $\tup{x^*} \in \iterEnsure{*}{v}$ leads to a contradiction with the minimality of $\tup{x}$ in $\iterEnsure{*}{v}$. 
\end{proof}

\begin{lemma}
    \label{lemma:correctStratMin1}
    For all $v \in V_1 \backslash \targetSet$, for all $\tup{x} \in \iter{*}{v} \backslash \{ \infty \}$,
    if $(v', \tup{x'}) = f^*_v(\tup{x})$ then, $\tup{x'} \in \iter{*}{v'}$ and $\tup{x'} = \tup{x} - \dWeight(v,v')$.
\end{lemma}

\begin{proof}
    In the proof we set $n = \firstOcc{\tup{x}}{v}$.

By construction and Proposition~\ref{prop:ensureI}, $\tup{x} = \tup{x'} + \dWeight(v,v')$, $\tup{x} \in \iter{n}{v} = \minimal(\iterEnsure{n}{v})$ and $\tup{x'} \in \iter{n-1}{v'} = \minimal(\iterEnsure{n-1}{v'})$.
The second part of the assertion is already proved. Let us prove the other one. \\

In order to obtain a contradiction, let us assume that there exists $\ell$ such that $n-1 < \ell \leq k^*$  and $\tup{x'} \not \in \iter{\ell}{v'}$.

Since $ n-1 < \ell$ and by Proposition~\ref{prop:fixPointReach}, we have that $\tup{x'}\in \iterEnsure{\ell}{v'}$. But as $\tup{x'} \not \in \iter{\ell}{v'}$ and $\iter{\ell}{v'} = \minimal(\iterEnsure{\ell}{v'})$ (by Proposition~\ref{prop:ensureI}),
that means that there exists $\tup{y'} \in \iterEnsure{\ell}{v'}$ such that $\tup{y'} < \tup{x'}$.
It follows that $\tup{y'} + \dWeight(v,v') < \tup{x'} + \dWeight(v,v') = \tup{x}$ and so $\tup{y'} + \dWeight(v,v') \in \iterEnsure{\ell+1}{v}$.

Because $n < \ell +1$, thanks to Lemma~\ref{lemma:monotonicity}, we have that $\tup{x} \in \iter{\ell+1}{v}$. Moreover, by Proposition~\ref{prop:ensureI}, $\iter{\ell+1}{v} = \minimal(\iterEnsure{\ell+1}{v})$. Thus because $\tup{y'} + \dWeight(v,v) < \tup{x}$ and  $\tup{y'} + \dWeight(v,v') \in \iterEnsure{\ell+1}{v}$, we obtain a contradiction with the fact that $\tup{x}$ is minimal in $\iterEnsure{\ell+1}{v}$. 
\end{proof}

We are now able to prove Proposition~\ref{prop:labeling}.

\begin{proof}[Proof of Proposition~\ref{prop:labeling}]

    Let $u \in V$ and $\tup{c} \in \iter{*}{u}\backslash \{ \tup{\infty} \}$. Let $\tree^* = \stratTree{\sigma^*_1}$. We define $\tau$ and prove Invariant~\eqref{item:inv1Label} to~\eqref{item:inv4Label} step by step, by induction on the length of $h \in \tree^*$.

    \textbf{Base case} If $h = uv$ for some $v \in V$.
    
    \begin{itemize}
    \item If $u \in V_1$: We define $\tau(uv) = f^*_u(\tau(u))[2] = f^*_u(\tup{c})[2]$. By construction $\tau(u) = \tup{c} \in \iter{*}{u}$ thus $f^*_u(\tau(u))$ is well defined. Since $u \in V_1$, $v = \sigma^*_1(u)$ and by definition of $\sigma^*_1$, $v = f^*_u(\tup{x})[1]$ where $\tup{x} = \min_{\leqL} \{ \tup{x'} \in \iter{*}{u} \mid \tup{x'} \lesssim \tup{c} - \dCost(u) \, \wedge \, \tup{x'} \leqL \tup{c} - \dCost(u) \} = \min_{\leqL} \{ \tup{x'} \in \iter{*}{u} \mid \tup{x'} \lesssim \tup{c}  \, \wedge \, \tup{x'} \leqL \tup{c}  \} = \min_{\leqL} \cover(u)$. Since $\tup{c} \in \cover(u)$, $\tup{x} \in \iter{*}{u}$. But $\tup{x}, \tup{c} \in \iter{*}{u} = \minimal (\ensure{u})$, by Theorem~\ref{thm:correctness}, implies $\tup{x} = \tup{c}$. It follows that $v = f^*_u(\tup{c})[1] = f^*_u(\tau(u))[1]$. Consequently, Invariants~\eqref{item:inv1Label} and~\eqref{item:inv2Label} are satisfied.

    Since $\tau(uv) = f^*_u(\tau(u))[2]$, by Lemma~\ref{lemma:correctStratMin1}, $\tau(uv) = \tau(u) - \dWeight(u,v)$. That implies Invariant~\eqref{item:inv3Label}. 

    Since $\tau(uv) = \tau(u) - \dWeight(u,v)$, $\tau(u) = \tup{c}$ and $\dWeight(u,v) = \dCost(uv)$, we have that $\tau(uv) \lesssim \tup{c} - \dCost(uv)$ and $\tau(uv) \leqL \tup{c} - \dCost(uv)$. Thus, $\tau(uv) \in \{ \tup{x'} \in \iter{*}{v} \mid \tup{x'} \lesssim \tup{c} - \dCost(uv) \, \wedge \, \tup{x'} \leqL \tup{c} - \dCost(uv) \} = \cover(uv)$. It remains to prove that $\tau(uv) = \min_{\leqL} \cover(uv)$. By contradiction, we assume that there exists $\tup{y} \in \iter{*}{v}$ such that \emph{(i)} $\tup{y} \lesssim \tup{c} - \dCost(uv)$, \emph{(ii)} $\tup{y} \leqL \tup{c} - \dCost(uv)$ and \emph{(iii)} $\tup{y} \strictLessL \tau(uv)$. Let us recall that $\tau(uv) = \tau(u) - \dWeight(u,v)$ and $\dWeight(u,v) = \dCost(uv)$. Therefore, by \emph{(i)}, we have that $\tup{y} \lesssim \tup{c} - \dWeight(u,v) = \tau(u) - \dWeight(u,v) = \tau(uv)$.
    Finally, as $\tup{y}, \tau(uv) \in \iter{*}{v} = \minimal(\ensure{v})$, by Theorem~\ref{thm:correctness}, $\tup{y} = \tau(uv)$ which is a contradiction with \emph{(iii)}. That concludes the proof of Invariant~\eqref{item:inv4Label}.

    \item If $u \in V_2$: we define $\tau(uv) = \tup{x}$ where $\tup{x} = \min_{\leqL} \{ \tup{x'} \in \iter{*}{v} \mid \tup{x'} \lesssim \tup{c} - \dCost(uv) \, \wedge \, \tup{x'} \leqL \tup{c} - \dCost(uv) \} = \min_{\leqL} \cover(uv)$.

    Since $u \in V_2$ and $\tau(u) = \tup{c} \in \iter{*}{u}$, by Corollary~\ref{cor:propOnFixpoint}, there exists $\tup{x''} \in \iter{*}{v}$ such that $\tup{x''} + \dWeight(u,v) \lesssim \tau(u) = \tup{c}$. That implies $\tup{x''} \leqL \tup{c} - \dCost(uv)$ as $\dCost(uv) = \dWeight(u,v)$. In particular, $\cover(uv) \neq \emptyset$ and so $\tau(uv) \in \iter{*}{v}$ and Invariants~\eqref{item:inv1Label} and~\eqref{item:inv4Label} hold. 

    Since $u \in V_2$, Invariant~\eqref{item:inv2Label} has not to be satisfied.

    It remains to prove Invariant~\eqref{item:inv3Label}. Since $\tau(uv) \in \cover(uv)$, $\tau(uv) \leqL \tup{c} - \dCost(uv) = \tau(u) - \dWeight(u,v)$.
    
    \end{itemize}

   \textbf{Induction Hypothesis} Let us assume that Invariant~\eqref{item:inv1Label} to~\eqref{item:inv4Label} hold for all $hv \in \tree^*$ such that $|hv|\leq k$.\\
   
    Let us now prove that for all $hvv' \in \tree^*$ such that $|hvv'| = k + 1$ these invariants are still satisfied.

   \begin{itemize}
       \item If $v \in V_1$: we define $\tau(hvv') = f^*_v(\tau(hv))[2]$. By induction hypothesis $\tau(hv) \in \iter{*}{v}$ thus $f^*_v(\tau(hv))$ is well defined. By definition of $\sigma^*_1$, we have that $v' = \sigma^*_1(hv) = f^*_v(\tup{x})[1]$  with $\tup{x} = \min_{\leqL} \{ \tup{x'} \in \iter{*}{v} \mid \tup{x'} \lesssim \tup{c} - \dCost(hv)\, \wedge \, \tup{x'} \leqL \tup{c} - \dCost(hv) \}$ and $\tup{x} = \tau(hv)$ by induction hypothesis.

       Moreover by Lemma~\ref{lemma:correctStratMin1}, $\tau(hvv') \in \iter{*}{v'}$. That proves Invariant~\eqref{item:inv1Label} and~\eqref{item:inv2Label}.

        Invariant~\eqref{item:inv3Label} is obtain thanks to the fact that $\tau(hvv') = \tau(hv) - \dWeight(v,v')$ (by Lemma~\ref{lemma:correctStratMin1}).

        It remains to prove Invariant~\eqref{item:inv4Label}. Thanks to the induction hypothesis we obtain:

        $$ \tau(hvv') = \tau(hv) - \dWeight(v,v') \leqL \tup{c} - \dCost(hv) - \dWeight(v,v') = \tup{c} - \dCost(hvv')$$

        and 

        $$ \tau(hvv') = \tau(hv) - \dWeight(v,v') \lesssim \tup{c} - \dCost(hv) - \dWeight(v,v') = \tup{c} - \dCost(hvv').$$

        Let us assume now that $\tau(hvv') \neq \min_{\leqL} \{ \tup{x'} \in \iter{*}{v'} \mid  \tup{x'} \lesssim \tup{c} - \dCost(hvv')\, \wedge \, \tup{x'} \leqL \tup{c} - \dCost(hvv') \}$. That means that there exists $\tup{y'} \in \iter{*}{v'}$ such that $\tup{y'} \lesssim \tup{c} - \dCost(hvv')$, $\tup{y'} \leqL \tup{c} - \dCost(hvv')$ and $\tup{y'} \strictLessL \tau(hvv')$.

        Then $\tup{y'} + \dWeight(v,v') \in \uparrow \iter{*}{v}$, thus there exists $\tup{z'} \in \iter{*}{v}$ such that $\tup{z'} \lesssim \tup{y'} + \dWeight(v,v')$. Since  $\tup{z'} \lesssim \tup{y'} + \dWeight(v,v')$ implies $\tup{z'} \leqL \tup{y'} + \dWeight(v,v')$, we have that $\tup{z'} \lesssim \tup{c}  - \dCost(hv)$ and $\tup{z'} \leqL \tup{c}  - \dCost(hv)$. Since $\tup{y'} \strictLessL \tau(hvv') = \tau(hv) - \dWeight(v,v')$, that leads to $\tup{z'} \strictLessL \tau(hv)$ which is a contradiction with the induction hypothesis $\tau(hv) = \min_{\leqL} \{ \tup{x'} \in \iter{*}{v} \mid \tup{x'} \lesssim \tup{c} - \dCost(hv)\, \wedge \, \tup{x'} \leqL \tup{c} - \dCost(hv) \}$.\\

        \item  If $v \in V_2$: we define $\tau(hvv') = \tup{x}$ where $\tup{x} = \min_{\leqL} \{ \tup{x'} \in \iter{*}{v'} \mid \tup{x'} \lesssim \tup{c} - \dCost(hvv') \, \wedge \, \tup{x'} \leqL \tup{c} - \dCost(hvv') \}$. Since $v \in V_2$ and $\tau(hv) \in \iter{*}{v}$, by Corollary~\ref{cor:propOnFixpoint}, there exists $\tup{x''} \in \iter{*}{v'}$ such that $ \tup{x''} + \dWeight(v,v') \lesssim \tau(hv)$. Which implies $ \tup{x''} + \dWeight(v,v') \leqL \tau(hv)$. Moreover, by induction hypothesis, $\tup{x''} \lesssim \tup{c} - \dCost(hv) - \dWeight(v,v') = \tup{c} - \dCost(hvv')$ and $\tup{x''} \leqL \tup{c} - \dCost(hvv')$. Therefore  $ \tau(hvv')$ and $\tup{x''}$ are in the set $\cover(hvv') = \{ \tup{x'} \in \iter{*}{v'} \mid \tup{x'} \lesssim \tup{c} - \dCost(hvv') \, \wedge \, \tup{x'} \leqL \tup{c} - \dCost(hvv') \}$. So, in particular $\tau(hvv') \in \iter{*}{v'}$ and Invariant~\eqref{item:inv1Label} is satisfied. Moreover, as $\tau(hvv')$ is the minimum of the elements of the set $\cover(hvv')$, we have that $\tau(hvv') \leqL \tup{x''} \leqL \tau(hv) - \dWeight(v,v')$. We can conclude that Invariants~\ref{item:inv3Label} and~\ref{item:inv4Label} are satisfied. As $v \in V_2$, the  second invariant has not to be satisfied.        
   \end{itemize}
\end{proof}

Before proving Theorem~\ref{thm:optiStrat} we still need two technical results. 

\begin{lemma}
    \label{lemma:correctStratMin2}
    For all $v \in V_1 \backslash \targetSet$, for all $\tup{x} \in \iter{*}{v} \backslash \{ \infty \}$,
    if $(v',\tup{x'}) = f^*_v(\tup{x})$ then, $\firstOcc{\tup{x'}}{v'} < \firstOcc{\tup{x}}{v}$.
\end{lemma}

\begin{proof}
    We set $n = \firstOcc{\tup{x}}{v}$ and $n' =\firstOcc{\tup{x'}}{v'} $.

    By construction, we have that $\tup{x} = \tup{x'} + \dWeight(v,v')$, $\tup{x} \in \iter{n}{v}$ and $\tup{x'} \in \iter{n-1}{v'}$.

    Thus, $n' \leq n-1$ holds by definition of $\firstOcc{\tup{x'}}{v'}$.
\end{proof}

\begin{lemma}
\label{lemma:correctStratMax}
    For all $v \in V_2 \backslash \targetSet$, for all $\tup{x} \in \iter{*}{v} \backslash \{ \infty \}$, for all $v' \in \successor(v)$ and for all $\tup{x'} \in \iter{*}{v'}$ such that $\tup{x'} + \dWeight(v,v') \leqL \tup{x}$, either, \emph{(i)} $\tup{x'} \strictLessL \tup{x}$ or, \emph{(ii)} $\firstOcc{\tup{x'}}{v'} < \firstOcc{\tup{x}}{v}$.
\end{lemma}

\begin{proof}
    Let $v \in V_2 \backslash \targetSet$, $\tup{x} \in \iter{*}{v} \backslash \{ \infty \}$, $v' \in \successor(v)$ and $\tup{x'} \in \iter{*}{v'}$ such that $\tup{x'} + \dWeight(v,v') \leqL \tup{x}$.

    To obtain a contradiction, we assume that $\neg(\tup{x'} \strictLessL \tup{x})$ and $\firstOcc{\tup{x'}}{v'} \geq \firstOcc{\tup{x}}{v}$. Since $\tup{x'} \leqL \tup{x}$ by hypothesis, $\neg(\tup{x'} \strictLessL \tup{x})$ implies $\tup{x'} = \tup{x}$. Therefore, $\dWeight(v,v') = \tup{0}$.

    Let $n = \firstOcc{\tup{x}}{v}$, by definition $\tup{x} \in \iter{n}{v}$ and by Proposition~\ref{prop:ensureI}, $\tup{x} \in \iterEnsure{n}{v}$. Since $\dWeight(v,v') = 0$, $\tup{x'} =  \tup{x} \in \iterEnsure{n-1}{v'}$.
    In conclusion, the contradiction we were looking for is given by $\firstOcc{\tup{x'}}{v'} \leq n-1 < \firstOcc{\tup{x}}{v}$.
\end{proof}

We are now able to prove Theorem~\ref{thm:optiStrat}.  This proof exploit the notions of tree and strategy tree already defined in Appendix~\ref{appendix:termination}.

\begin{proof}[Proof of Theorem~\ref{thm:optiStrat}]
Let $u \in V$ and $ \tup{c} \in \iter{*}{u} \backslash \{ \tup{\infty} \}$. Let $\sigma^*_1 \in \stratSet{1}{u}$ as defined in Definition~\ref{def:optiStrat}. Let us consider the strategy tree $\stratTree{\sigma^*_1}$.

The first step of the proof is to prove that all branches of $\stratTree{\sigma^*_1}$ are finite and end with a node $n$ such that $\last(n) \in \targetSet$. 

Let us proceed by contradiction and assume that there exists a branch $b = n_0n_1n_2\ldots$ which is infinite. By Statement~\ref{item:inv3Label} of Proposition~\ref{prop:labeling}, we know that the sequence $(\tau(n_k))_{k \in \N}$ is non increasing w.r.t. $\leqL$ and it is lower bounded by $\tup{0}$. It follows that:

\begin{equation} \exists \xi \in \N \text{ such that } \forall \ell \in \N, \, \tau(n_\xi) = \tau(n_{\xi+\ell}).\label{eq:fixpointTau}\end{equation}

Either there exists $\ell \in \N$ such that $\last(n_{\xi + \ell})\in \targetSet$ which contradicts the fact that branch $b$ is infinite. Or, for all $\ell \in \N$, $$\firstOcc{\tau(n_{\xi + \ell +1})}{\last(n_{\xi+\ell+1})} < \firstOcc{\tau(n_{\xi + \ell })}{\last(n_{\xi+\ell})}. $$
Let $hv =  n_{\xi + \ell +1}$ then, $h = n_{\xi + \ell}$. 
\begin{itemize}
    \item If $\last(h) \in V_1$, by Statement~\ref{item:inv2Label} of Proposition~\ref{prop:labeling} we have $(v, \tau(hv)) = f^*_{\last(h)}(\tau(h))$. Moreover, by Lemma~\ref{lemma:correctStratMin2}, we obtain $\firstOcc{\tau(hv)}{v} < \firstOcc{\tau(h)}{\last(h)}$.
    \item If $\last(h) \in V_2$, by Statement~\ref{item:inv3Label} of Proposition~\ref{prop:labeling}, we have that $\tau(hv) + \dWeight(\last(h),v) \allowbreak \leqL \tau(h)$. Additionnaly, by Lemma~\ref{lemma:correctStratMax}, either $\tau(hv) \strictLessL \tau(h)$ which is assumed to be impossible by Eq.~\eqref{eq:fixpointTau} or $ \firstOcc{\tau(hv)}{v} < \firstOcc{\tau(h)}{\last(h)}$.
\end{itemize}

That means that the sequence $\left(\firstOcc{\tau(n_{\xi+\ell})}{\last(n_{\xi+\ell})}\right)_{\ell \in \N}$ is strictly decreasing w.r.t. the classical order $<$ on the natural numbers and is lower bounded by $0$. It follows that such an infinite branch cannot exist. \\

In what precedes, we proved that $\targetSet$ is reached whatever the behavior of $\playerTwo$, in particular each branch $b = n_0n_1 \ldots n_k$ ends in a node $n_k$ which is a leaf, and such that $\last(n_k) \in \targetSet$. Thus, $\tau(n_k) = 0$. Moreover, if $n_k = hv$, the cost of the branch corresponds to $\dCost(hv)$ and by Proposition~\ref{prop:labeling}, we have that $\tau(hv) \lesssim \tup{c} - \dCost(hv)$. That inequality implies that $\dCost(hv) \lesssim \tup{c}$.\end{proof}


In the first part of this section we proved, given $u \in V$ and $c \in \iter{*}{u} \backslash \{ \tup{\infty} \}$, how to obtain a startegy $\sigma^*_1$ of $\playerOne$ that ensures $\tup{c}$ from $u$ (see Definition~\ref{def:optiStrat} and Theorem~\ref{thm:optiStrat}). Thus, in particular, $\sigma^*_1$ is both a lexico-optimal strategy from $u$ and a $\tup{c}$-Pareto-optimal strategy from $u$. However, $\sigma^*_1$ requires finite-memory.

In the remainder of this section, we prove that if we consider the lexicographic order, the strategy $\vartheta^*_1$, given in Proposition~\ref{prop:stratOptiLexicoPositional}, is a positional lexico-optimal strategy from $u$.

\begin{restate}{Proposition}{\ref{prop:stratOptiLexicoPositional}}
\restateStratOptiLexicoPosition
\end{restate}

We now prove that the strategy $\vartheta^*_1$, as defined in Proposition~\ref{prop:stratOptiLexicoPositional}, is a lexico-optimal strategy from $u$. 
We proceed in the same way as we proved that $\sigma^*_1$ ensures $\tup{c}$ from $u$: we prove that  a labeling function of the strategy tree $\stratTree{\vartheta^*_1}$ exists  and has the same kind of  properties as in Proposition~\ref{prop:labeling}. From that follows, for the same arguments as these exploited in the proof of Theorem~\ref{thm:optiStrat}, that $\vartheta^*_1$ is a lexico-optimal strategy from $u$.

\begin{proposition}
If $\lesssim$ is the lexicographic order, for $u \in V$ and $\tup{c} \in \iter{*}{u}\backslash \{ \infty \}$, if $\stratTree{\tau^*_1}$ is the strategy tree of the strategy $\vartheta^*_1$ as defined in Proposition~\ref{prop:stratOptiLexicoPositional} then, there exists a labeling function $\tau: \stratTree{\vartheta^*_1} \longrightarrow \N^\di$ such that, $\tau(u) = \tup{c} \in \iter{*}{u}$ and, for all $hv \in \stratTree{\vartheta^*_1}$ such that $|hv| \geq 1$:
    \begin{enumerate}
        \item $\tau(hv) \in \iter{*}{v}$; \label{item:inv1LabelLexico}
        \item If $\last(h) \in V_1$ then, $(v,\tau(hv)) = f^*_{\last(h)}(\tau(h))$;\label{item:inv2LabelLexico}

        \item $\tau(hv) \leqL \tau(h) - \dWeight(\last(h),v)$; \label{item:inv3LabelLexico}
        \item $\tau(hv) \leqL \tup{c} - \dCost(hv) $. \label{item:inv4LabelLexico}
    \end{enumerate}
\end{proposition}

\begin{proof}
 Let $u \in V$ and $\tup{c} \in \iter{*}{u} \backslash \{ \tup{\infty} \}$. Let $\tree^* = \stratTree{\vartheta^*_1}$. We define $\tau$ and prove Invariant~\eqref{item:inv1LabelLexico} to~\eqref{item:inv4LabelLexico} step by step, by induction on the length of $h \in \tree^*$.

 Before beginning the proof, we recall that,  because we consider the lexicographic order, each time we consider some $\tup{x}, \tup{x'} \in \iter{*}{v}$ for some $v \in V$, we have that $\tup{x} = \tup{x'}$ since $\iter{*}{v}$ is a singleton.

\textbf{Base case} If $h = uv$ for some $v \in V$.

\begin{itemize}
    \item If $u \in V_1$: we define $\tau(uv)= f^*_u(\tau(u))[2]$. By hypothesis, $\tau(u) = \tup{c} \in \iter{*}{u}$ so $f^*_u(\tau(u))$ is well defined. Let us prove that the invariants are satisfied.
    \begin{itemize}
        \item Invariant~\eqref{item:inv1LabelLexico}. By Lemma~\ref{lemma:correctStratMin1}, as $\tau(uv) = f^*_u(\tau(u))[2]$, $\tau(uv) \in \iter{*}{v}$.
        \item Invariant~\eqref{item:inv2LabelLexico}. By construction of $\mathcal{T}^*$, $v = \vartheta^*_1(u)$ and by definition of $\vartheta^*_1$, $\vartheta^*_1(u) = f^*_1(\tup{x})[1]$ where $\tup{x}$ is the only cost profile in $\iter{*}{u}$. Thus, $\tup{x} = \tup{c}$. Finally, since $\tau(u) = \tup{c}$, we obtain that $v = f^*_u(\tau(u))[1]$.
        \item Invariants~\eqref{item:inv3LabelLexico} and~\eqref{item:inv4LabelLexico}. Since $\tau(uv) = f^*_u(\tau(u))[2]$, by Lemma~\ref{lemma:correctStratMin1}, $\tau(uv) = \tau(u) - \dWeight(u,v)$.
        Moreover, because $\tau(u) = \tup{c}$ and $\dWeight(u,v) = \dCost(uv)$, we also have that $\tau(uv) = \tup{c} - \dCost(uv)$.
    \end{itemize}
    \item If $u \in V_2$: we define $\tau(uv) = \tup{x}$ where $\tup{x}$ is the only cost profile in $\iter{*}{v}$. 
    Let us prove that the invariants are satisfied.
    \begin{itemize}
        \item Invariant~\ref{item:inv1LabelLexico}. We have that $\tau(uv) \in \iter{*}{v}$ by construction.
        \item Invariant~\ref{item:inv2LabelLexico}. It does not have to be satisfied since $u \in V_2$.
        \item Invariants~\ref{item:inv3LabelLexico} and~\ref{item:inv4LabelLexico}. We have that $v \in \successor(u)$ and $\tau(u) \in \iter{*}{u}$, thus by Corollary~\ref{cor:propOnFixpoint}, there exists $\tup{x'} \in \iter{*}{v}$ such that $\tup{x'} + \dWeight(u,v) \leqL \tau(u)$.
        But, $\tup{x},\tup{x'} \in \iter{*}{v}$ implies that $\tup{x} = \tup{x'}$ ($\iter{*}{v}$ is a singleton). Moreover $\tau(uv) = \tup{x}$, it follows that $\tau(uv) \leqL \tau(u) - \dWeight(u,v)$. Finally, since $\tau(u) = \tup{c}$ and $\dWeight(u,v) = \dCost(uv)$, we also obtain that $\tau(uv) \leqL \tup{c} - \dCost(uv)$.
    \end{itemize}
\end{itemize}

   \textbf{Induction Hypothesis} Let us assume that Invariant~\eqref{item:inv1LabelLexico} to~\eqref{item:inv4LabelLexico} hold for all $hv \in \tree^*$ such that $|hv|\leq k$.\\
   
    Let us now prove that for all $hvv' \in \tree^*$ such that $|hvv'| = k + 1$ those invariants are still satisfied.

    \begin{itemize}
        \item If $v \in V_1$: we define $\tau(hvv') = f^*_v(\tau(hv))[2]$. By induction hypothesis, $\tau(hv) \in \iter{*}{v}$, so $f^*_v(\tau(hv))$ is well defined. We have that $v' = \vartheta^*_1(hv)$ and by definition of $\vartheta^*_1$, $\vartheta^*_1(hv) = f^*_v(\tup{x})[1]$ with $\tup{x} \in \iter{*}{v}$. Since we consider the lexicographic order, $\iter{*}{v}$ is a singleton and $\tup{x} = \tau(hv)$. Moreover, by Lemma~\ref{lemma:correctStratMin1}, $\tau(hvv') \in \iter{*}{v'}$. It follows that Invariants~\eqref{item:inv1LabelLexico} and~\eqref{item:inv2LabelLexico} are satisfied.

        Invariant~\eqref{item:inv3LabelLexico} is obtained thanks to Lemma~\ref{lemma:correctStratMin1} and the fact that $\tau(hvv') = \tau(hv) - \dWeight(v,v')$.
        It remains to prove that $\tau(hvv') \leqL \tup{c} - \dCost(hvv')$ (Invariant~\ref{item:inv4LabelLexico}). By induction hypothesis, we know that $\tau(hv) \leqL \tup{c} - \dCost(hv)$. Since $\tau(hvv') = \tau(hv) - \dWeight(v,v')$, we have: $\tau(hvv') = \tau(hv)-\dWeight(v,v') \leqL \tup{c} - \dCost(hv) - \dWeight(v,v')$ by induction hypothesis. The fact that $\dCost(hv) - \dWeight(v,v') = \dCost(hvv')$ concludes the proof.\\

        \item if $v \in V_2$: we define $\tau(hvv') = \tup{x}$ where $\tup{x}$ is the only cost profile in $ \iter{*}{v'}$. Notice that in this way, $\tau(hvv') \in \iter{*}{v'}$ (Invariant~\ref{item:inv1LabelLexico}) is already satisfied. 

        Since $v \in V_2$, we do not have to check if Invariant~\eqref{item:inv2LabelLexico} holds.

        As by induction hypothesis $\tau(hv) \in \iter{*}{v}$, we have by Corollary~\ref{cor:propOnFixpoint} that there exists $\tup{x'} \in \iter{*}{v'}$ such that $\tup{x'} + \dWeight(v,v') \leqL \tau(hv)$. But since we consider the lexicographic order, the set $\iter{*}{v'}$ is a singleton, meaning that $\tup{x} = \tup{x'}$. The fact that $\tup{x} = \tau(hvv')$ allows to conclude that Invariant~\ref{item:inv3LabelLexico} is satisfied.

        By what we have just proved $\tau(hvv') \leqL \tau(hv) - \dWeight(v,v')$ and $\tau(hv) - \dWeight(v,v') \leqL \tup{c} - \dCost(hv) - \dWeight(v,v') = \tup{c} - \dCost(hvv')$ by induction hypothesis. It is exactly what Invariant~\eqref{item:inv4LabelLexico} states.     
    \end{itemize}
\end{proof}


\section{Additional content of Section~\ref{section:constrainedExistence}: Constrained Existence}
\label{app:constrainedExistence}


\begin{proposition}
\label{prop-PSPACEc}
If $\lesssim$ is the componentwise order, the constrained existence problem is $\PSPACE$-complete, even if $\di = 2$ and $|\targetSet| = 1$.
\end{proposition}

\subsubsection*{$\PSPACE$-easiness of the constrained existence problem} Proposition~\ref{prop:noCycle}
allows us to prove that the 
 constrained existence 
problem is in $\APTime$. The alternating Turing machine works as follows: all vertices of the game owned by $\playerOne$ (resp. $\playerTwo$) correspond to disjunctive states (resp. conjunctive states). A path of length $|V|$ is accepted if and only if, \emph{(i)} the target set is reached along that path and \emph{(ii)} the sum of the weights until a element of the target set is $\leqC \tup{x}$. If such a path exists, there exists a strategy of $\playerOne$ that ensures the cost profile $\tup{x}$.  This procedure is done in polynomial time and since $\APTime = \PSPACE$, we get the result.\\

The hardness of Proposition~\ref{prop-PSPACEc} is obtained thanks to a polynomial reduction from the \textsc{Quantified Subset-Sum} problem which is proved $\PSPACE$-complete~\cite[Lemma 4]{Travers06}. Although some intuition on the $\PSPACE$-hardness is provided in Section~\ref{section:constrainedExistence}, we provide hereunder a formal proof of this result.

\subsubsection*{$\PSPACE$-hardness of the  constrained existence problem}

\begin{proof}[Formal proof of the $\PSPACE$-hardness]

For all odd (resp. even) numbers $k$, $1\leq k \leq n$, we denote by $\valExists{k}: \{0,1\}^{k-1} \longrightarrow \{0,1\}$ (resp. $\valForall{k}:  \{0,1\}^{k-1} \longrightarrow \{0,1\}$) the valuation of the variable $x_k$ taking into account the valuation of previous variables $x_1, \ldots, x_{k-1}$. We assume that $\valExists{1}: \emptyset \longrightarrow \{0,1\}$. Given a sequence $\valExists{} = \valExists{1}, \valExists{3},  \ldots$ and a sequence $\valForall{} = \valForall{2}, \valForall{4}, \ldots$, we define the function $\nu_{(\valExists{},\valForall{})}: \{x_1,\ldots,x_n\} \longrightarrow \{0,1\}$ such that  $\nu_{(\valExists{},\valForall{})}(x_1) = \valExists{1}(\emptyset)$, $\nu_{(\valExists{},\valForall{})}(x_2) = \valForall{2}(\nu_{(\valExists{},\valForall{})}(x_1))$, $\nu_{(\valExists{},\valForall{})}(x_3) = \valExists{3}(\nu_{(\valExists{},\valForall{})}(x_1)\nu_{(\valExists{},\valForall{})}(x_2))$, $\ldots$ We also define the set $S(\valExists{},\valForall{}) = \{ p \mid \nu_{(\valExists{},\valForall{})}(x_p) = 1 \}$.\\

Thanks to these notations we rephrase the \textsc{Quantified Subset-Sum} problem as: does there exist a sequence of functions $\valExists{} = \valExists{1},\valExists{3},\ldots$  such that for all sequences $\valForall{} = \valForall{2}, \valForall{4}, \ldots$, $$\sum_{p \in S(\valExists{},\valForall{})} a_p = T?$$\\

 We now describe the reduction from the \textsc{Quantified Subset-Sum} problem to the constrained existence problem.

The $\arena_2 = (V_1,V_2,E, \dWeight)$ of the initialized $2$-weighted reachability game $(\game_2,v_0) = (\arena_2, \targetSet, \dCost)$ is given in Figure~\ref{fig:hardnessGame}.   Formally, the game is built as follows:
\begin{itemize}
    \item $V_1$ is composed by the following vertices: a vertex $y$, for each variable $x_p$ under an existential quantifier there is a vertex $x_p$ and finally for all $a_p \in I$ there are two vertices $x^0_p$ and $x^1_p$, ;
    \item $V_2$ is the set of vertices denoted by $x_p$ such that the variable $x_p$ is under an universal quantifier. Notice that in Figure~\ref{fig:hardnessGame} we assume that $n$ is odd, and so $x_n$ is under an existential quantifier.
    
    \item $E$ is composed of the edges of the form: \begin{itemize} \item  $(x_\ell, x^0_\ell)$ and $(x_\ell,x^1_\ell)$, for all $1\leq \ell \leq n$; \item $(x^1_\ell,x_{\ell+1})$ and $(x^0_\ell,x_{\ell+1})$, for all $1\leq \ell \leq n-1$;
    \item $(x^1_n,y)$, $(x^0_n,y)$ and $(y,y)$.
    \end{itemize}
    \item the weight function $\dWeight$ is defined as: \begin{itemize} \item $\dWeight(x_\ell,x^1_\ell) = (a_\ell,0)$ and $\dWeight(x_\ell,x^0_\ell) = (0, a_\ell)$, for all $1 \leq \ell \leq \di$; \item for all the other edges $e \in E$, $\dWeight(e) = (0,0)$.
    \end{itemize}
    \item $F = \{ y \}$;
    \item $v_0 = x_1$.
\end{itemize}

    We  prove the following equivalence.
    
    There exists a sequence $\valExists{} = \valExists{1}, \valExists{3}, \ldots $ such that for all sequences $\valForall{} = \valForall{2}, \valForall{4}, \ldots$, \\$\sum_{p \in S(\valExists{},\valForall{})} a_p = T$
    if and only if
    there exists a finite-memory strategy $\sigma_1$ of $\playerOne$ from $x_1$ such that for all strategies $\sigma_2$ of $\playerTwo$ from $x_1$, \\
    $\dCost(\outcome{\sigma_1}{\sigma_2}{x_1}) \leqC (T, \sum_{1 \leq p \leq n} a_p - T)$.\\

\textbf{Remark} Before proving this equivalence, let us notice that:
\begin{enumerate}
    \item \label{item-proofPSPACE1}For each sequence $\valExists{} = \valExists{1}, \valExists{3}, \ldots$ (resp. each sequence $\valForall{}= \valForall{2}, \valForall{4}, \ldots$),  there exists a corresponding finite-memory strategy $\sigma_1$ of $\playerOne$ (resp. $\sigma_2$ of $\playerTwo$) in $(\game_2,x_1)$;
    \item\label{item-proofPSPACE2} For each finite-memory strategy $\sigma_1$ of $\playerOne$ (resp. each strategy $\sigma_2$ of $\playerTwo$), there exists a corresponding sequence $\valExists{}= \valExists{1}, \valExists{3}, \ldots$ (resp. $\valForall{} = \valForall{2}, \valForall{4},\ldots$). 
    
\end{enumerate}

Construction of strategies of Statement~\ref{item-proofPSPACE1}. Let $\valExists{} = \valExists{1}, \valExists{3}, \ldots$ and $\valForall{} = \valForall{2}, \valForall{4}, \ldots$. We define $\sigma_1$: for all $1 \leq \ell \leq n$ such that $\ell$ is odd, $\sigma_1(x_1) = x^i_1$ if $\valExists{1}(\emptyset) = i$ and $\sigma_1(x_1v_1x_2v_2\ldots x_\ell) = x^i_\ell$ if $\valExists{\ell}(\overline{v}_1\overline{v}_2\ldots \overline{v}_{\ell-1}) = i$ with, for all $ 1 \leq p \leq \ell$, $v_p \in \{ x^1_p, x^0_p \}$ and $\overline{v}_p = 1$ if $v_p = x^1_p$ and $\overline{v}_p = 0$ otherwise.
The strategy $\sigma_2$ is defined exactly in the same way for all $1\leq \ell \leq n$ such that $\ell$ is even, except that $\valExists{\ell}$ is replaced by $\valForall{\ell}$. 

Construction of strategies of Statement~\ref{item-proofPSPACE2}. Let $\sigma_1$ be a finite-memory strategy of $\playerOne$ and $\sigma_2$ be a strategy of $\playerTwo$. We build $\valExists{}$ as follows:
$\valExists{1}(\emptyset) = i$ if $\sigma_1(x_1)= x^i_1$ and, for all $1 \leq \ell \leq n$ such that $\ell$ is odd, $\valExists{\ell}(\overline{v}_1\ldots\overline{v}_{\ell-1}) = i$ if $ \sigma_1(x_1v_1x_2v_2\ldots x_{\ell}) = x^i_\ell$ with for all $1 \leq p \leq \ell-1$, $\overline{v}_p \in \{0,1\}$ and $v_p = x^1_p$ if $\overline{v}_p = 1$ and $v_p = x^0_p$ otherwise.
The $\valForall{\ell}$ are defined exactly in the same way for all $1 \leq \ell \leq n$ such that $\ell$ is even and by replacing $\sigma_1$ by $\sigma_2$.\\

We come back to the proof of the equivalence.

\noindent$(\Rightarrow)$  We assume that there exists a sequence $\valExists{} = \valExists{1}, \valExists{3}, \ldots$ such that for all sequences $\valForall{} = \valForall{2}, \valForall{4}, \ldots $, $\sum_{p \in S(\valExists{},\valForall{})} a_p = T$. 

We consider $\sigma_1$ as defined previously (Remark, Statement~\ref{item-proofPSPACE1}). We have to prove that for all strategies $\sigma_2$ of $\playerTwo$ : $\dCost(\outcome{\sigma_1}{\sigma_2}{x_1}) \leqC (T, \sum_{1\leq p \leq n} a_p - T)$.

Let $\sigma_2$ be a strategy of $\playerTwo$. As explained in Remark, Statement~\ref{item-proofPSPACE2}, we have that $\sigma_2$ corresponds to some sequence $\valForall{} = \valForall{2}, \valForall{4}, \ldots$.
Thus, by construction of the game arena and by hypothesis we have:

\begin{itemize}
\item $\cost_1(\outcome{\sigma_1}{\sigma_2}{x_1}) = \sum_{p \in S(\valExists{},\valForall{}) } a_p = T$ and 
\item $\cost_2(\outcome{\sigma_1}{\sigma_2}{x_1}) = \sum_{p \not \in S(\valExists{},\valForall{})} a_p =  \sum_{1\leq p \leq n} a_p - \sum_{p \in S(\valExists{},\valForall{})} a_p = \sum_{1\leq p \leq n} a_p - T$.
\end{itemize}

\noindent$(\Leftarrow)$ Let us assume that there exists a finite-memory strategy $\sigma_1$ of $\playerOne$ such that for all strategies $\sigma_2$ of $\playerTwo$, $\dCost(\outcome{\sigma_1}{\sigma_2}{x_1}) \leqC (T, \sum_{1 \leq p \leq n} a_p - T)$.

We define the sequence $\valExists{} = \valExists{1}, \valExists{3}, \ldots$ as explained in Remark, Statement~\ref{item-proofPSPACE2}. Let $\valForall{} = \valForall{2}, \valForall{4}, \ldots$, we have to prove that $ \sum_{p \in S(\valExists{},\valForall{})} a_p = T$. 

By Remark, Statement~\ref{item-proofPSPACE1}, the sequence $\valForall{}$ corresponds to a strategy $\sigma_2$ of $\playerTwo$ in $(\game_2,x_1)$. It follows by hypothesis and construction of the game arena:
\begin{itemize}
    \item $
 \cost_1(\outcome{\sigma_1}{\sigma_2}{x_1})=  \sum_{p \in S(\valExists{},\valForall{})} a_p$ \item $\cost_1(\outcome{\sigma_1}{\sigma_2}{x_1}) \leq T$
 \end{itemize}

 and 

\begin{itemize}
\item $\cost_2(\outcome{\sigma_1}{\sigma_2}{x_1}) = \sum_{p \not \in S(\valExists{},\valForall{})} a_p = \sum_{1 \leq p \leq n} a_p - \sum_{p \in S(\valExists{},\valForall{})} a_p$ \item $\cost_2(\outcome{\sigma_1}{\sigma_2}{x_1}) \leq \sum_{1 \leq p \leq n} a_p -T.$
\end{itemize}   

Thus, we can conclude that $\sum_{p \in S(\valExists{},\valForall{})} a_p = T$.
\end{proof}


\section{Additional Contents of Section~\ref{section:permissiveMultiStrat}: \nameref{section:permissiveMultiStrat}}
\label{appendix:permissiveMultiStrat}

In this section we provide the formal definition of the  associated extended game of a quantitative reachability game as introduced in Section~\ref{section:fromMultiToPermi}.

\begin{definition}[Extended game] Let $\game = (\arena, \targetSet, \cost)$ be a quantitative reachability game such that $\arena = (V_1,V_2,E,\edgeCost)$, its associated extended game $\extendedGame_\kappa = (\arena^X_\kappa, \targetSet^X, \dCost^X)$, where $\arena^X_\kappa = (V^X_1,V^X_2,E^X, \dWeight^X)$, is a multi-weighted reachability game defined as follows:

\begin{itemize}
    \item for each $v \in V_1$ there exists a corresponding vertex $v^X \in V^X_1$ and there is no other kind of vertex in $V^X_1$;
    \item the set  of vertices $V^X_2$ comprises two types of vertices:  \emph{(i)} for each vertex $v \in V_2$, there exists a corresponding vertex $v^X$ in $V^X_2$ and \emph{(ii)} for each vertex $v \in V_1$ and each set $A \subseteq \successor(v) \setminus \{ \emptyset \}$, there exists an associated vertex $v^X_A$ in $V^X_2$; 
    \item the set of edges $E^X$ is made up of the following edges: 
    \begin{itemize}
        \item for each $v \in V_2$ and $v' \in V_1 \cup V_2$, $(v,v') \in E$ if and only if $(v^X,v'^X) ±\in E^X$ and $\dWeight^X(v^X,v'^X) = \begin{cases} (c,0) & \text{ if } \kappa = \CP \\ (0, c) & \text{ if } \kappa = \PC \end{cases}$, where $c = \edgeCost(v,v')$ ; 
        \item
    for each $v \in V_1$ and $A \subseteq \successor(v) \setminus \{ \emptyset \} $, let $v^X$ be the corresponding vertex of $v$ in $V^X_1$ and $v^X_A $ be the associated vertex of $v$ and $A$ in $V^X_2$:
    \begin{itemize} \item  $(v,v^X_A) \in E^X$ and  $\dWeight^X(v^X,v^X_A) = \begin{cases} (0,p) & \text{ if } \kappa = \CP \\ (p,0) & \text{ if } \kappa = \PC \end{cases}$, \\ where $p =  \sum_{u \in \successor(v) \setminus A } \edgePenal(v,u)$;
    \item  moreover for all $u \in A$, if $u^X$ is the corresponding vertex of $u$ in $V^X$, then $(v^X_A,u^X) \in E^X$ and $\dWeight^X(v^X_A,u^X) = \begin{cases} (c,0) & \text{ if } \kappa = \CP \\ (0,c) & \text{ if } \kappa = \PC \end{cases}$, where $c = \edgeCost(v,u)$;
    \end{itemize}
    \end{itemize}
    \item for each $v \in V$, let $v^X$ be its corresponding vertex in $V^X$, we have that $v \in \targetSet$ if and only if $v^X \in \targetSet^X$.
\end{itemize}
\end{definition}

\end{document}